\documentclass[acmsmall]{acmart}

\setcopyright{cc}
\setcctype{by}
\acmJournal{PACMMOD}
\acmYear{2025} \acmVolume{3} \acmNumber{3 (SIGMOD)}
\acmArticle{150} \acmMonth{6} \acmPrice{}\acmDOI{10.1145/3725287}

\usepackage{listings}
\usepackage{float} 
\usepackage{newfloat}
\DeclareFloatingEnvironment[
    fileext=los,
    listname={List of Listings},
    name=Listing,
    placement=tb,
    within=section, 
]{lstfloat}
\usepackage{balance}  
\usepackage{graphicx}
\usepackage{subcaption}
\usepackage{array}
\newcolumntype{M}[1]{>{\arraybackslash}m{#1}}
\usepackage{comment}
\usepackage{amsmath,amsthm}
\usepackage{url}
\usepackage{epsfig}
\usepackage{savesym}
\usepackage{verbatim}
\usepackage{multirow}
\usepackage{multicol}
\usepackage{epstopdf}
\usepackage{caption}
\usepackage{subcaption}
\usepackage{lipsum}
\usepackage{setspace}
\usepackage{etoolbox}
\usepackage{tikz}
\usepackage{balance}
\usepackage{diagbox}
\usetikzlibrary{automata, arrows, calc}
\usepackage{tkz-graph}

\usepackage{enumitem}
\usepackage{mathrsfs}
\usepackage[show]{chato-notes}
\usepackage{xcolor}
\usepackage{pifont}
\usepackage{tcolorbox}
\usepackage{tablefootnote}
\usepackage{svg}

\definecolor{monokaiBackground}{HTML}{272822}
\definecolor{monokaiGray}{HTML}{75715E}
\definecolor{monokaiGreen}{HTML}{A6E22E}
\definecolor{monokaiYellow}{HTML}{E6DB74}
\definecolor{monokaiRed}{HTML}{F92672}

\definecolor{LightGray}{gray}{0.9}

\usepackage{titlesec}

\titleformat{\section}[block]{\normalfont\Large\bfseries}{\thesection}{1em}{}
\titlespacing*{\section}{0pt}{0.6\baselineskip}{0.3\baselineskip}

\titleformat{\subsection}[block]{\normalfont\large\bfseries}{\thesubsection}{1em}{}
\titlespacing*{\subsection}{0pt}{0.3\baselineskip}{0\baselineskip}

\usepackage[linesnumbered,algoruled,boxed,lined,noend,ruled,vlined]{algorithm2e}
\usepackage{breqn}
\usepackage{hyperref}
 
\SetVlineSkip{0pt}

\SetCommentSty{mycommfont}

\newtheorem{definition}{Definition}[section]

\newtheorem{theorem}{Theorem}[section]
\newtheorem{lemma}[theorem]{Lemma}
\newtheorem{property}[theorem]{Property}

\newtheorem{example}{Example}[section]
\newcommand{\ycc}[1]{{#1}}

\begin{document}

\newcommand{\cal}[1]{\mathcal{#1}}


\definecolor{apipink}{rgb}{0.858, 0.188, 0.478}

\lstdefinestyle{apistyle}{
    belowcaptionskip=1\baselineskip,
    breaklines=true,
    frame=none,
    numbers=left, 
    xrightmargin=1em,
    framexrightmargin=1em,
    basicstyle=\footnotesize\ttfamily,
    keywordstyle=\bfseries\color{apipink},
    commentstyle=\itshape\color{green!40!black},
    identifierstyle=\color{black},
    backgroundcolor=\color{gray!10!white},
    linewidth=.97\columnwidth,
    numbersep=3pt,
}
\newcommand{\jiaxin}[1]{\textcolor{blue}{#1}}
\newcommand{\siyuan}[1]{\textcolor{orange}{#1}}
\newcommand{\keyres}[1]{#1}
\newcommand{\tr}[1]{}
\newcommand{\TODO}[1]{\textcolor{BurntOrange}{\textbf{TODO:} #1}}
\newcommand{\Remind}[1]{\textcolor{RubineRed}{\textbf{Reminder:}#1}}
\newcommand{\byronsuggestion}[1]{\textcolor{Green}{\textbf{Reminder:}#1}}
\newcommand{\red}[1]{\textcolor{red}{#1}}
\newcommand{\what}[1]{\textcolor{blue}{?#1?}}
\newcommand{\eat}[1]{}
\newcommand{\kw}[1]{{\ensuremath {\textsf{#1}}}\xspace}
\newenvironment{tbi}{\begin{itemize}
		\setlength{\topsep}{0.6ex}\setlength{\itemsep}{0ex}} 
	{\end{itemize}} 
\newcommand{\ei}{\end{itemize}}

\newcommand{\Gontology}{G_{Ont}}
\newcommand{\PRADS}{\textsf{\small PADS}}
\newcommand{\KPADS}{\textsf{\small KPADS}}
\newcommand{\BPADS}{\textsf{\small BPADS}}
\newcommand{\ADS}{\textsf{\small ADS}}
\newcommand{\PKD}{\textsf{\small PKD}}
\newcommand{\FRAMEWORK}{\kw{FRAMEWORK\_NAME}}
\newcommand{\ALGO}{\kw{ALGO\_NAME}}
\newcommand{\PEval}{\kw{PEval}}
\newcommand{\IncEval}{\kw{IncEval}}
\newcommand{\Assemble}{\kw{Assemble}}
\newcommand{\ARef}{\textsf{\small ARefine}}
\newcommand{\ACmpl}{\textsf{\small AComplete}}
\newcommand{\DataGraph}{$G$}
\newcommand{\Ghier}{\mathbb{G}}
\newcommand{\qhier}{\mathbb{Q}}
\newcommand{\subclassOf}{\mathsf{SubClassOf}}
\newcommand{\subtypeOf}{\mathsf{SubTypeOf}}
\newcommand{\answer}{\mathsf{ans}}
\newcommand{\answerb}{\mathsf{ans}^b}
\newcommand{\answerf}{\mathsf{ans}^f}
\newcommand{\Bisim}{\mathsf{Bisim}}
\newcommand{\rank}{\mathsf{rank}}
\newcommand{\dist}{\mathsf{dist}}
\newcommand{\Next}{\mathsf{next}}
\newcommand{\distance}{\mathsf{d}}
\newcommand{\reach}{\mathsf{reach}}
\newcommand{\Generalization}{\mathsf{Gen}}
\newcommand{\Specialization}{\mathsf{Spec}}
\newcommand{\desc}{\mathsf{desc}}
\newcommand{\support}{\mathsf{sup}}
\newcommand{\distort}{\mathsf{DT}}
\newcommand{\degree}{\mathsf{deg}}
\renewcommand{\equiv}{\mathsf{equiv}}
\newcommand{\equivv}[1]{[#1]_\mathsf{equiv}}
\newcommand{\irchi}[2]{\raisebox{\depth}{$#1\chi$}}
\newcommand{\Summarization}{\mathpalette\irchi\relax}
\newcommand{\Configuration}{C}
\newcommand{\radius}{d_{max}}
\newcommand{\Eval}{\mathsf{eval}}
\newcommand{\peval}{\mathsf{peval}}
\newcommand{\F}{{\cal F}}
\newcommand{\G}{{\cal G}}
\newcommand{\Score}{\mathsf{scr}}
\newcommand{\ie}{\emph{i.e.,}\xspace}
\newcommand{\eg}{\emph{e.g.,}\xspace}
\newcommand{\wrt}{\emph{w.r.t.}\xspace}
\newcommand{\aka}{\emph{a.k.a.}\xspace}

\newcommand{\equi}{\mathsf{equi}}
\newcommand{\sn}{\mathsf{sn}}
\newcommand{\METIS}{\mathsf{\small METIS}}

\newcommand{\maxsf}{\mathsf{max}}
\newcommand{\PIE}{\mathsf{PIE}}
\newcommand{\PINE}{\mathsf{PINE}}

\newcommand{\threshold}{\tau}

\newcommand{\cost}{\mathsf{cost}}
\newcommand{\costq}{\mathsf{cost}_\mathsf{q}}
\newcommand{\CR}{\mathsf{compress}}
\newcommand{\DT}{\mathsf{distort}}
\newcommand{\FP}{\mathsf{fp}}
\newcommand{\maxSAT}{\mathsf{maxSAT}}
\newcommand{\True}{\mathsf{T}}
\newcommand{\False}{\mathsf{F}}
\newcommand{\OptGen}{\mathsf{OptGen}}
\newcommand{\freq}{\mathsf{freq}}

\newcommand{\match}{\mathsf{match}}

\newcommand{\vpd}{V_{pd}}
\newcommand{\vtp}{V_{tbp}}
\newcommand{\BiGindex}{\mathsf{BiG\textnormal{-}index}}
\newcommand{\filter}{\mathsf{filter}}
\newcommand{\ans}{\mathsf{ans\_graph\_gen}}
\newcommand{\Azero}{\mathbb{A}}
\newcommand{\boost}{\mathsf{boost}}
\newcommand{\bkws}{\mathsf{bkws}}
\newcommand{\fkws}{\mathsf{fkws}}
\newcommand{\rkws}{\mathsf{rkws}}
\newcommand{\dkws}{\mathsf{dkws}}
\newcommand{\knk}{\mathsf{knk}}
\newcommand{\ksp}{\mathsf{ksp}}
\newcommand{\config}{\mathsf{config}}
\newcommand{\content}{\mathsf{isKey}}
\newcommand{\pcnt}{\mathsf{pcnt}}
\newcommand{\private}{\textsf{isPrivate}}

\newcommand{\Path}{{\mathcal{P}}}
\renewcommand{\P}{{\mathcal P}}

\newcommand{\ppkws}{\textsf{\small PPKWS}\xspace}

\newcommand{\DKWS}{\textsf{\small DKWS}\xspace}
\newcommand{\kDKWS}{\textsf{\small $k$DKWS}\xspace}
\newcommand{\SKWS}{\textsf{\small BFKWS}\xspace}
\newcommand{\KWS}{\textsf{\small KWS}\xspace}
\newcommand{\VU}{\mathbb{V}}
\newcommand{\VI}{\mathcal{V}}
\newcommand{\VM}{\bar{\mathcal{V}}}
\newcommand{\vsf}{\mathsf{v}}
\newcommand{\usf}{\mathsf{u}}
\newcommand{\answerset}{\mathcal{A}}
\newcommand{\candanswerset}{\bar{\mathcal{A}}}
\newcommand{\prune}{S}
\newcommand{\Ud}{\hat{\dist}}
\newcommand{\Ld}{\check{\dist}}
\newcommand{\invert}{\mathscr{I}}
\newcommand{\MB}{u.b}
\newcommand{\MF}{\mathsf{f}}
\newcommand{\Queue}{\mathcal{P}}
\newcommand{\Visit}{\mathsf{Vis}}
\newcommand{\invertV}{V_{\invert}}
\newcommand{\invertE}{E_{\invert}}
\newcommand{\tnormal}[1]{\textnormal{#1}}

\newcommand{\DKWSBF}{\textsf{\small BF}\xspace}
\newcommand{\DKWSNP}{\textsf{\small BF+PADS+NP}\xspace}
\newcommand{\DKWSPADS}{\textsf{\small BF+PADS}\xspace}
\newcommand{\DKWSPINE}{\textsf{\small BF+ALL}\xspace}
\newcommand{\Notify}{\mathsf{Notify}}
\newcommand{\Push}{\mathsf{Push}}
\newcommand{\Parameters}{\mathscr{X}}
\newcommand{\grape}{\kw{GRAPE}}
\newcommand{\Buffer}{\mathbb{B}}
\newcommand{\SI}{\kw{SI}}

\newcommand{\SBGindex}{\mathsf{SBGIndex}}

\newcommand{\SGIndex}{\mathsf{SGIndex}} 

\newcommand{\Portal}{\mathbb{P}}

\newcommand{\stitle}[1]{\vspace{0.4ex}\noindent{\bf #1}}
\newcommand{\expstitle}[1]{\vspace{0.4ex}\noindent{\underline{\bf #1}}}
\newcommand{\etitle}[1]{\vspace{0.8ex}\noindent{\underline{\em #1}}}
\newcommand{\eetitle}[1]{\vspace{0.6ex}\noindent{{\em #1}}}

%
\newcommand{\techreport}[2]{#2}
\newcommand{\SGFrame}{\mathsf{SGFrame}} 

\newcommand{\stab}{\rule{0pt}{8pt}\\[-2.0ex]}
\newcommand{\tab}{\hspace{4ex}}

\newcommand{\Q}{{\cal Q}}



\newcommand{\eop}{\hspace*{\fill}\mbox{\qed}}

\setlength{\floatsep}{0.2\baselineskip plus 0.1\baselineskip minus 0.1\baselineskip}
\setlength{\textfloatsep}{0.2\baselineskip plus 0.1\baselineskip minus 0.1\baselineskip}
\setlength{\intextsep}{0.4\baselineskip plus 0.1\baselineskip minus 0.1\baselineskip}
\setlength{\dbltextfloatsep}{0.2\baselineskip plus 0.1\baselineskip minus 0.1\baselineskip}
\setlength{\dblfloatsep}{0.2\baselineskip plus 0.1\baselineskip minus 0.1\baselineskip}

\renewcommand{\textfraction}{0.05}
\renewcommand{\topfraction}{0.95}
\renewcommand{\bottomfraction}{0.2}
\renewcommand{\floatpagefraction}{0.99}

\newcommand{\Src}{S}
\newcommand{\Dst}{T}
\newcommand{\SD}{At least $k$ $\mathsf{S}$-$\mathsf{T}$ maximum-flow}
\newcommand{\SDMF}{$\mathsf{kSTMF}$}
\newcommand{\SDMFG}{$\mathsf{kSTMF}^g$}
\newcommand{\TEMSDMF}{$\mathsf{TEM}$-$\mathsf{kSTMF}$}
\newcommand{\STcore}{$\mathsf{ST}$-FCore}
\newcommand{\Core}{\mathsf{Core}}
\newcommand{\MFlow}{\mathsf{MaxFlow}}
\newcommand{\MFavg}{MF}
\newcommand{\DKS}{\mathsf{DkS}}
\newcommand{\CBB}{C^T}
\newcommand{\fT}{f^T}
\newcommand{\tsf}{\mathsf{t}}
\newcommand{\csf}{\mathsf{c}}
\newcommand{\fsf}{\mathsf{f}}
\newcommand{\Tsf}{\mathsf{T}_{max}}
\newcommand{\algo}{\mathsf{algo}}
\newcommand{\FnDense}{\mathsf{FnDense}}
\newcommand{\FnSparse}{\mathsf{FnSparse}}
\newcommand{\Transform}{\hat{G}}
\newcommand{\TFNet}{$\mathsf{TF}$-$\mathsf{Network}$}
\newcommand{\RTFNet}{$\mathsf{RTF}$-$\mathsf{Network}$}
\newcommand{\Baseline}{\mathsf{Baseline}}
\newcommand{\Greedy}{\mathsf{Greedy}}
\newcommand{\BWCC}{$\mathsf{WCC}$}
\newcommand{\Tran}{\mathsf{T}}
\newcommand{\ts}{\tau}
\newcommand{\name}{\red{\mathsf{STMflow}}}
\newcommand{\FF}{$\mathsf{Flow}$-$\mathsf{Force}$}


\newcommand{\Spade}{$\mathsf{Spade}$}
\newcommand{\TC}{$\mathsf{TC}$}
\newcommand{\TDS}{$\mathsf{TDS}$}
\newcommand{\KCDS}{$k\mathsf{CLiDS}$}
\newcommand{\KCore}{$k$-$\mathsf{core}$}
\newcommand{\KTip}{$k$-$\mathsf{tip}$}
\newcommand{\KClique}{$k$-$\mathsf{clique}$}
\newcommand{\ALENEX}{$\mathsf{ALENEX}$}
\newcommand{\PKMC}{$\mathsf{PKMC}$}
\newcommand{\FWA}{$\mathsf{FWA}$}

\newtheorem{manualtheoreminner}{Theorem}
\newenvironment{manualtheorem}[1]{%
  \renewcommand\themanualtheoreminner{#1}%
  \manualtheoreminner
}{\endmanualtheoreminner}

\newtheorem{manuallemmainner}{Lemma}
\newenvironment{manuallemma}[1]{%
  \renewcommand\themanuallemmainner{#1}%
  \manuallemmainner
}{\endmanualtheoreminner}

\newcommand{\Seq}{O}
\newcommand{\Grab}{\mathsf{Grab}}
\newcommand{\DENG}{\mathsf{DG}}
\newcommand{\DENGW}{\mathsf{DW}}
\newcommand{\Fraudar}{\mathsf{FD}}
\newcommand{\IncDENG}{$\mathsf{IncDG}$}
\newcommand{\IncDENGW}{$\mathsf{IncDW}$}
\newcommand{\IncFraudar}{$\mathsf{IncFD}$}
\newcommand{\IncDENGU}{$\mathsf{IncDGU}$}
\newcommand{\IncDENGWU}{$\mathsf{IncDWU}$}
\newcommand{\IncFraudarU}{$\mathsf{IncFDU}$}
\newcommand{\permutation}{\alpha}
\newcommand{\AFF}{$\mathsf{AFF}$}
\newcommand{\TCal}{$\mathcal{T}$}
\newcommand{\latency}{\mathcal{L}}
\newcommand{\Elapsed}{\mathcal{E}}
\newcommand{\Ratio}{\mathcal{R}}

\newcommand{\DG}{$\mathsf{DG}$}
\newcommand{\DW}{$\mathsf{DW}$}
\newcommand{\FD}{$\mathsf{FD}$}
\newcommand{\GBBS}{$\mathsf{GBBS}$}
\newcommand{\PBBS}{$\mathsf{PBBS}$}
\newcommand{\PEELB}{$\mathsf{Dupin}$}
\newcommand{\PEELLT}{$\mathsf{DupinGPO}$}
\newcommand{\PEELLTP}{$\mathsf{DupinLPO}$}
\newcommand{\kClist}{$\mathsf{kCLIST}$}
\newcommand{\PEELK}{$\mathsf{PEELB}$-$k\mathsf{CLi}$}
\newcommand{\NS}{-}
\newcommand{\TLE}{$\mathsf{TLE}$}


\newcommand{\Dupin}{$\mathsf{Dupin}$}
\newcommand{\remove}{U}
\newcommand{\DDS}{$\mathsf{DSD}$}

\newcommand*\circled[1]{\tikz[baseline=(char.base)]{
            \node[shape=circle,draw,inner sep=2pt] (char) {#1};}}



\newcommand{\rev}[2][1]{\noindent\emph{\underline{{#1}}: {#2\vspace{1mm}}}\\}
\newcommand{\res}[1]{\noindent\textbf{RESPONSE: }{#1}}

\captionsetup[table]{aboveskip=2pt, belowskip=1pt}
\captionsetup[figure]{aboveskip=2pt, belowskip=1pt}

\title{Dupin: A Parallel Framework for Densest Subgraph Discovery in Fraud Detection on Massive Graphs (Technical Report)}

\settopmatter{authorsperrow=3}

\author{Jiaxin Jiang}
\affiliation{%
    \institution{National University of Singapore}
    \country{Singapore}
}
\email{jxjiang@nus.edu.sg}
\orcid{0000-0001-8748-3225}

\author{Siyuan Yao}
\affiliation{%
    \institution{National University of Singapore}
    \country{Singapore}
}
\email{siyuan.y@u.nus.edu}
\orcid{0009-0001-8243-3947}

\author{Yuchen Li}
\affiliation{%
    \institution{Singapore Management University}
    \country{Singapore}
}
\email{yuchenli@smu.edu.sg}
\orcid{0000-0001-9646-291X}

\author{Qiange Wang}
\affiliation{
  \institution{National University of Singapore}
  \country{Singapore}
}
\email{wang.qg@nus.edu.sg}

\author{Bingsheng He}
\affiliation{%
    \institution{National University of Singapore}
    \country{Singapore}
}
\email{hebs@comp.nus.edu.sg}
\orcid{0000-0001-8618-4581}


\author{Min Chen}
\affiliation{%
  \institution{GrabTaxi Holdings}
  \country{Singapore}
}
\email{min.chen@grab.com}

\renewcommand{\shortauthors}{Jiaxin Jiang et al.}


\received{October 2024}
\received[revised]{January 2025}
\received[accepted]{February 2025}





\begin{abstract}
Detecting fraudulent activities in financial and e-commerce transaction networks is crucial. One effective method for this is Densest Subgraph Discovery (\DDS{}). However, deploying \DDS{} methods in production systems faces substantial scalability challenges due to the predominantly sequential nature of existing methods, which impedes their ability to handle large-scale transaction networks and results in significant detection delays. To address these challenges, we introduce \Dupin{}, a novel parallel processing framework designed for efficient \DDS{} processing in billion-scale graphs. \Dupin{} is powered by a processing engine that exploits the unique properties of the peeling process, with theoretical guarantees on detection quality and efficiency. \Dupin{} provides user-friendly APIs for flexible customization of \DDS{} objectives and ensures robust adaptability to diverse fraud detection scenarios. Empirical evaluations demonstrate that \Dupin{} consistently outperforms several existing \DDS{} methods, achieving performance improvements of up to 100 times compared to traditional approaches. On billion-scale graphs, \Dupin{} demonstrates the potential to enhance the prevention of fraudulent transactions from 45\% to 94.5\% and reduces density error from 30\% to below 5\%, as supported by our experimental results. These findings highlight the effectiveness of \Dupin{} in real-world applications, ensuring both speed and accuracy in fraud detection.
\end{abstract}


\maketitle

\section{Introduction}\label{sec:intro}

In today's digital landscape, graph structures are essential in various applications, including transaction networks~\cite{ye2021gpu,chen2024rush,jiang2024spade+,xu2025bursting}, communication systems~\cite{klimat2004enron,magdon2003locating}, knowledge graphs~\cite{Jiang2020PPKWSAE,jiang2023dkws,jiang2019generic}, and social media platforms~\cite{liao2022dcore,yao2024ublade,fang2016effective,su2025revisiting}. A challenge in these domains is detecting fraudulent activities, which can significantly undermine business integrity and consumer trust. The Densest Subgraph Discovery (\DDS{}) problem, first explored by~\cite{goldberg1984finding}, has emerged as a critical tool in this context, facilitating tasks such as link spam identification~\cite{gibson2005discovering,beutel2013copycatch}, community detection~\cite{dourisboure2007extraction,chen2010dense,niu2025sans}, and notably, fraud detection~\cite{hooi2016fraudar,chekuri2022densest,shin2016corescope}. The peeling process~\cite{tsourakakis2015k,hooi2016fraudar,bahmani2012densest,chekuri2022densest,boob2020flowless} is central to many \DDS{} algorithms. This process removes vertices based on density metrics. \emph{The density metrics are mainly differed by how the weight function is defined. They serve as indicators for the community's suspiciousness—denser subgraphs indicate more intense and cohesive interactions that are more likely to be fraudulent.} Examples of such metrics include vertex degree or the cumulative weight of adjacent edges. While these methods have demonstrated efficiency and robustness, their practical deployment in real-world fraud detection systems faces significant challenges. For instance, business requirements necessitate the detection of fraudulent communities within billion-scale graphs in mere seconds. Recent studies, including those by \cite{dailyreport,ye2021gpu}, reveal that malicious bot activity in e-commerce accounted for a staggering 21.4\% of all traffic in 2018, highlighting the urgency of developing effective detection mechanisms.

\begin{figure}[tb]
    \centering
    \begin{subfigure}{0.49\linewidth}
        \includegraphics[width=\linewidth]{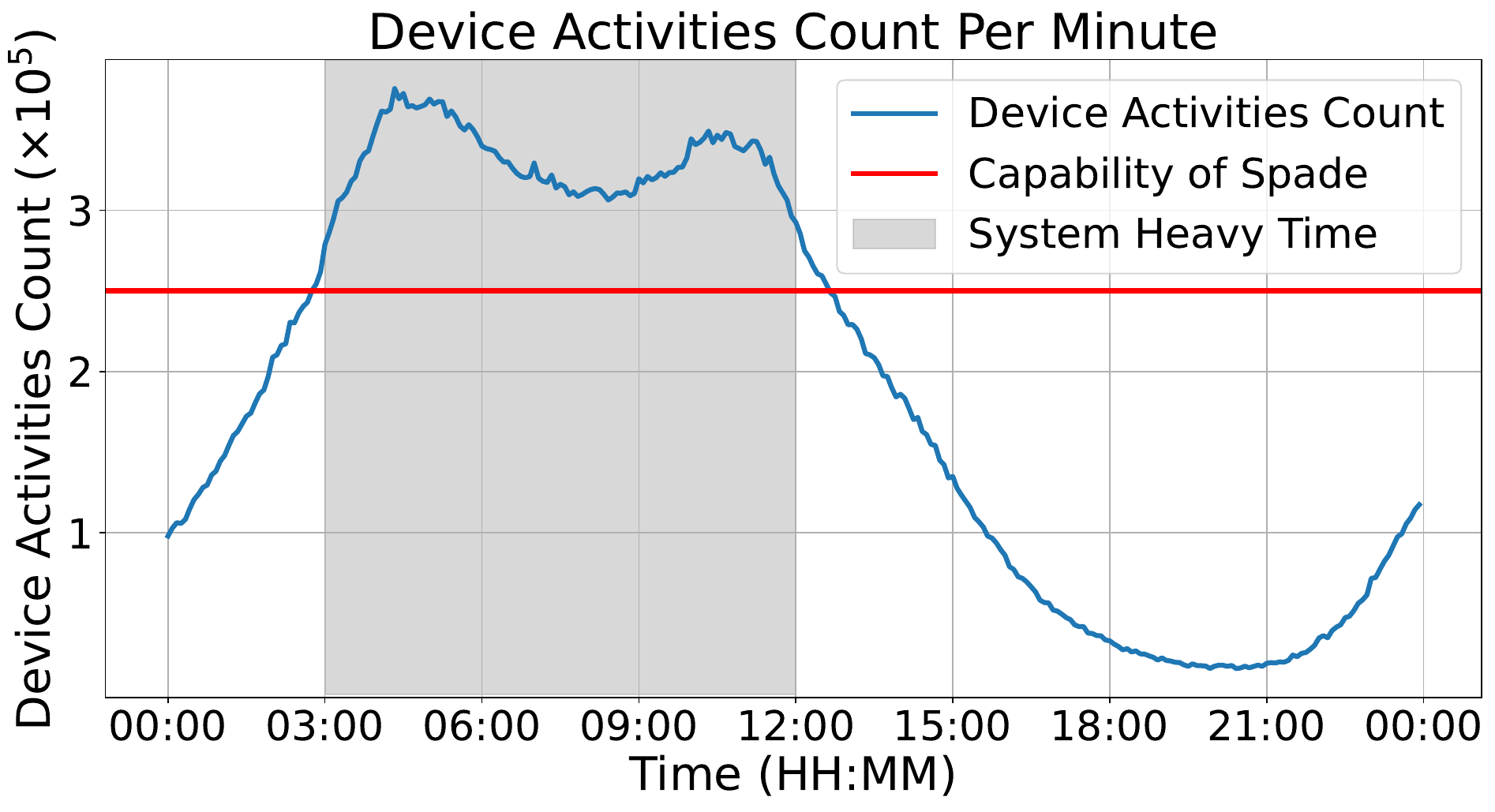}
        \caption{System Load Distribution.}\label{fig:spade-heavy-time}
    \end{subfigure}
    \begin{subfigure}{0.49\linewidth}
        \includegraphics[width=\linewidth]{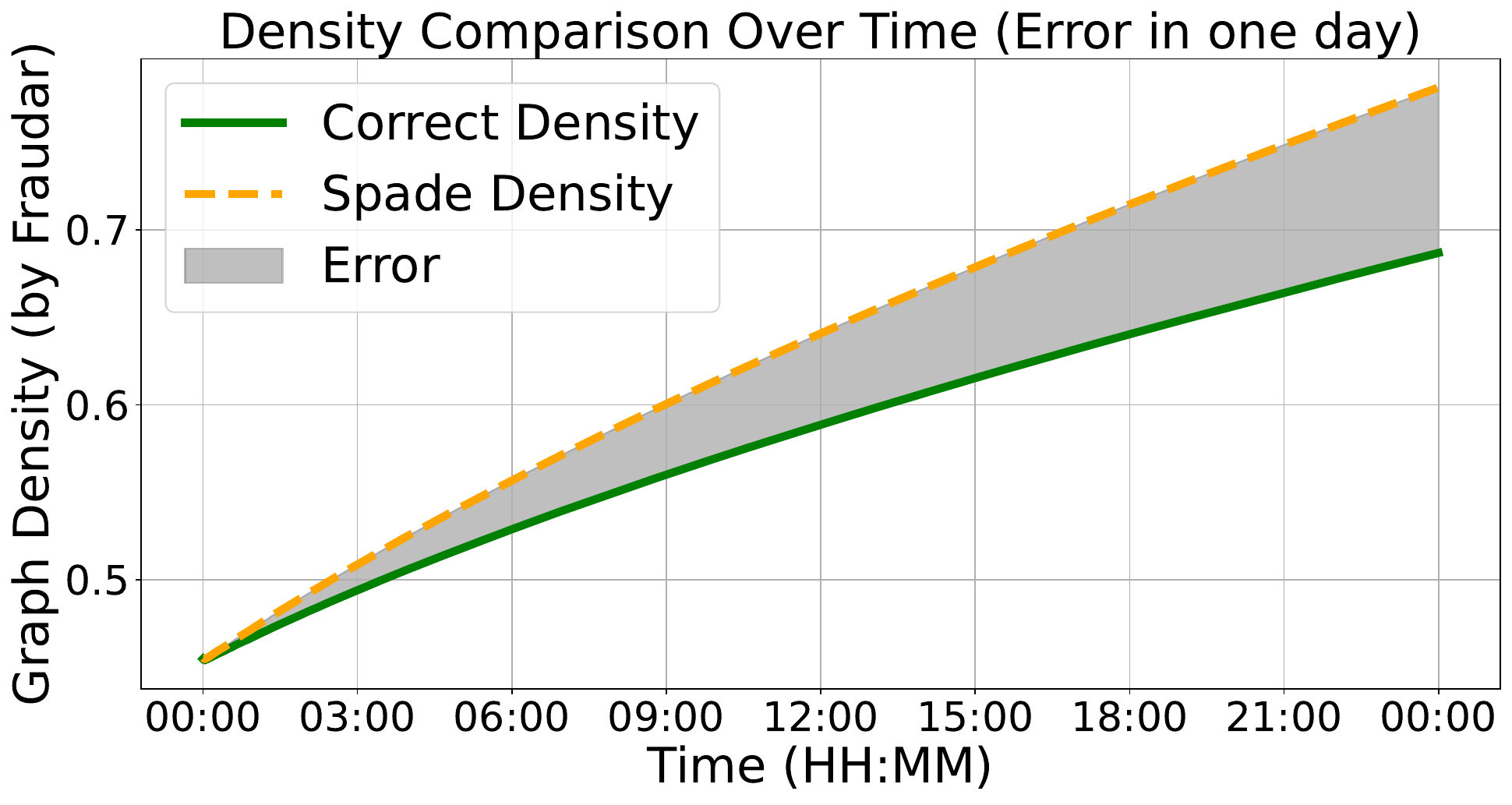}
        \caption{Error Accumulation Trends.}\label{fig:density-comparison}
    \end{subfigure}
    \caption{Activity analysis from our industry partner, Grab. Data normalized for privacy.}\label{fig:heavy}
\end{figure}

\noindent
\stitle{Existing Challenges.} Over the years, extensive research has advanced \DDS{} through incremental peeling algorithms~\cite{bahmani2012densest,epasto2015efficient,shin2017densealert,jiang2023spade,jiang2024spade+} and parallel peeling algorithms~\cite{shi2021parallel,dhulipala2020graph,danisch2018listing,shin2016corescope}. Despite these developments, practical fraud detection in real-world settings continues to face several limitations:
\begin{enumerate}[leftmargin=*]
    \item \textbf{Lack of Flexibility in Density Metrics.} Many existing algorithms rely on fixed scoring functions tailored to \DDS{}. Methods such as \cite{beutel2013copycatch,jiang2014inferring} employ static detection criteria that become less effective as fraudulent behaviors diversify. In fraud detection, for instance, a density metric that incorporates transaction amounts as weights~\cite{gudapati2021search} highlights collusion patterns—such as frequent high-value or consistent transactions indicative of coupon or rebate abuse. Similarly, when detecting fraud communities, the density metric might utilize structural measures (\eg triangle counts~\cite{tsourakakis2015k}) to better capture community characteristics.
    \item \textbf{Inefficient Incremental Updates.} Current incremental algorithms often fail to scale to billion-edge graphs (\(G\)) with high transaction rates (\ie graph updates \(\Delta G\)). For example, consider the state-of-the-art incremental algorithm \Spade{}~\cite{jiang2023spade} (Figure~\ref{fig:spade-heavy-time}), where the vertical axis represents the Device Activities Count (i.e., the number of new edges generated per minute) and the gray area indicates System Heavy Time. \Spade{} can process approximately \(2.65 \times 10^5\) edge updates per minute, which is insufficient to handle scenarios with high transaction volumes. Moreover, these methods often assume static edge weights, even though newly inserted edges can alter the weights of existing edges. For example, the cumulative density error of \cite{jiang2023spade} can exceed 15\% in a single day (Figure~\ref{fig:density-comparison}).
    \item \textbf{Ineffective Parallel Pruning.} Although parallelization can speed up \DDS{} ~\cite{shi2021parallel}, existing parallel frameworks lack efficient pruning strategies, resulting in an excessive number of peeling iterations that degrade performance on large-scale datasets.
\end{enumerate}

\noindent
\stitle{Industry Use Case.} In collaboration with $\Grab$, a leading technology company in Southeast Asia, we address real-world fraud detection challenges on a digital payment and food delivery platform. Two predominant patterns are observed in practice (see Figure~\ref{fig:intro:example}): (i)~\textit{Legitimate Community}, such as popular items attracting buyers, forms large but relatively sparse subgraphs. These communities reflect normal coupon usage and regular customer transactions across merchants. (ii)~\textit{Fraudulent Community}, on the other hand, are characterized by coupon abuse and collusion between customers and merchants. These involve smaller groups of participants engaging in frequent, repetitive transactions, resulting in highly dense subgraphs. The contrasting structures demonstrate the utility of density-based approaches for distinguishing legitimate transactions from fraudulent collusion. \DDS{} algorithms then identify dense subgraphs indicative of fraudulent communities.

\begin{figure}[tb]
        \includegraphics[width=0.9\linewidth]{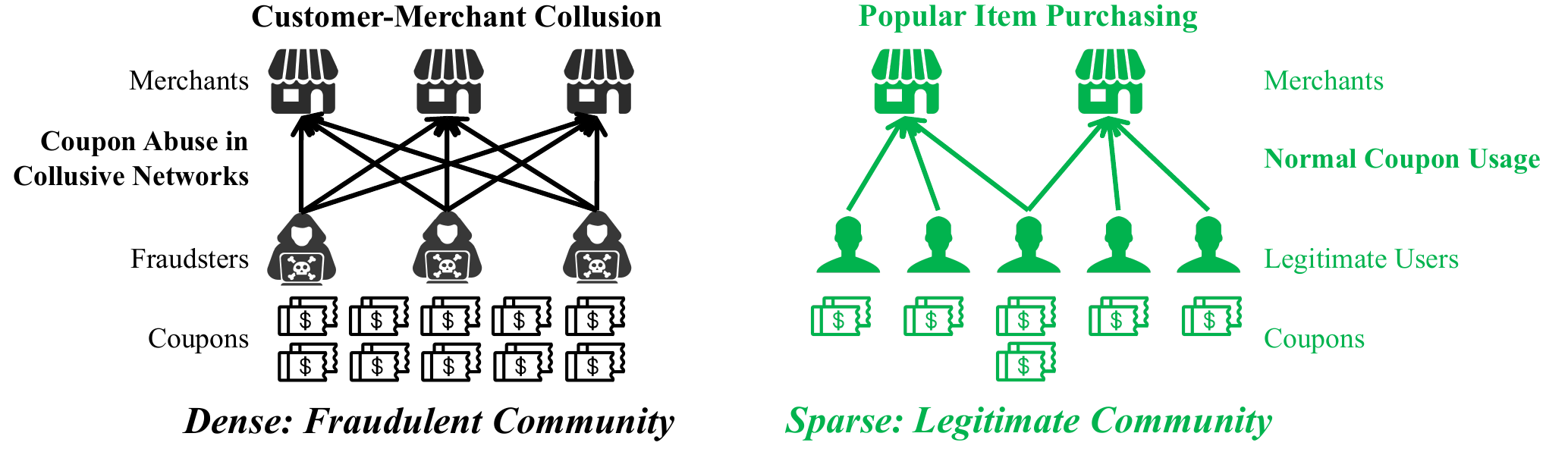}
  \caption{Fraudulent coupon abuse vs. normal coupon usage in customer-merchant networks.}\label{fig:intro:example}
\end{figure}

\noindent
\stitle{Efficiency, Latency, and Prevention Ratio.} Existing \DDS{} approaches struggle to meet real-time latency targets and maintain high fraud-prevention ratios under heavy transaction loads:
\begin{itemize}[leftmargin=*]
    \item \emph{Latency:} A representative incremental peeling approach is \Spade{}~\cite{jiang2023spade,jiang2024spade+}, which leverages existing peeling results on \(G\) and incrementally updates them when new edges or vertices (\(\Delta G\)) arrive \emph{to detect the dense subgraph on} \(G \oplus \Delta G\). However, when \(\Delta G\) is large, \Spade{} incurs substantial overhead from frequent data structure updates and density recalculations. For instance, processing even a single batch of 1K edges on a 2-billion-edge graph (CPU Xeon X5650) can exceed 200 seconds, which is infeasible for real-time detection on modern e-commerce platforms that may see up to 400K edge insertions per minute.
    \item \emph{Prevention Ratio:} Parallel \DDS{} frameworks often compute dense subgraphs from scratch in parallel directly on \(G \oplus \Delta G\). While this approach exploits multiple cores, it typically exhibits latency over 1000 seconds and suffers reduced detection effectiveness when scaling from million-edge to billion-edge datasets (e.g., dropping from 92.47\% to 45\% in fraud prevention).
\end{itemize}

To address these challenges, we pose the following research question: \textit{How can we enhance the flexibility, efficiency, and scalability of \DDS{} methods to effectively detect fraud in massive graphs?}

\stitle{Contributions.} We propose a \DDS{} framework for fraud detection. Our key contributions are:

\begin{itemize}[leftmargin=*]
\item \textbf{A \DDS{} Parallelization Framework.} We propose \Dupin{}, a novel parallel processing framework for Densest Subgraph Discovery (\DDS{}) that breaks the sequential dependencies inherent in traditional density metrics, significantly enhancing scalability and efficiency on both weighted and unweighted graphs. Unlike existing methods that focus on specific density metrics, \Dupin{} allows developers to define custom fraud detection semantics through flexible suspiciousness functions for edges and vertices. It supports a wide range of \DDS{} semantics, including $\DENG$~\cite{charikar2000greedy}, $\DENGW$~\cite{gudapati2021search}, $\Fraudar$~\cite{hooi2016fraudar}, \TDS~\cite{tsourakakis2015k}, and \KCDS{}~\cite{danisch2018listing}. \Dupin{} not only leverages parallelism but also provides theoretical guarantees on approximation ratios for the supported density metrics. To the best of our knowledge, this is the \emph{first} generic parallel processing framework that accommodates a broad class of density metrics in \DDS{}, offering both flexibility and theoretical assurance.
\item \textbf{Exploiting Long-Tail Properties.} We identify and address the \emph{long-tail} problem in the peeling process of \DDS{} algorithms, where late iterations remove only a small fraction of vertices and contribute little to the quality of the densest subgraph, leading to inefficiency. To tackle this, we introduce two novel optimizations: \emph{Global Pruning Optimization}, which proactively removes low-contribution vertices within each peeling iteration to reduce computational overhead, and \emph{Local Pruning Optimization}, which eliminates disparate structures formed after vertex removal to enhance the density and quality of the detected subgraph. These optimizations significantly reduce the number of iterations required, improving both performance and result quality. 
\item \textbf{Empirical Validation and Real-World Impact:} We conduct comprehensive experiments on large-scale industry datasets to validate the effectiveness of \Dupin{}. Our results show that \Dupin{} achieves up to 100 times faster fraud detection compared to the state-of-the-art incremental method \Spade{}~\cite{jiang2023spade} and parallel methods like \PBBS{}~\cite{shi2021parallel}. Moreover, \Dupin{} is capable of preventing up to 94.5\% of potential fraud, demonstrating significant practical benefits and reinforcing system integrity in real-world applications.
\end{itemize}




\stitle{Organization.} The remainder of this paper is structured as follows: Section~\ref{sec:background} presents the background and the literature review. Section~\ref{sec:framework} introduces the architecture of \Dupin{}. In Section~\ref{sec:seqpeel}, we propose the theoretical foundations for parallel peeling-based algorithms for \DDS{}. Section~\ref{sec:peellongtail} introduces long-tail pruning techniques to improve efficiency and effectiveness. The experimental evaluation of \Dupin{} is presented in Section~\ref{sec:exp}. Finally, Section~\ref{sec:conclusion} concludes the paper and outlines avenues for future research.

\section{Background and Related Work}\label{sec:background}

\subsection{Preliminary} \label{subsec:preliminary}

This subsection presents the preliminary concepts and notations. 
\begin{table}[tb]
	  \caption{Frequently used notations.}
	  \setlength{\tabcolsep}{0.5em}
	\label{tab-notations}
	\begin{footnotesize}
		\begin{center}
			\begin{tabular}
				{r|l} \hline \textbf{Notation} & \textbf{Definition} \\
				\hline
                                \hline
                $G$ / $\Delta G$ & a transaction graph / updates to graph $G$ \\ \hline
				$N(u)$ & the neighbors of $u$ \\ \hline
				$a_i$ / $c_{ij}$ & the weight on vertex $u_i$ / on edge $(u_i,u_j)$ \\ \hline
				$f(S)$ & the sum of the suspiciousness of vertex set $S$ \\ \hline
				$g(S)$ & the suspiciousness density of  vertex set $S$ \\ \hline
				$w_{u}(S)$ & the peeling weight, \ie the decrease in $f$ by peeling $u$ from $S$ \\ \hline
				$Q$ & the peeling algorithms \\ \hline
				$S^P$ & the vertex set returned by $Q$ \\ \hline
				$S^*$ & the optimal vertex set \\ \hline
			\end{tabular}
		\end{center}
	\end{footnotesize}
\end{table}

\stitle{Graph $G$.} In this work, we operate on a \textit{weighted} graph denoted as $G=(V,E)$, where $V$ represents the set of vertices and $E$ ($\subseteq V \times V$) represents the set of edges. The graph can be either directed or undirected. Each edge $(u_i,u_j) \in E$ is assigned a \textit{nonnegative} weight, expressed as $c_{ij}$. Additionally, $N(u)$ denotes the set of neighboring vertices connected to vertex $u$.

\stitle{Induced Subgraph.} For a given subset $S$ of vertices in $V$, we define the induced subgraph as $G[S] = (S, E[S])$, where $E[S] = \{(u, v) \mid (u, v) \in E \land u, v \in S\}$. The size of $S$ is denoted by $|S|$.

\stitle{Density Metrics $g$ and Weight Function $f$.} We utilize the class of metrics $g$ as defined in prior works~\cite{hooi2016fraudar,charikar2000greedy,jiang2023spade}, where for a given subset of vertices $S \subseteq V$, the density metric is computed as:
$g(S) = \frac{f(S)}{|S|}$.
Here, $f(S)$ represents the total weight of the induced subgraph $G[S]$, which is defined as the sum of the weights of vertices in $S$ and the weights of edges in $E[S]$. 
More formally:
\begin{equation}\label{eq:density}
    f(S)=\sum_{u_i\in S} a_i + \sum_{(u_i,u_j)\in E[S]} c_{ij}
\end{equation}

The weight of a vertex $u_i$, denoted as $a_i$ ($a_i \geq 0$), quantifies the suspiciousness of user $u_i$. The weight of an edge $(u_i, u_j)$, represented by $c_{ij}$ ($c_{ij} > 0$), reflects the suspiciousness of the transaction between users $u_i$ and $u_j$. Intuitively, the density metric $g(S)$ encapsulates the overall suspiciousness density of the induced subgraph $G[S]$. Larger $g(S)$ indicates a higher density of suspiciousness within $G[S]$. The vertex set that maximizes the density metric $g$ is denoted by $S^*$.

We next introduce five representative density metrics employed in fraud detection applications.

\stitle{Dense Subgraphs ($\DENG{}$)~\cite{charikar2000greedy}.} The $\DENG{}$ metric is employed to quantify the connectivity of substructures within a graph. It has found extensive applications in various domains, such as detecting fake comments in online platforms~\cite{kumar2018community} and identifying fraudulent activities in social networks~\cite{ban2018badlink}. Given a subset of vertices $S \subseteq V$, the weight function of $\DENG{}$ is defined as $f(S) = |E[S]|$, where $|E[S]|$ denotes the number of edges in the induced subgraph $G[S]$.

\stitle{Dense Subgraph on Weighted Graphs ($\DENGW$)~\cite{gudapati2021search}.} In transaction graphs, it is common to have weights assigned to the edges, representing quantities such as the transaction amount. The $\DENGW{}$ metric extends the concept of dense subgraphs to accommodate these weighted relationships. Given a subset of vertices $S \subseteq V$, the weight function of $\DENGW{}$ is defined as $f(S) =\sum_{(u_i,u_j)\in E[S]}c_{ij}$.

\stitle{Fraudar ($\Fraudar{}$)~\cite{hooi2016fraudar}.} In order to mitigate the effects of fraudulent actors deliberately obfuscating their activities, Hooi et al.~\cite{hooi2016fraudar} introduced the $\Fraudar{}$ algorithm. This approach innovatively incorporates weights for edges and assigns prior suspiciousness scores to vertices, potentially utilizing side information to enhance detection accuracy. Given a subset of vertices $S \subseteq V$, the weight function of $\Fraudar{}$ is defined as Equation~\ref{eq:density} where $a_i$ denotes the suspicious score of vertex $u_i$ and $c_{ij} = \frac{1}{\log (x+c)}$. Here $x$ is the degree of the object vertex between $u_i$ and $u_j$, and $c$ is a positive constant ~\cite{hooi2016fraudar}.

\stitle{Triangle Densest Subgraph (\TDS)~\cite{tsourakakis2015k}.} Triangles serve as a crucial structural motif in graphs, offering insights into connectivity and cohesion. In \TDS{}, the density of a subgraph \(G[S]\) is measured by \(f(S) = t(S)\), where \(t(S)\) is the number of triangles in \(S\). Although this directly counts triangles rather than vertices or edges, it can still be expressed in our general density framework \(f(S) = \sum_{u_i \in S} a_i + \sum_{(u_i,u_j)\in E[S]} c_{ij}\) by assigning \(\smash{c_{ij}=0}\) and \(\smash{a_i = \tfrac{t_i}{3}}\), where \(t_i\) is the number of triangles involving vertex \(u_i\). Summing \(\tfrac{t_i}{3}\) over all vertices in \(S\) counts each triangle exactly once, thus yielding \(f(S) = t(S)\). Consequently, \TDS{} fits naturally under the same density metric framework used for other weighted subgraphs.

\stitle{$k$-Clique Densest Subgraph (\KCDS)~\cite{danisch2018listing}.} Extending \TDS{} beyond triangles ($k=3$) to any clique size $k$, \KCDS{} defines the weight function as $f(S) = k(S)$, $k(S)$ is the number of $k$-cliques in $G[S]$. A $k$-Clique is a complete subgraph of $k$ vertices, representing a tightly-knit group of vertices. Just as in the \TDS{} case, \KCDS{} can be embedded in our general framework by assigning $c_{ij} = 0$ for all edges and distributing each $k$-Clique’s contribution among its $k$ vertices. Specifically, let $a_i = k_i/k$, where $k_i$ denotes the number of $k$-Cliques that include vertex $u_i$. Summing $k_i/k$ over all $u_i \in S$ counts each $k$-clique exactly once, thereby recovering $f(S) = k(S)$ within the same density formulation.

\stitle{Problem Statement.} Given a graph $G = (V, E)$ and a community detection semantic with density metric $g(S)$, our goal is to find a vertex subset $S^* \subseteq V$ that maximizes $g(S^*)$.

\subsection{Sequential Peeling Algorithms} \label{subsec:SPA}

Peeling algorithms have gained popularity for their efficiency, robustness, and proven worst-case performance in various applications~\cite{hooi2016fraudar,tsourakakis2015k,charikar2000greedy}. Next, we provide a concise overview of the typical execution flow of these algorithms.

\stitle{Peeling Weight.} We use $w_{u_i}(S)$ to indicate the decrease in the value of \ycc{the weight function} $f$ when the vertex $u_i$ is removed from a vertex set $S$, \ie the \emph{peeling weight}. 



\begin{algorithm}[tb]
    \caption{Sequential Peeling Algorithm}\label{algo:peeling}
    \DontPrintSemicolon
    \footnotesize
    \SetKwFunction{FindMinWeight}{FindMinWeight}
    \SetKwProg{Fn}{Function}{}{}
    \KwIn{Graph $G = (V, E)$, density metric $g(S)$}
    \KwOut{Subset $S^*$ maximizing $g(S)$}
    $S_0 \leftarrow V$, $S^* \leftarrow S_0$, $i \leftarrow 1$ \tcc*{Initialization}\label{algo:peeling:init}

    \While{$S_{i-1} \neq \emptyset$}{
        $u_i \leftarrow \arg\min_{u \in S_{i-1}} w_u(S_{i-1})$ \tcc*{The vertex to be peeled}\label{algo:peeling:select}
        $S_i \leftarrow S_{i-1} \setminus \{u_i\}$ \tcc*{Peel $u_i$ from $S_{i-1}$}\label{algo:peeling:remove}
        \If{$g(S_i) > g(S^*)$}{
            $S^* \leftarrow S_i$ \tcc*{Update optimal subset}\label{algo:peeling:update}
        }
        $i \leftarrow i + 1$ \label{algo:peeling:increment}
    }

    \Return{$S^*$} \tcc*{The subset that maximizes $g(S)$}\label{algo:peeling:return}
    
\end{algorithm}

\noindent
\stitle{Sequential Peeling Algorithm (Algorithm~\ref{algo:peeling}).} The peeling process starts with the full vertex set \( S_0 = V \) and initializes the optimal subset \( S^* \) to \( S_0 \) (Line~\ref{algo:peeling:init}). The algorithm iterates as long as the current subset \( S_{i-1} \) is not empty (Line~\ref{algo:peeling:select}). In each iteration, the vertex \( u_i \) with the smallest peeling weight is identified and peeled from \( S_{i-1} \) to form \( S_i \) (Lines~\ref{algo:peeling:select} and \ref{algo:peeling:remove}). If the density metric \( g(S_i) \) exceeds that of the current optimal subset \( S^* \), \( S^* \) is updated to \( S_i \) (Line~\ref{algo:peeling:update}). This process generates a sequence of nested subsets \( S_0, S_1, \ldots, S_{|V|} \), ultimately returning the subset \( S^* \) that maximizes the density metric \( g(S^*) \) (Line~\ref{algo:peeling:return}). The sequential peeling algorithm is a foundational approach widely adopted in existing studies~\cite{tsourakakis2015k,hooi2016fraudar,bahmani2012densest,chekuri2022densest,jiang2023spade,jiang2024spade+}.

\stitle{Approximation.} An algorithm $Q$ is defined as an $\alpha$-approximation for the Densest Subgraph Discovery (\DDS{}), where $\alpha \geq 1$, if for any given graph it returns a subset $S$ satisfying the condition $g(S) \geq \frac{g(S^*)}{\alpha}$,
where $S^*$ is the optimal subset that maximizes the density metric $g$.

\begin{theorem}[\cite{hooi2016fraudar}]\label{thm:2ppr}
For the vertex set $S^p$ returned by Algorithm~\ref{algo:peeling} and the optimal vertex set $S^*$, it holds that $g(S^p) \geq \frac{g(S^*)}{2}$ for $\DENG$, $\DENGW$ and $\Fraudar$ as the density metrics.
\end{theorem}

The proofs for the theorems presented in this section are provided in Appendix~\ref{appendix:implementation} of~\cite{dupin2024}.

\begin{figure*}[t]
    \includegraphics[width=\linewidth]{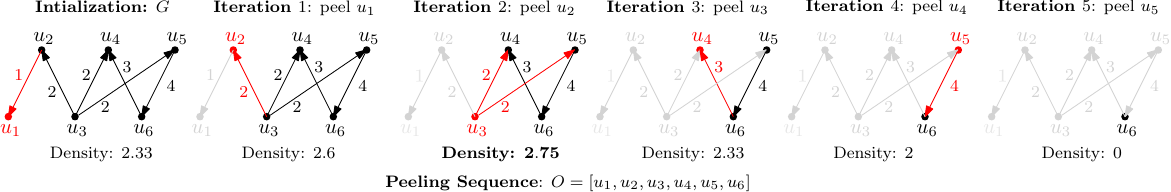}
    \caption{Example of sequential peeling algorithms.}\label{fig:peeling}
\end{figure*}

\begin{tcolorbox}[colback=gray!5, colframe=black,boxrule=0.5pt,boxsep=-2pt]
\begin{example}[Sequential Peeling Algorithm Process] \label{exp:2.1} 
Consider the graph $G$ as shown in Figure~\ref{fig:peeling}. The process begins with an initial graph density of $2.33$. In the first iteration, vertex $u_1$ is peeled due to its smallest peeling weight among all vertices. Subsequently, vertices are peeled in the following order: $u_2$, $u_3$, $u_4$, $u_5$, and $u_6$ based on the updated criteria after each peeling operation. The sequence of peeling continues until the last vertex, $u_6$, is peeled, resulting in a final graph density of $0$. The peeling sequence, therefore, is $O = [u_1, u_2, u_3, u_4, u_5, u_6]$, completely disassembling the graph. Notably, after peeling $u_2$, the graph density increases to 2.75, reaching its maximum before beginning to decrease. This non-monotonic behavior is characteristic of the peeling process. The induced subgraph $G[S]$ of $S = [u_3, u_4, u_5, u_6]$ has the greatest density of $2.75$, thus $G[S]$ is returned.
\end{example}
\end{tcolorbox}

Theorem~\ref{thm:3ppr_cli} extends the analysis to \TDS{} and \KCDS{}, proof is omitted due to space limitation.

\begin{theorem}[\cite{tsourakakis2015k}]\label{thm:3ppr_cli}
For the vertex set $S^p$ returned by Algorithm~\ref{algo:peeling} and the optimal vertex set $S^*$, it holds that $g(S^p) \geq \frac{g(S^*)}{k}$ for \TDS{} and \KCDS{} as the density metrics where $k=3$ for \TDS.
\end{theorem}

The sequential peeling algorithm is designed with a 2-approximation guarantee; that is, when using density metrics such as \(\DENG\), \(\DENGW\), and \(\Fraudar\) to evaluate suspiciousness, the density (and hence the suspiciousness) of the returned subgraph is at least half that of the optimal subgraph—in other words, the optimal subgraph's suspiciousness cannot exceed twice that of the sequential algorithm’s output. In contrast, when the density metrics \TDS\ and \KCDS\ (which are based on triangle counts and \(k\)-Clique counts, respectively) are employed as indicators of suspiciousness, the sequential peeling algorithm guarantees a 3-approximation. Thus, the suspiciousness of the returned subgraph is within a factor of 3 of the optimal value.

\subsection{Related Work}\label{sec:related}

\begin{table}[tb]
\centering
\caption{Comparison of Algorithms Across Key Dimensions.}
\label{tab:algorithm-comparison}
\resizebox{0.85\linewidth}{!}{
\begin{tabular}{lcccc}
\toprule
\textbf{Systems\tablefootnote{Entries marked with (*) indicate methods that support only a single density metric and require modifications to the codebase for more metrics.}} & \textbf{Density Metric Support} & \textbf{Parallelizability} & \textbf{Weighted Graph} & \textbf{Pruning} \\
\midrule
Spade      & $\DENG$, $\DENGW$, $\Fraudar$, \TDS{}, \KCDS{} & Sequential & Yes & No \\
GBBS*      & $\DENG$, $\DENGW$, $\Fraudar$                 & Parallel   & No  & No \\
PKMC*      & $\DENG$, $\DENGW$, $\Fraudar$                 & Parallel   & No  & No \\
FWA*       & $\DENG$, $\DENGW$, $\Fraudar$                 & Parallel   & No  & No \\
ALENEX*    & $\DENG$, $\DENGW$, $\Fraudar$                 & Parallel   & No  & No \\
kCLIST     & \TDS{}, \KCDS{}                               & Parallel   & No  & No \\
PBBS       & \TDS{}, \KCDS{}                               & Parallel   & No  & No \\ \hline
Dupin      & $\DENG$, $\DENGW$, $\Fraudar$, \TDS{}, \KCDS{} & Parallel   & Yes & Yes \\
\bottomrule
\end{tabular}
}
\end{table}

\stitle{Densest Subgraph Discovery (DSD).} DSD represents a core area of research in graph theory, extensively explored in literature~\cite{lee2010survey,gionis2015dense,seidman1983network,sariyuce2018peeling,DBLP:conf/approx/Charikar00,chekuri2022densest}. The foundational max-flow-based algorithm by Goldberg et al.~\cite{goldberg1984finding} set a benchmark for exact DSD, with subsequent studies introducing more scalable exact solutions~\cite{mitzenmacher2015scalable,ma2022convex}. While speed enhancements have been achieved through various greedy approaches~\cite{hooi2016fraudar,ma2022convex,charikar2000greedy,chekuri2022densest}, these methods are inherently sequential and difficult to parallelize. For example, they often peel one vertex at a time, which results in too many peeling iterations. This sequential nature limits their scalability and efficiency on massive graphs. In graph-based fraud detection, techniques like \cite{beutel2013copycatch} and \cite{jiang2014inferring} leverage local search heuristics to identify dense subgraphs in bipartite networks, frequently applied in fraud, spam, and community detection across social and review platforms~\cite{hooi2016fraudar,shin2016corescope,ren2021ensemfdet}. Despite their widespread use, they suffer from efficiency issues in large graphs, often involving numerous iterations and lacking effective pruning techniques, exacerbating the long-tail problem. Moreover, they are limited to specific definitions of dense subgraphs. In contrast, \Dupin{} offers a versatile, parallel DSD approach. It facilitates the simultaneous peeling of multiple vertices, reducing the computational overhead. Notably, \Dupin{} introduces unique global and local pruning techniques that enhance efficiency, allowing it to handle large-scale graphs effectively. Moreover, \Dupin{} provides APIs that allow users to customize definitions of dense subgraphs, ensuring that the resultant graph density adheres to theoretical bounds. This flexibility and efficiency make \Dupin{} a significant advancement over existing methods.

\stitle{DSD on Dynamic Graphs.} To tackle real-time fraud detection, several variants have been designed for dynamic graphs~\cite{epasto2015efficient, shin2017densealert, jiang2023spade,jiang2024spade+,chen2024rush}. \cite{jiang2023spade} addresses real-time fraud community detection on evolving graphs (edge insertions) using an incremental peeling sequence reordering method, while \cite{jiang2024spade+} extends this approach to fully dynamic graphs (edge insertions and deletions). \cite{shin2017densealert} proposes an incremental algorithm that maintains and updates a dense subtensor in a tensor stream. \cite{chu2019online} and \cite{qin2022mining} utilize sliding windows to detect dense subgraphs and bursting cores within specific time windows. These methods have three limitations. First, response time is an issue. Although they respond quickly to benign communities, the results can change drastically for newly formed fraudulent communities, causing significant delays in incremental calculations for fraud detection. Second, efficiency is problematic. While incremental processing is faster than recomputation for a single transaction on billion-scale graphs, frequent incremental evaluations can be slower than a full recomputation. Third, in terms of flexibility, most methods support only a single semantic. In contrast, \Dupin{} introduces a unified parallel framework for peeling-based dense subgraph detection. \Dupin{} can respond within seconds on billion-scale graphs and defines multiple density metrics that can be used on the \Dupin{} platform  addressing the limitations of existing dynamic graph methods.

\stitle{Parallel \DDS{}.} Several parallel methods for \DDS{} have been proposed~\cite{bahmani2012densest,shi2021parallel,dhulipala2020graph,danisch2018listing,shin2016corescope,sariyuce2018peeling}. Bahmani et al.~\cite{bahmani2012densest} proposed a parallel algorithm for \DG{} using the MapReduce framework.  \FWA{}~\cite{danisch2017large} proposed a parallel algorithm for \DW{} based on the Frank-Wolfe algorithm~\cite{jaggi2013revisiting} using OpenMP~\cite{dagum1998openmp}. Some works propose parallel $k$-core algorithms~\cite{shi2021parallel,sukprasert2024practical,luo2023scalable}. Others introduce parallel algorithms for $k$-Clique detection, with approaches offering a $(1+\epsilon)$-approximation ratio and parallel implementations for the \KCDS{} problem~\cite{dhulipala2020graph,sun2020kclist++}. There are also algorithms for $k$-Clique densest subgraph detection using OpenMP~\cite{danisch2018listing}, and methods for dense subtensor detection~\cite{shin2016corescope}. However, these systems have two main limitations. First, they typically support only a single-density metric and do not handle weighted graphs. In real-world applications, transactions between users can have varying degrees of importance based on amounts, frequencies~\cite{jiang2023spade}, and other factors. Extending these systems to support other dense metrics with theoretical guarantees is non-trivial. Second, they often lack efficient pruning techniques, requiring numerous iterations and suffering from the long-tail issue, which hinders scalability on massive graphs. In contrast, \Dupin{} defines the characteristics of a set of supported dense metrics and accommodates weighted subgraphs, allowing for more nuanced analysis reflective of real-world data. Our method introduces efficient pruning strategies by peeling multiple vertices in parallel based on a threshold mechanism, significantly reducing the number of iterations and mitigating the long-tail problem. This makes \Dupin{} more effective and scalable in diverse real-world scenarios compared to previous approaches.


\section{Architecture of \Dupin{}}\label{sec:framework}

Parallelization offers a viable solution to these challenges by distributing the computational load across multiple processors. Modern CPUs are equipped with numerous cores that are cost-effective, enabling parallel processing. If the peeling process of \DDS{} is parallelized, different processors can simultaneously evaluate and peel different nodes, thereby increasing scalability and efficiency. To enable users to design complex \DDS{} algorithms to detect fraudsters based on our framework, \Dupin{}, we have introduced the architecture, APIs, and supported density metrics in this section.



\begin{figure}
        \includegraphics[width=0.9\linewidth]{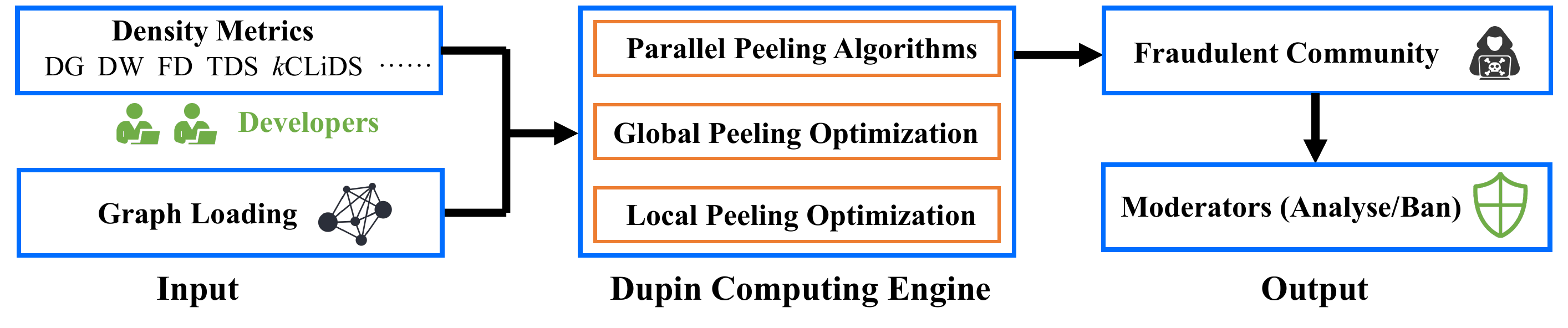}
  \caption{Architecture of \Dupin.}\label{fig:architecture}
\end{figure}

\stitle{Architecture of \Dupin{}.} Figure~\ref{fig:architecture} illustrates the architecture of \Dupin{}, designed to detect fraudulent communities through a flexible and efficient computational pipeline. Developers can define custom fraud detection semantics by leveraging APIs to specify density metrics like $\DENG$, $\DENGW$, $\Fraudar$, \TDS, and \KCDS{}. \Dupin{} begins with graph loading and processes the input through the \Dupin{} computing engine, which incorporates parallel peeling algorithms for scalable processing. To further optimize performance, both global and local peeling optimizations are implemented, enabling faster convergence and improved subgraph density at each iteration. The identified fraudulent communities are flagged and forwarded to moderators for detailed analysis and appropriate actions such as banning suspicious entities. This modular design ensures extensibility and adaptability for various graph-based fraud detection tasks.

\stitle{APIs of \Dupin{} (Figure~\ref{fig:architecture}).} \Dupin{} simplifies the process of defining detection semantics by allowing users to specify a few simple functions. Additionally, \Dupin{} offers APIs that enable users to balance detection speed and accuracy effectively. The key APIs encompass several components designed to facilitate the customization and optimization of fraud detection processes. These components provide a robust framework for users to implement their fraud detection strategies, ensuring both high performance and precision in identifying fraudulent activities. Details are listed below:

\begin{itemize}[leftmargin=*]
    \item \underline{$\mathsf{VSusp}$ and $\mathsf{ESusp}$.} Given a vertex/edge, these components are responsible for deciding the suspiciousness of the endpoint of the edge or the edge with a user-defined strategy.
    \item \underline{$\mathsf{isBenign}$.} This component is used to decide whether a vertex is benign during the whole peeling process (Section~\ref{sec:peellongtail}). If the vertex is benign, it is peeled within the current iteration.
    \item \underline{\textsf{setEpsilon}.} This function is utilized to control the precision of fraud detection. During periods of high system load, moderators can increase $\epsilon$ to achieve higher throughput. Conversely, when the system load is low, $\epsilon$ can be decreased to enhance accuracy.
    \item \underline{\textsf{setK}.} This function defines the $k$-Clique parameter for analysis.
\end{itemize}

\stitle{Example (Listing~\ref{lst:fraudar}).} To implement $\Fraudar$~\cite{hooi2016fraudar} on \Dupin{}, users simply need to plug in the suspiciousness function $\mathsf{vsusp}$ for the vertices by calling $\mathsf{VSusp}$ and the suspiciousness function $\mathsf{esusp}$ for the edges by calling $\mathsf{ESusp}$. Specifically, 1) $\mathsf{vsusp}$ is a constant function, \ie given a vertex $u$, $\mathsf{vsusp}(u) = a_i$ and 2) $\mathsf{esusp}$ is a logarithmic function such that given an edge $(u_i,u_j)$, $\mathsf{esusp}(u_i,u_j) = \frac{1}{\log (x+c)}$, where $x$ is the degree of the object vertex between $u_i$ and $u_j$, and $c$ is a positive constant ~\cite{hooi2016fraudar}. Developers can easily implement customized peeling algorithms with \Dupin{}, significantly reducing the engineering effort. The implementation of other metrics, including $\DENG$, $\DENGW$, \TDS, and \KCDS, follows a similar process. We provide detailed implementations for these metrics in Appendix~\ref{appendix:implementation} of~\cite{dupin2024}.

\begin{minipage}{\columnwidth}
\begin{lstlisting}[language=C++, style=apistyle, caption=Implementation of $\Fraudar{}$ on \Dupin., captionpos=b, label=lst:fraudar]
double vsusp(const Vertex& v, const Graph& g) {
  return g.weight[v]; // Side information on vertex }
double esusp(const Edge& e, const Graph& g) {
  return 1.0 / log(g.deg[e.src] + 5.0); }
int main() {
  Dupin dupin;
  dupin.VSusp(vsusp); // Plug in vsusp
  dupin.ESusp(esusp); // Plug in esusp
  dupin.setEpsilon(0.1);
  dupin.LoadGraph("graph_sample_path");
  vector<Vertex> fraudsters = dupin.ParDetect();
  return 0; }
\end{lstlisting}
\end{minipage}

\noindent
\stitle{Density Metrics in \Dupin{}.} To demonstrate the flexibility of \Dupin{} in accommodating various density metrics, we establish comprehensive sufficient conditions outlined in Property~\ref{property:parallelization}. These conditions ensure that, for any density metric \( g(S) \) satisfying the specified criteria, \Dupin{} can parallelize the peeling process with theoretical guarantees. This enables \Dupin{} to integrate a diverse array of metrics for effectively assessing graph density, which is essential for sophisticated fraud detection and network analysis tasks.

\begin{property}\label{property:parallelization}
Given a density metric \( g(S) \), if the following conditions are satisfied, then \Dupin{} can parallelize the peeling process with theoretical guarantees:
\begin{enumerate}
\item $g(S) = \frac{f(S)}{|S|}$, where $f$ is a monotone increasing function.
\item $a_i \geq 0$, ensuring non-negative vertex weight function.
\item $c_{ij} \geq 0$, ensuring non-negative edge weight function.
\end{enumerate}
\end{property}


\section{Parallelable peeling algorithms}\label{sec:seqpeel}

\ycc{
For sequential peeling algorithms, only one vertex can be removed at a time, with each vertex's removal depending on the removal of previous vertices. This dependency prevents the algorithm from executing multiple vertex removals simultaneously.
}

\ycc{
In this section, we introduce a parallel peeling paradigm to break this dependency chain. This new approach peels vertices in parallel without affecting the programming abstractions that \Dupin{} provides to end users. We then demonstrate how this paradigm offers fine-grained trade-offs between parallelism and approximation guarantees for all density metrics studied in this work. 
}

\subsection{Parallel Peeling Paradigm}\label{sec:batpeel}

Instead of peeling the vertex with the maximum peeling weight, we can peel all vertices that lead to a significant decrease in the density scores of the remaining graph. This decrease is controlled by a tunable parameter $\epsilon$. A larger $\epsilon$ allows more vertices to be peeled in parallel, while a smaller $\epsilon$ provides a better approximation ratio.

\begin{algorithm}[ht]
    \caption{\PEELB: Parallel Peeling}\label{algo:peeling:batch}
    \footnotesize
    \DontPrintSemicolon
    \SetKwProg{Fn}{Function}{}{}
    \SetKwFor{ParallelForEach}{parallel_for}{do}{EndParallelForEach}
    \SetKwFunction{ParSum}{parallel\_sum}
    \KwIn{$G = (V, E)$, a density metric $g(S)$, and $\epsilon \geq 0$}
    \KwOut{A subset of vertices $S$ that maximizes the density metric $g(S)$}
    $S_0 \gets V$ \tcc*[f]{Initialize the subset with all vertices}\\
    $i \gets 1$ \tcc*[f]{Initialize the index}\\
    \While{$S_{i-1} \neq \emptyset$}{
        \ParallelForEach{$u_j \in S_{i-1}$}{
            $a_j \gets \mathsf{vsusp}(u_j)$
        }
        \ParallelForEach{$(u_j, u_k) \in E$ \textbf{and} $u_j, u_k \in S_{i-1}$}{
            $c_{jk} \gets \mathsf{esusp}(u_j, u_k)$
        }
        $f(S_{i-1}) \gets$ \ParSum{$a_j$} $+$ \ParSum{$c_{jk}$} \\
        $g(S_{i-1}) \gets $ $f(S_{i-1})/$ $|S_{i-1}|$\\
        $\tau \gets k(1+\epsilon) \cdot g(S_{i-1})$ \tcc*[f]{The threshold for peeling}\\
        $\remove \gets \emptyset$ \tcc*[f]{Initialize the removal set}\\
        \ParallelForEach{$u \in S_{i-1}$}{
        \If{$w_u(S_{i-1}) \leq \tau$}{
                $\remove \gets \remove \cup \{ u \}$ \tcc*[f]{Vertices to be peeled}
                }
        }        
        $S_i \gets S_{i-1} \setminus \remove$ \tcc*[f]{Update the subset} \\
        $i \gets i + 1$ \tcc*[f]{Increment the index}
        }
    \Return{$\arg\max_{S_i} g(S_i)$} \tcc*[f]{The subset that maximizes $g(S)$}
\end{algorithm}



\stitle{Parallel Peeling (Algorithm~\ref{algo:peeling:batch}).} Let $S_i$ denote the vertex set after the $i$-th peeling iteration. Initially, $S_0 = V$ (Line~\ref{algo:peeling:init}). In each iteration, the algorithm computes vertex and edge suspiciousness in parallel, then aggregates the total weight $f(S_{i-1})$ using a \emph{parallel} summation. Based on the density, the algorithm concurrently evaluates each vertex $u \in S_{i-1}$ to determine if its peeling weight $w_u(S_{i-1})$ satisfies $w_u(S_{i-1}) \leq k(1+\epsilon)g(S_{i-1})$. Vertices meeting this condition are simultaneously added to the removal set $U$. The algorithm then removes all vertices in $U$ from $S_{i-1}$. Parameter $k$ is determined by the metric used. $k=2$ works for $\DENG$, $\DENGW$ and $\Fraudar$. $k=3$ works for \TDS{} and $k$ is the size of the cliques for \KCDS{}. This process is repeated recursively until no vertices are left, resulting in a series of vertex sets $S_0, \ldots, S_{R}$, where $R$ is the number of iterations. The algorithm then returns the set $S_i$ ($i\in [0, R]$) that maximizes $g(S_i)$, denoted as $S^p$. 

In the following two examples, $\DENGW$ and \KCDS{} are the density metrics for \DDS{}, respectively.

\begin{figure}[tb]
   \includegraphics[width=\linewidth]{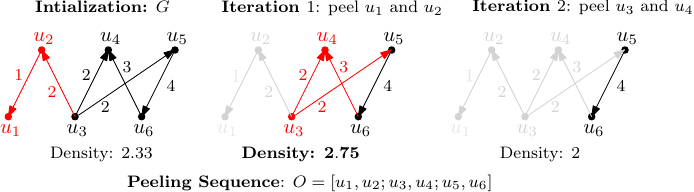}
\caption{Illustration of Parallel \DW{}. Nodes marked in red indicate those to be peeled in the current iteration. The grayed areas represent nodes that have been peeled.}\label{fig:parallel}   
\end{figure}

\begin{figure}[tb]
   \includegraphics[width=\linewidth]{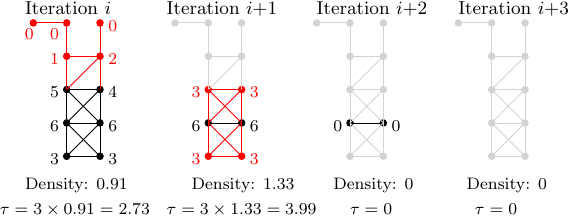}
\caption{Illustration of Parallel \KCDS{} Peeling for $k=3$. Labels adjacent to each node denote the count of $k$-Cliques that the node participates in within the current subgraph.}\label{fig:parallel-cli}   
\end{figure}

\begin{tcolorbox}[colback=gray!5, colframe=black,boxrule=0.5pt,boxsep=-2pt]
\begin{example}[Parallel Peeling Algorithm]\label{example:parallel}
Consider the graph $G$ as illustrated in Figure~\ref{fig:parallel}. The parallel peeling algorithm significantly reduces the number of iterations needed to peel the graph. In the first iteration, the peeling weights of all vertices are computed concurrently. Vertices \(u_1\) and \(u_2\) have peeling weights of 1 and 3, respectively—both of which are below twice the current density—so they satisfy the peeling criteria and are removed simultaneously. This results in an updated graph density of 2.75. The second iteration sees the simultaneous peeling of $u_3$ and $u_4$, leaving the remaining vertices $u_5$ and $u_6$ disconnected and thus effectively peeled from the graph as well. The parallel algorithm completes the peeling process in only two iterations, a notable improvement over sequential methods. Consequently, the peeling sequence is succinctly represented as $O = [u_1, u_2; u_3, u_4; u_5, u_6]$, where semicolons denote the parallel groups peeled in each iteration. The peeling process is non-monotonic; the induced subgraph $G[S]$ with $S = \{u_3, u_4, u_5, u_6\}$ attains the highest density of 2.75 and is thus returned.
\end{example}
\end{tcolorbox}

\begin{tcolorbox}[colback=gray!5, colframe=black,boxrule=0.5pt,boxsep=-2pt]
\begin{example}\label{example:parallel:clique}
In Figure~\ref{fig:parallel-cli}, we demonstrate the prowess of \Dupin{} in executing parallel node peeling for \TDS{} with $k=3$. Consider the $i$-th iteration illustrated on the left: the five vertices in red would require five iterations to peel using sequential algorithms. Prior methodologies, such as those described in~\cite{shi2021parallel}, could necessitate up to three iterations for complete peeling, as they are constrained to peeling vertices with identical weights concurrently and the recalculations of peeling weights are time-intensive. In contrast, \Dupin{} computes, in parallel, the number of \(k\)-Cliques each vertex participates in and checks whether this count falls below the peeling threshold. This parallel evaluation enables the simultaneous removal of all five vertices in a single iteration, thereby producing a denser subgraph for the subsequent iteration. Overall, \Dupin{} reduces the total peeling process to just three iterations to extract the target dense subgraph, highlighting its superior efficiency in leveraging parallelism.
\end{example}
\end{tcolorbox}

\subsection{Theoretical Analysis}
\ycc{
In this section, we analyze the theoretical properties of the parallel peeling process. We begin by examining the number of iterations required for parallel peeling. Built upon the result, we then determine the time complexity and the approximation ratio of the parallel peeling process.}
\eat{
\stitle{Theoretical Analysis Overview.} In Theorem~\ref{theorem:finite}, we prove that Algorithm~\ref{algo:peeling:batch} terminates in a finite number of iterations by demonstrating that the number of vertices peeled in each iteration is non-zero. Additionally, in Theorem~\ref{theorem:2ppr}, we establish that our peeling framework has a theoretical bound on density, ensuring the quality of the detected subgraphs.
}
\begin{lemma}\label{lemma:finite}
Let \( R \) denote the upper bound on the number of rounds required for the peeling process. For any \( \epsilon > 0 \), the peeling process requires at most \( R < \log_{1+\epsilon} |V| \) rounds.
\end{lemma}
\begin{proof}
We prove that for any $\epsilon \geq 0$, the size of the remaining vertex set is bounded by $|S_{i}| < \frac{1}{1+\epsilon}|S_{i-1}|$ for all $i \in [1,R]$. \ycc{We first examine the peeling weights of the vertices in $S_{i-1}$ for the density metrics $\DENG{}$. $\DENGW{}$ and $\Fraudar{}$ where $k=2$.} As each edge weight in the peeling weight calculation contributes twice (once for each endpoint), and the vertex weights contribute once, we have
\begin{equation}\label{eq:twice}
\sum_{u \in S_{i-1}} w_u(S_{i-1}) \leq 2f(S_{i-1}),    
\end{equation}
The peeling weights of the vertices in $S_{i-1}$ are calculated as follows:
\begin{equation}\label{eq:smaller}
\centering
\begin{split}
    \sum\limits_{u\in S_{i-1}} w_{u}(S_{i-1}) & = \sum\limits_{u\in S_{i-1}\setminus S_{i}}w_u(S_{i-1}) + \sum\limits_{u\in S_{i}}w_u(S_{i-1}) \geq \sum\limits_{u\in S_{i}}w_u(S_{i-1}) > 2(1+\epsilon)|S_i| g(S_{i-1})
\end{split}
\end{equation}
Combining Equation~\ref{eq:twice} and Equation~\ref{eq:smaller}, we have:
\begin{equation}
     2(1+\epsilon) |S_i| g(S_{i-1}) < 2f(S_{i-1}) = 2|S_{i-1}|g(S_{i-1})
\end{equation}
For all $i\in [1,R]$,
\begin{equation}
    |S_{i}| < \frac{1}{1+\epsilon}|S_{i-1}|
\end{equation}
\begin{equation}\label{eq:expdecrease}
 \Rightarrow   |V| = |S_0| > \frac{1}{1+\epsilon} |S_1| > \ldots > (\frac{1}{1+\epsilon})^R |S_R|
\end{equation}
Therefore, we have $R < \log_{1+\epsilon}|V|$. 

\ycc{For the density metrics \TDS{} and \KCDS{}, we use the fact that each edge contributes to the peeling weight calculation exactly $k$ times due to the counting of each edge within every $k$-Clique in the graph ($k=3$ for \TDS{}). This multiplicity of edge contributions leads to the relationship:
\begin{equation}\label{eq:ktimes}
\sum_{u \in S_{i-1}} w_u(S_{i-1}) \leq kf(S_{i-1}),    
\end{equation}
}
We then adapt Equations~\ref{eq:smaller}-\ref{eq:expdecrease} and the proof follows.
\end{proof}




\tr{
\begin{lemma}\label{lemma:rm}
Given any $\epsilon \geq 0$, $\forall i \in [1,R]$, $|S_{i-1}| - |S_{i}| > 0$.
\end{lemma}
\begin{proof}
We prove this in contradiction by assuming that $|S_{i-1}| - |S_{i}|=0$. Since $S_{i}\subseteq S_{i-1}$, $S_{i} = S_{i-1}$ which implies that:
\begin{equation}
\forall u \in S_{i-1}, w_u(S_{i-1}) > 2(1+\epsilon)g(S_{i-1}).
\end{equation}
Therefore, we have the following.
\begin{equation}\label{eq:left}
    \sum\limits_{u\in S_{i-1}} w_{u}(S_{i-1}) > 2(1+\epsilon)|S_{i-1}|g(S_{i-1}) = 2(1+\epsilon)f(S_{i-1})
\end{equation}
Since the weights on vertices are calculated once, while those on edges are calculated twice, we have:
\begin{equation}\label{eq:right}
    \sum\limits_{u\in S_{i-1}} w_{u}(S_{i-1}) \leq 2f(S_{i-1})
\end{equation}
The equality is established if all the weights of the vertices are equal to $0$. When we combine Equation~\ref{eq:left} and Equation~\ref{eq:right}, we have the following.
\begin{equation}
    2f(S_{i-1}) > 2(1+\epsilon)f(S_{i-1})
\end{equation}
which contradicts that $\epsilon \geq 0$. Hence, $|S_{i-1}| - |S_{i}| > 0$.
\end{proof}
}

\tr{
Given the guarantee from Lemma~\ref{lemma:nonempty_peel} that at least one vertex is peeled in every iteration, we can deduce that the number of iterations $R$ will not exceed the total number of vertices, $|V|$. This leads us to a more robust and conclusive statement:
}




\stitle{Time Complexity.} Algorithm~\ref{algo:peeling:batch} operates for at most $R = \log_{1+\epsilon}|V|$ iterations, with each iteration scanning all vertices, taking $O(|V|)$ time. As vertices are peeled, updates are necessitated for the peeling weights of their neighboring vertices. Over the course of $R$ iterations, there are no more than $|E|$ such updates for $\DENG$, $\DENGW$ and $\Fraudar$. 
\ycc{
Consequently, Algorithm~\ref{algo:peel2} incurs $O((|E|+ |V|)\log_{1+\epsilon}|V|)$ node/edge weight updates. 
For $\DENG$, $\DENGW$ and $\Fraudar$, the update costs are $O(1)$.
For \TDS{} and \KCDS{}, we follow previous studies~\cite{chiba1985arboricity,danisch2018listing} to compute the peeling weights, and the complexity is $O(k|E|\alpha(G)^{k-2})$, where $\alpha(G)$ is the arboricity of the graph. Hence, the overall complexity is $O(k|E|\alpha(G)^{k-2}(|E|+ |V|)\log_{1+\epsilon}|V|)$. The cost of peeling weight updates is orthogonal to the \Dupin{} framework. 
}

\stitle{Approximation Ratio.}
Next, we are ready to show that \Dupin{} provides a theoretical guarantee with a $k(1+\epsilon)$-approximation for all density metrics studied in this work. We denote $u_i \in S^p$ as the first vertex in $S^*$ peeled by \Dupin{}, and $S'$ as the peeling subsequence starting from $u_i$.

\tr{
\begin{lemma}\label{lemma:opt}
    $\forall u_i\in S^*$, $w_{u_i}(S^*) \geq g(S^*)$.
\end{lemma}
\begin{proof}
We prove it in contradiction. We assume that $w_{u_i}(S^*) < g(S^*)$. By peeling $u_i$ from $S^*$, we have the following.
    \begin{equation}
        \begin{split}
            g(S^{*}\setminus \{u_i\}) & = \frac{f(S^{*}) - w_{u_i}(S^{*})}{|S^{*}|-1} > \frac{f(S^{*}) - g(S_i)}{|S^{*}|-1} \\
            & = \frac{f(S^{*}) - g(S^*)}{|S^{*}|-1} = \frac{f(S^{*}) - \frac{f(S^{*})}{|S^{*}|}}{|S^{*}|-1} = g(S^*)
        \end{split}       
    \end{equation}
Therefore, a better solution can be obtained by removing $u_i$ from $S^*$, which contradicts the fact that $S^*$ is the optimal solution.
\end{proof}
}

\begin{theorem}\label{theorem:2ppr}
For the vertex set $S^p$ returned by \Dupin{} and the optimal vertex set $S^*$, $g(S^p) \geq \frac{g(S^*)}{k(1+\epsilon)}$.
\end{theorem}
\begin{proof}

We first prove that $\forall u_i\in S^*$, $w_{u_i}(S^*) \geq g(S^*)$. We prove it in contradiction. We assume that $w_{u_i}(S^*) < g(S^*)$. By peeling $u_i$ from $S^*$, we have the following.
    \begin{equation}
        \begin{split}
            g(S^{*}\setminus \{u_i\}) & = \frac{f(S^{*}) - w_{u_i}(S^{*})}{|S^{*}|-1} > \frac{f(S^{*}) - g(S_i)}{|S^{*}|-1} = \frac{f(S^{*}) - g(S^*)}{|S^{*}|-1} = \frac{f(S^{*}) - \frac{f(S^{*})}{|S^{*}|}}{|S^{*}|-1} = g(S^*)
        \end{split}       
    \end{equation}
Therefore, a better solution can be obtained by removing $u_i$ from $S^*$, which contradicts the fact that $S^*$ is the optimal solution. We can conclude that $w_{u_i}(S^*) \geq g(S^*)$.

Next, we assume that $u_i \in S^p$ is the first vertex in $S^*$ peeled by the peeling algorithm. Since $S^*$ is the optimal vertex set, we have
    \begin{equation}
        g(S^*) \leq w_{u_i}(S^*),    
    \end{equation}
    because the optimal value $g(S^*)$ must be less than or equal to the weight contribution of any single vertex in $S^*$, in particular $u_i$.
    Since $S^* \subseteq S'$, it is clear that
    \begin{equation}
        w_{u_i}(S^*) \leq w_{u_i}(S')
    \end{equation}
    By the peeling condition, we have 
    \begin{equation}
        w_{u_i}(S') \leq k(1+\epsilon) g(S')
    \end{equation}
    \begin{equation}
    \Rightarrow    g(S^*) \leq w_{u_i}(S^*) \leq w_{u_i}(S') \leq k(1+\epsilon) g(S') \leq k(1+\epsilon) g(S^p)
    \end{equation}
    Hence, we can conclude that $g(S^p) \geq \frac{g(S^*)}{k(1+\epsilon)}$.
\end{proof}


\eat{
\ycc{
Next, we extend our analysis to show that \Dupin{} provides $k(1+\epsilon)$-approximation for \DDS{} with \KCDS{} as the density metric.
The analysis can be directly applied to \TDS{} since \TDS{} is a special for \KCDS{} when $k=3$.
To adapt \KCDS{} as the density metric for parallel peeling described in Algorithm~\ref{algo:peeling:batch}, we need to update the peeling weight $w_u(S_{i-1})$ to reflect the number of k-Cliques removed when $u$ is removed from $S_{i-1}$. The challenge in our analysis lies in the fact that the computation of a node's peeling weight under \KCDS{} extends beyond its immediate neighbors. 
}
}

\eat{
\begin{lemma}\label{lemma:finite:kcds}
Given any $\epsilon > 0$, $R < \log_{1+\epsilon}|V|$.
\end{lemma}
\begin{proof}
    We prove that for any $\epsilon \geq 0$, the size of the remaining vertex set is bounded by $|S_{i}| < \frac{1}{1+\epsilon}|S_{i-1}|$ for all $i \in [1,R]$. For \KCDS{}, each edge contributes to the peeling weight calculation multiple times, precisely $k$ times due to the counting of each edge within every $k$-Clique in the graph. This multiplicity of edge contributions leads to the relationship:
\begin{equation}\label{eq:ktimes}
\sum_{u \in S_{i-1}} w_u(S_{i-1}) = kf(S_{i-1}),    
\end{equation}
The peeling weights of the vertices in $S_{i-1}$ are calculated as follows:
\begin{equation}\label{eq:smaller:kclique}
\centering
\begin{split}
    \sum\limits_{u\in S_{i-1}} w_{u}(S_{i-1}) & = \sum\limits_{u\in S_{i-1}\setminus S_{i}}w_u(S_{i-1}) + \sum\limits_{u\in S_{i}}w_u(S_{i-1}) \\
    & \geq \sum\limits_{u\in S_{i}}w_u(S_{i-1}) > k(1+\epsilon)|S_i| g(S_{i-1})
\end{split}
\end{equation}
Combing Equation~\ref{eq:ktimes} and Equation~\ref{eq:smaller:kclique}, we have:
\begin{equation}
     k(1+\epsilon) |S_i| g(S_{i-1}) < kf(S_{i-1}) = k|S_{i-1}|g(S_{i-1})
\end{equation}
For all $i\in [1,R]$,
\begin{equation}
    |S_{i}| < \frac{1}{1+\epsilon}|S_{i-1}|
\end{equation}
Then we have:
\begin{equation}
    |V| = |S_0| > \frac{1}{1+\epsilon} |S_1| > \ldots > (\frac{1}{1+\epsilon})^R |S_R|
\end{equation}
Therefore, we have $R < \log_{1+\epsilon}|V|$.
\end{proof}
}

\tr{
\ycc{
\begin{lemma}\label{lemma:k_nonempty_peel}
The peeled vertex set $U$ is not empty in Algorithm~\ref{algo:peeling:batch} when \KCDS{} is used as the density metric, \ie $|U_1| > 0$.
\end{lemma}
}
\begin{proof}
Assume for contradiction that in the $i$-th iteration, $U_1 = \emptyset$. This means that for all $u \in S_{i-1}$, we have $w_u(S_{i-1}) > k(1+\epsilon)g(S_{i-1})$. Summing these inequalities over all vertices in $S_{i-1}$, we get
\begin{equation}
\sum_{u \in S_{i-1}} w_u(S_{i-1}) > k(1+\epsilon)|S_{i-1}|g(S_{i-1}) = k(1+\epsilon)f(S_{i-1})    
\end{equation}
%
For \KCDS{}, each edge contributes to the peeling weight calculation multiple times, precisely $k$ times due to the counting of each edge within every $k$-Clique in the graph. This multiplicity of edge contributions leads to the relationship:
\begin{equation}\label{eq:ktimes}
\sum_{u \in S_{i-1}} w_u(S_{i-1}) = kf(S_{i-1}),    
\end{equation}
where the sum of the weights \( w_u(S_{i-1}) \) for all vertices \( u \) in the subset \( S_{i-1} \) is bounded above by \( k \) times the function \( f(S_{i-1}) \), which denotes the total weight of all $k$-Cliques in the subset \( S_{i-1} \). This equality leads to a contradiction since $kf(S_{i-1}) > k(1+\epsilon)f(S_{i-1})$ cannot be true for any $\epsilon \geq 0$. Therefore, our assumption that $U = \emptyset$ must be false, and there must be at least one vertex peeled in each iteration.
\end{proof}
}

\eat{
\begin{theorem}\label{thm:kppr}
For the vertex set $S^p$ returned by \Dupin{} and the optimal vertex set $S^*$, the relationship $g(S^p) \geq \frac{g(S^*)}{k(1+\epsilon)}$ holds.
\end{theorem}
\begin{proof}
    We assume that $u_i \in S^p$ is the first vertex in $S^*$ peeled by the peeling algorithm. Since $S^*$ is the optimal vertex set, we have
    \begin{equation}
        g(S^*) \leq w_{u_i}(S^*),    
    \end{equation}
    because the optimal value $g(S^*)$ must be less than or equal to the weight contribution of any single vertex in $S^*$, in particular $u_i$.
    Since $S^* \subseteq S'$, it is clear that
    \begin{equation}
        w_{u_i}(S^*) \leq w_{u_i}(S')
    \end{equation}
    By the peeling condition, we have 
    \begin{equation}
        w_{u_i}(S') \leq k(1+\epsilon) g(S')
    \end{equation}
    Therefore, we have
    \begin{equation}
        g(S^*) \leq w_{u_i}(S^*) \leq w_{u_i}(S') \leq k(1+\epsilon) g(S') \leq k(1+\epsilon) g(S^p)
    \end{equation}
    where the last term $g(S')$ is at most $g(S^p)$.
Therefore, we can conclude that $g(S^p) \geq \frac{g(S^*)}{k(1+\epsilon)}$
\end{proof} 
}


\stitle{Summary.} 
\ycc{
\Dupin{} can balance efficiency and the worst-case approximation ratio with a single parameter, $\epsilon$. A larger $\epsilon$ allows more vertices to be peeled in parallel, reducing peeling iterations. Meanwhile, the worst-case approximation ratio of parallel peeling is only affected by a factor of $(1+\epsilon)$ compared to the ratio for sequential peeling \wrt all density metrics studied in this work.
}

\eat{
\begin{lemma}\label{lemma:opt2}
    $\forall S\subseteq S^*$, $\frac{w_{S}(S^*)}{|S|} \geq g(S^*)$.
\end{lemma}
\begin{proof}
We prove it in contradiction. We assume that $\frac{w_{S}(S^*)}{|S|} < g(S^*)$. By peeling $S$ from $S^*$, we have the following.
    \begin{equation}
        \begin{split}
            g(S^{*}\setminus S) & = \frac{f(S^{*}) - w_{S}(S^{*})}{|S^{*}|-|S|} > \frac{f(S^{*}) - |S|g(S^*)}{|S^{*}|-|S|} \\
            & = \frac{f(S^{*}) - |S|\frac{f(S^{*})}{|S^*|}}{|S^{*}|-|S|} = \frac{f(S^*)(1-\frac{|S|}{|S^*|})}{|S^{*}|-|S|} =  \frac{f(S^{*})}{|S^*|} = g(S^*)
        \end{split}
    \end{equation}
Therefore, a better solution can be obtained by removing $u_i$ from $S^*$, which contradicts the fact that $S^*$ is the optimal solution.
\end{proof}
}


\section{Long-Tail Pruning}\label{sec:peellongtail}

We observe that the parallel peeling algorithm encounters two significant issues. First, there is a long-tail problem where subsequent iterations fail to produce denser subgraphs. Instead, each iteration only peels off a small fraction of vertices, which significantly impacts the efficiency. Second, during batch peeling, each peeling iteration often results in disparate structures. These are primarily caused by the peeling of certain vertices, which leaves their neighboring structure disconnected and fragmented. This phenomenon affects the continuity and coherence of the subgraphs.

As shown in Table~\ref{tab:peeling_rounds_disconnect_nodes}, on the \texttt{la} dataset (the largest social network dataset tested in our experiments), $\DENG$, $\DENGW$, and $\Fraudar$ experience 24.67\%, 46.84\%, and 3.01\% of rounds as ineffective long-tail iterations, severely limiting response time. Additionally, during the peeling process, if the batch size is large (e.g., $\epsilon=0.5$), approximately 25.34\%, 29.46\%, and 7.16\% of nodes become disconnected and sparse, significantly impacting the quality of the detected subgraphs. 

\begin{table}[h]
\centering
\caption{Impact of $\mathsf{GPO}$ and $\mathsf{LPO}$ on Peeling Rounds for Various Metrics. Experimented on dataset \texttt{la}.}
\label{tab:peeling_rounds_disconnect_nodes}
\resizebox{0.8\linewidth}{!}{
\begin{tabular}{@{}lccc@{}}
\toprule
\textbf{Method} & \textbf{$\DENG$} & \textbf{$\DENGW$} & \textbf{$\Fraudar$} \\ \midrule
Rounds without $\mathsf{GPO}$  & 17{,}637 & 150{,}223 & 112{,}074 \\ \midrule 
Rounds with $\mathsf{GPO}$     & 13{,}287 & 79{,}835 & 108{,}706 \\
Long-tail vertices  & 45{,}017{,}232 & 48{,}248{,}685 & 5{,}658{,}425 \\
\% Reduction in Rounds & 24.67\% & 46.84\% & 3.01\% \\ \midrule 
Rounds with $\mathsf{LPO}$     & 3{,}221 & 10{,}832 & 101{,}255 \\
Sparse vertices     & 13{,}324{,}405 & 15{,}487{,}382 & 3{,}762{,}288 \\
\% Reduction in Rounds & 81.74\% & 92.79\% & 9.65\% \\ \bottomrule
\end{tabular}
}
\end{table}

To address the challenges of high iteration counts in peeling algorithms, \Dupin{} introduces two long-tail pruning techniques: Global Peeling Optimization ($\mathsf{GPO}$) to maintain a global peeling threshold, and Local Peeling Optimization ($\mathsf{LPO}$) to improve the density of the current peeling iteration's subgraph, thereby increasing the peeling threshold. We tested three algorithms, $\DENG$, $\DENGW$, and $\Fraudar$, and found that they required 17,637, 150,223, and 112,074 iterations respectively, which is significantly high. By strategically pruning long-tail vertices and sparse vertices, the number of iterations can be reduced by up to 46.84\%. Additionally, $\mathsf{LPO}$ can further reduce the number of iterations by up to 92.79\%. These optimizations significantly enhance the efficiency of the peeling process, as discussed in this section.

\subsection{Global Peeling Optimization}


Initially, we define what constitutes a \textit{long-tail vertex} during the peeling process. Peeling these long-tail vertices results in ineffective peeling iterations. To address this, we maintain a global maximum peeling threshold to determine which vertices can be peeled directly. When the global peeling threshold is greater than the current iteration's threshold, we use the global peeling threshold.

\begin{algorithm}[ht]
    \caption{\PEELLT: Global Peeling Optimization}\label{algo:peel2}
    \footnotesize
    \DontPrintSemicolon
    \SetKwProg{Fn}{Function}{}{}
    \SetKwFor{ParallelForEach}{parallel_for}{do}{EndParallelForEach}
    \SetKwFunction{ParSum}{parallel\_sum}
    \KwIn{$G = (V, E)$, a density metric $g(S)$, and $\epsilon \geq 0$}
    \KwOut{A subset of vertices $S^p$ that maximizes the density metric $g(S)$}
    $S_0 \gets V$ \tcc*[f]{Initialize the subset with all vertices}\\
    $i \gets 1$ \tcc*[f]{Initialize the index}\\
    $\tau_{\max} \gets 0$ \label{algo:peeling:init:tau} \tcc*[f]{Initialize the global threshold}\\
    \While{$S_{i-1} \neq \emptyset$}{
        \ParallelForEach{$u_j \in S_{i-1}$}{
            $a_j \gets \mathsf{vsusp}(u_j)$
        }
        \ParallelForEach{$(u_j, u_k) \in E$ \textbf{and} $u_j, u_k \in S_{i-1}$}{
            $c_{jk} \gets \mathsf{esusp}(u_j, u_k)$
        }
        $f(S_{i-1}) \gets$ \ParSum{$a_j$} $+$ \ParSum{$c_{jk}$} \\
        $g(S_{i-1}) \gets $ $f(S_{i-1})/$ $|S_{i-1}|$\\
        $\tau_{\max} \gets \max\{ \tau_{\max}, \frac{g(S_{i-1})}{k(1+\epsilon)} \}$ \label{algo:peeling:init:refine}\tcc*[f]{Refine the global threshold}\\
        $\tau \gets \max\{\tau_{\max}, k(1+\epsilon) \cdot g(S_{i-1})\}$ \tcc*[f]{Peeling threshold}\\
        $\remove \gets \emptyset$ \tcc*[f]{Initialize the removal set}\\
        \ParallelForEach{$u \in S_{i-1}$}{
            \If{$w_u(S_{i-1}) \leq \tau$}{
                $\remove \gets \remove \cup \{ u \}$ \tcc*[f]{Vertices to be peeled}
            }
        }
        $S_i \gets S_{i-1} \setminus \remove$ \tcc*[f]{Update the subset}\\
        $i \gets i + 1$ \tcc*[f]{Increment the index}\\
    }    
    \Return{$\arg\max_{S_i} g(S_i)$} \tcc*[f]{The subset that maximizes $g(S)$}
\end{algorithm}


\begin{definition}[Long-Tail vertex] Given an $\alpha$-$approximation$ \DDS{} algorithm and a set of vertices $S\subseteq V$, if $w_u(S) < \frac{g(S^*)}{\alpha}$ and $S^*$ is the optimal vertex set, then $u$ is a long-tail vertex.
\end{definition}

\begin{lemma}\label{lemma:opt}
    If $\exists u \in S_i$ is a long-tail vertex, then $S_i \neq S^{P}$.
\end{lemma}

\begin{proof}
    We prove it in contradiction by assuming that $S_i = S^{P}$. Since $u$ is a long-tail vertex in $S_i$, $w_u(S_i) < \frac{g(S^*)}{\alpha} < g(S^P)$. By peeling $u$ from $S^{P}$, we have the following:
    \begin{equation}
        \begin{split}
            g(S^{P}\setminus \{u\}) & = \frac{f(S^{P}) - w_{u}(S^{P})}{|S^{P}|-1} > \frac{f(S^{P}) - g(S^P)}{|S^{P}|-1} = \frac{f(S^{P}) - \frac{f(S^{P})}{|S^{P}|}}{|S^{P}|-1} = g(S^P)
        \end{split}       
    \end{equation}
Peeling $u$ from $S^P$ yields a denser subgraph, contradicting that $S_i = S^{P}$. Hence, $S_i\not = S^{P}$.
\end{proof}

\stitle{Identifying Long-Tail Vertices}. Based on Lemma~\ref{lemma:opt}, we can also establish that a vertex $u$ in set $S$ can be classified as a long-tail vertex if $w_u(S) < \frac{g(S^*)}{k(1+\epsilon)}$. This allows us to implement a global peeling threshold, denoted as $\tau_{\max}$. During each iteration, we compare $\frac{g(S_i)}{k(1+\epsilon)}$ with $\tau_{\max}$. If the former exceeds the latter, \Dupin{} updates $\tau_{\max}$ to $\frac{g(S_i)}{k(1+\epsilon)}$, thereby enhancing the efficiency of the peeling.

\stitle{Peeling with Global Threshold (Algorithm~\ref{algo:peel2}).} Using $\Fraudar{}$ as an example, we can set $k$ to 2. In Lemma~\ref{lemma:opt}, we recognize that long-tail vertices in set $S_{i}$ can be peeled directly during iteration $i$ for $\Fraudar{}$. To facilitate this, \Dupin{} initiates with a global peeling threshold $\tau_{\max}$ (Line~\ref{algo:peeling:init:tau}). After identifying a denser subgraph post-peeling, $\tau_{\max}$ is updated to $\frac{g(S_i)}{2(1+\epsilon)}$, refining the peeling process (Line~\ref{algo:peeling:init:refine}).


\begin{tcolorbox}[colback=gray!5, colframe=black,boxrule=0.5pt,boxsep=-2pt]
\begin{example}
In the initial depiction provided in the leftmost panel of Figure~\ref{fig:LTP} (the $i$-th iteration), the density of the graph is $4.28$. Traditionally, this scenario would necessitate two additional iterations to complete the peeling process, which could notably increase computational time. As illustrated, vertex $u$ cannot be peeled in the $(i+1)$-th iteration and would typically require an extra iteration for its removal. However, \Dupin{} can promptly peel vertices that fall below this threshold by employing our algorithm's global peeling threshold, which is calculated as half of the current density ($\tau_{\max}=4.28/2 = 2.14$). This approach enables the simultaneous peeling of three nodes in the $(i+1)$-th iteration. Without this pruning strategy, an additional iteration would be essential to peel vertex $u$, underscoring our method's enhanced efficiency in the peeling process.
\end{example}
\end{tcolorbox}

\begin{figure}[tb]
   \includegraphics[width=0.8\linewidth]{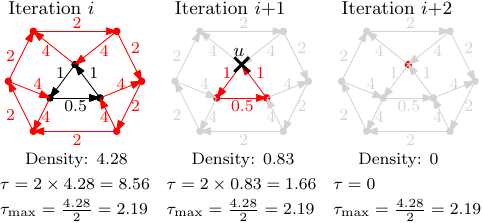}
\caption{Demonstrating Long-Tail Pruning. Nodes highlighted in red are targeted for removal in the current iteration, while the nodes shaded in gray indicate those that have been peeled.}\label{fig:LTP}   
\end{figure}


\subsection{Local Peeling Optimization}

\begin{algorithm}[ht]
    \caption{\PEELLTP: Local Peeling Optimization}\label{algo:peel3}
    \footnotesize
    \DontPrintSemicolon
    \SetKwProg{Fn}{Function}{}{}
    \SetKwFor{ParallelForEach}{parallel_for}{do}{EndParallelForEach}
    \SetKwFunction{ParSum}{parallel\_sum}
    \KwIn{$G = (V, E)$, a density metric $g(S)$, and $\epsilon \geq 0$}
    \KwOut{A subset of vertices $S^p$ representing the densest subgraph}
    $S_0 \gets V$ \tcc*[f]{Initialize the subset with all vertices}\\
    $\tau_{\max} \gets 0$ \tcc*[f]{Initialize the global threshold}\\
    $i \gets 1$ \tcc*[f]{Initialize the index}\\
    \While{$S_{i-1} \neq \emptyset$}{
        \ParallelForEach{$u_j \in S_{i-1}$}{
            $a_j \gets \mathsf{vsusp}(u_j)$
        }
        \ParallelForEach{$(u_j, u_k) \in E$ \textbf{and} $u_j, u_k \in S_{i-1}$}{
            $c_{jk} \gets \mathsf{esusp}(u_j, u_k)$
        }
        $f(S_{i-1}) \gets$ \ParSum{$a_j$} $+$ \ParSum{$c_{jk}$} \\
        $g(S_{i-1}) \gets $ $f(S_{i-1})/$ $|S_{i-1}|$\\
        $\tau_{\max} \gets \max\{ \tau_{\max}, \frac{g(S_{i-1})}{k(1+\epsilon)} \}$ \tcc*[f]{Refine the global threshold}\\
        $\tau_1 \gets \max\{\tau_{\max}, k(1+\epsilon) \cdot g(S_{i-1})\}$ \tcc*[f]{Peeling threshold}\\       
        $\remove_1 \gets \emptyset$ \tcc*[f]{Initialize the first removal set}\\
        \ParallelForEach{$u \in S_{i-1}$}{
            \If{$w_u(S_{i-1}) \leq \tau_1$}{
                $\remove_1 \gets \remove_1 \cup \{ u \}$ \tcc*[f]{Vertices to be peeled}
            }
        }
        $S_i \gets S_{i-1} \setminus \remove_1$ \tcc*[f]{Update the subset}\\        
        \While{$\exists u \in S_i \text{ such that } w_u(S_i) < g(S_i)$ \label{algo:peel3:trimming:start}}{
            $\tau_2 \gets \max\{\tau_{\max}, g(S_i)\}$ \tcc*[f]{Refine the local threshold}\\
            $\remove_2 \gets \emptyset$ \tcc*[f]{Initialize the second removal set}\\
            \ParallelForEach{$u \in S_i$}{
                \If{$w_u(S_i) < \tau_2$}{
                    $\remove_2 \gets \remove_2 \cup \{ u \}$ \tcc*[f]{Vertices to be peeled}
                }
            }
            $S_i \gets S_i \setminus \remove_2$ \label{algo:peel3:trimming:end}\tcc*[f]{Trim the subset} \\
        }
        $i \gets i + 1$ \tcc*[f]{Increment the index} \\
    }
    \Return{$\arg\max_{S_i} g(S_i)$} \tcc*[f]{The subset that maximizes $g(S)$}
\end{algorithm}

        

It is important to note that each peeling iteration often results in disparate structures. This is primarily caused by the peeling of certain vertices, which leaves their neighboring structures disconnected and fragmented. This not only affects the density of the detected subgraph but also lowers the peeling threshold. Lemma~\ref{lemma:opt} establishes a property of vertices within a dense graph.


\begin{lemma}\label{lemma:opt}
    $\forall u\in S$, if $w_{u}(S) < g(S)$, $g(S\setminus \{u\}) > g(S)$.
\end{lemma}
\begin{proof}
By peeling $u$ from $S$, we have the following.
    \begin{equation}
        \begin{split}
            g(S\setminus \{u\}) & = \frac{f(S) - w_{u}(S)}{|S|-1} > \frac{f(S) - g(S_i)}{|S|-1} = \frac{f(S) - g(S)}{|S^{*}|-1} = \frac{f(S) - \frac{f(S)}{|S|}}{|S|-1} = g(S)
        \end{split}       
    \end{equation}
Therefore, a denser graph can be obtained by peeling $u$ from $S$.
\end{proof}

According to Lemma~\ref{lemma:opt}, we introduce local peeling optimization within the iteration. After each peeling iteration, if the peeling weight of a vertex is smaller than the current density, \Dupin{} will trim it. We implement a local peeling optimization (Lines~\ref{algo:peel3:trimming:start}-\ref{algo:peel3:trimming:end}, Algorithm~\ref{algo:peel3}) before the next iteration. After a vertex $u$ is peeled, \Dupin{} invokes \textsf{updateNgh} to update the peeling weights of $u$'s neighboring nodes. If any of these neighbors have a peeling weight less than the current density, they are also peeled, thereby increasing the density of the current iteration. Consequently, when an iteration yields globally optimal results, a denser structure is obtained. The benefits of this approach extend beyond density improvements; performance is also enhanced. By trimming these vertices, we reduce the number of vertices to traverse in subsequent iterations, and a higher density enables more efficient iterations due to the possibility of using larger peeling thresholds.

\begin{lemma}\label{lemma:not_trimmed}
If $u_i \in S^p$ is the first vertex in $S^*$ removed by \Dupin{}, $u_i$ will not be removed during the local peeling optimization process.
\end{lemma}
\begin{proof}
Assume, for the sake of contradiction, that $u_i$ is trimmed during the local peeling optimization process. Let $S'$ be the set of vertices left before $u_i$ is trimmed. By definition, we have:
\[
w_{u_i}(S') < g(S^*).
\]

Since $u_i$ is the first vertex in $S^*$ to be removed, $S^* \subseteq S'$, therefore:
\[
w_{u_i}(S^*) \leq w_{u_i}(S') < g(S^*).
\]

This contradicts Lemma~\ref{lemma:opt}, which states that $\forall u_i \in S^*$:
\[
w_{u_i}(S^*) \geq g(S^*).
\]

Otherwise, a better solution can be obtained by removing $u$ from $S^*$, which contradicts the fact that $S^*$ is the optimal solution. Thus, $u_i$ cannot be trimmed. Hence, $u_i$ will not be removed during the local peeling optimization process.
\end{proof}

Due to Lemma~\ref{lemma:not_trimmed}, $u_i$ will only be removed during the peeling process. Hence, the approximation ratios (Theorem~\ref{theorem:2ppr}) for different density metrics are preserved.

\section{EXPERIMENTAL STUDY}\label{sec:exp}

\subsection{Experimental Setup}\label{subsec:setup}

Our experiments are run on a machine with the following specifications: Intel Xeon X5650 CPU\footnote{The Xeon X5650 was selected to ensure a fair comparison with the setup reported in Spade’s original work. To assess the impact of different hardware, we evaluated the performance of \Spade{}, Dupin, \FWA{}, \GBBS{}, and \PBBS{} on two CPUs in Appendix~\ref{appendix:exp} of~\cite{dupin2024}, showing that Dupin achieves better speedup compared to the other systems in the modern hardware.}, 512 GB RAM, and running on Ubuntu 20.04 LTS. The implementation is memory-resident and developed in C++. All codes are compiled with optimization flag -$O3$.

\begin{table}[tb]
\centering
\footnotesize
\setlength{\tabcolsep}{0.2em}
\caption{Statistics of real-world datasets.}\label{table:Statistics}
\begin{tabular}{c|c|c|c|c|c}
  \hline  
  {\bf Datasets} & {\bf Name} & $\mathbf{|V|}$ & $\mathbf{|E|}$  & \textbf{avg. degree}   & \textbf{Type} \\
  \hline
  \hline
    \texttt{GFG} & \texttt{gfg} & 3,646,185 & 28,635,763 & 17 &  Transaction \\ 
  \hline
  \texttt{soc-twitter}& \texttt{soc} & 28,504,110 & 531,000,244 & 18 & Social network \\ \hline
    \texttt{Web-uk-2005} & \texttt{uk} & 39,454,748 & 936,364,284 & 24 & Web graph  \\ \hline
    \texttt{Twitter-rv}& \texttt{rv} & 41,652,230 & 1,468,365,182 & 35 & Social network\\ \hline
    \texttt{kron-logn21} & \texttt{kron} & 1,544,088 & 91,042,012 & 58 & Cheminformatics  \\ \hline
    \texttt{Web-sk-2005} & \texttt{sk} & 50,636,151 & 1,949,412,601 & 38 & Web graph  \\ \hline
    \texttt{links-anon}& \texttt{la} & 52,579,682 & 1,963,263,821 & 37 & Social network \\ \hline
    \texttt{biomine}& \texttt{bio} & 1,508,587 & 32,761,889 & 22 & Biologic graph \\ \hline
\end{tabular}
\end{table}

\stitle{Datasets.} We report the datasets used in our experiment in Table~\ref{table:Statistics}. One industrial dataset is from $\Grab$ (\texttt{gfg}), which represents a real-world scenario of fraud detection in a large-scale transaction network. This dataset is particularly relevant as it contains diverse transaction patterns and user behaviors, reflecting the complexities of real-world fraud detection. Datasets \texttt{links-anon(\emph{la})} and \texttt{Twitter-rv (\emph{rv})} were obtained from Stanford Network Dataset Collection ~\cite{snapnets}, providing insights into social network structures. Datasets \texttt{kron-logn21(\emph{kron})}, \texttt{soc-twitter(\emph{soc})}, \texttt{Web-uk-2005(\emph{uk})}, and \texttt{Web-sk-2005(\emph{sk})} were obtained from Network Data Repository~\cite{nr}, representing various web graph characteristics. Dataset \texttt{\emph{bio}} was from the BIOMINE~\cite{biomine2019}, which offers biological network structures. The selection of these datasets was based on their ability to represent different types of graphs, thus enhancing the generalizability of our results.

\stitle{Baseline Methods.} We evaluate five commonly used density metrics pivotal in network structure analysis: $\DENG$~\cite{charikar2000greedy}, $\DENGW$~\cite{gudapati2021search}, $\Fraudar$~\cite{hooi2016fraudar}, \TDS{}~\cite{tsourakakis2015k}, and \KCDS{}~\cite{danisch2018listing}. These metrics are implemented within our framework, \Dupin{}, and their performance is compared against several existing frameworks: \Spade{}~\cite{jiang2023spade}, \kClist{}~\cite{danisch2018listing}, \GBBS{}~\cite{dhulipala2020graph}, \PBBS{}~\cite{shi2021parallel}, \PKMC{}~\cite{luo2023scalable}, \FWA{}~\cite{danisch2017large}, and \ALENEX{}~\cite{sukprasert2024practical}. Specifically, \kClist{} and \PBBS{} are parallel algorithms for \KCDS{}. Since \GBBS{} does not natively support weighted graphs, we precompute the peeling weights offline and import them into \GBBS{} to enable support for $\DENGW$ and $\Fraudar$; the reported runtime excludes this preprocessing time. Additionally, \FWA{}, \ALENEX{}, and \PKMC{} support the $\DENG$ metric, and we have extended their algorithms to support $\DENGW$ and $\Fraudar$ by modifying their codebases accordingly. For \TDS{} and \KCDS{}, we adapted the API of \Spade{} for implementation. We utilized the latest codebase $\mathsf{Spade+}$ of \Spade{} to ensure consistency and reliability in our performance measurements. To standardize the comparison, we denote the incremental peeling process—which iteratively removes graph elements based on density scores—as \Spade{}-X, where X corresponds to each metric. This follows the experimental conditions in \cite{jiang2023spade}, with the batch size set to 1K as per their default setting; we report the average runtime per batch for \Spade{}. Our parallel peeling algorithms, incorporating global and local peeling optimizations, are denoted as \PEELLT{}-X and \PEELLTP{}-X, respectively.

\stitle{Default Settings.} In our experiments, the parameter \( \epsilon \) is set to $0.1$ by default, and the number of threads \( t \) is configured to 128. These settings are used consistently unless otherwise specified.

\subsection{Efficiency of \Dupin{}}

\eat{
\begin{figure*}
    \centering
    \includegraphics[width=0.9\linewidth]{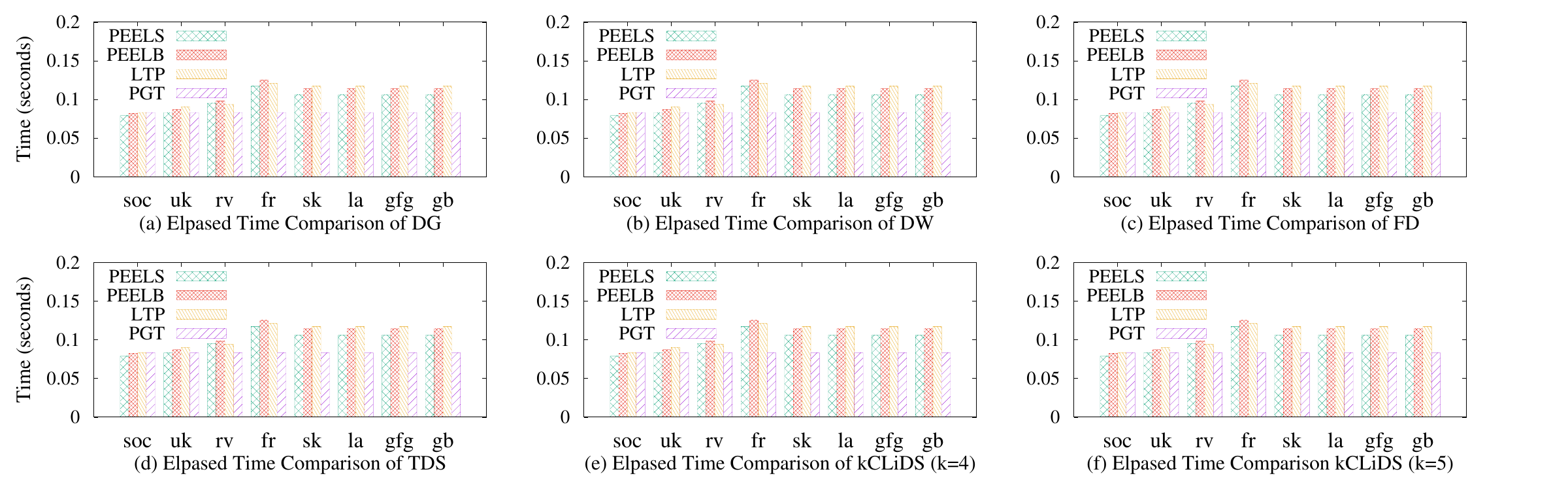}
    \caption{Overall Performance of \Dupin{}}
    \label{fig:overall_performance}
\end{figure*}
}

\begin{figure*}
    \begin{center}
    \includegraphics[width=0.85\linewidth]{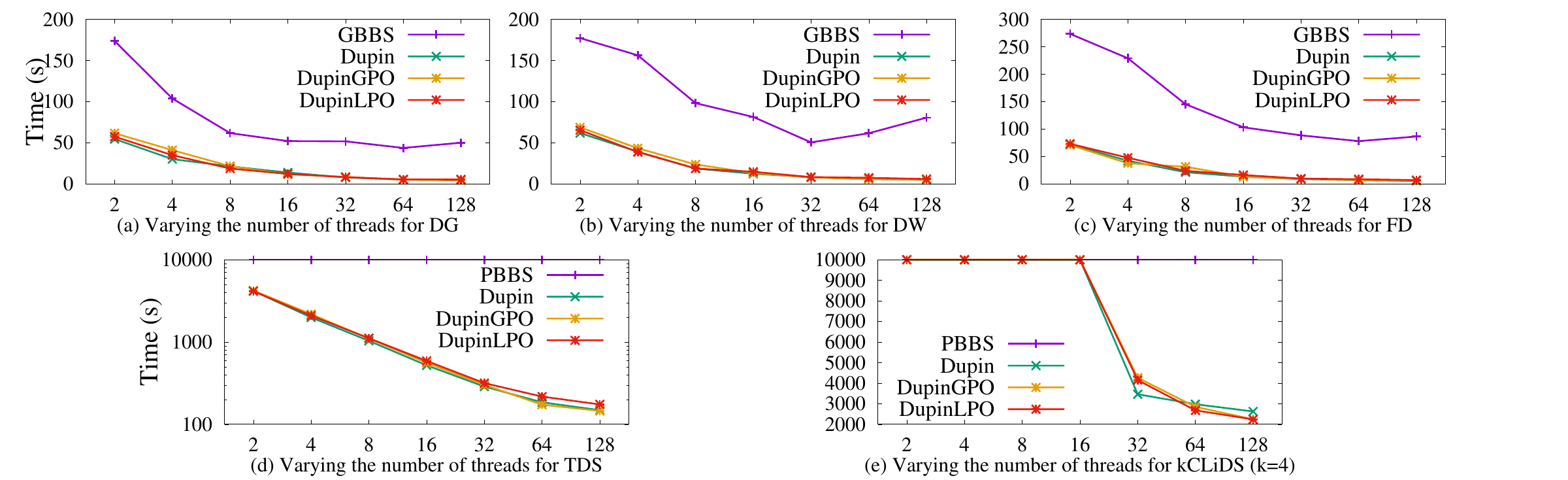}
    \caption{Scalability of \GBBS{}, \PBBS{}, and \Dupin{} by varying the number of the threads.}\label{fig:threads}
    \end{center}
\end{figure*}

\expstitle{\Spade{} vs \PEELB{}.} Our initial comparison focuses on the performance disparity between the incremental and parallel peeling algorithms. Specifically, \PEELB{}-$\DENG$ (resp. \PEELB{}-$\DENGW$, \PEELB{}-$\Fraudar$, \PEELB{}-\TDS{}, and \PEELB{}-\KCDS{}) demonstrates a speedup factor of $6.97$ (resp. $7.71$, $7.71$, $7.65$, and $8.46$) over \Spade{}-$\DENG$ (resp. \Spade{}-$\DENGW$, \Spade{}-$\Fraudar$, \Spade{}-\TDS{}, and \Spade{}-\KCDS{}) as shown in Table~\ref{tab:efficiency-deng} and Table~\ref{tab:efficiency-tds}. Notably, \Spade{}-\TDS{} and \Spade{}-\KCDS{} only complete computations within $7{,}200$ seconds on \texttt{gfg}. The bottleneck for \Spade{} arises in larger graphs where initial calculations of the triangle and $k$-Clique counting are unable to finish. \PEELB{} outpaces \Spade{} because \Spade{} suffers from the frequent reordering of peeling sequences when there are many increments, which proves to be slower than recalculating.

\expstitle{\GBBS{} vs \PEELB{}.} Among the parallel methods for $\DENG$, $\DENGW$, and $\Fraudar$, \FWA{}, \ALENEX{}, and \PKMC{} are several to dozens of times slower than \GBBS{}. Therefore, we select \GBBS{} as a representative for comparison with \PEELB{}. The runtimes of the other three methods are provided in Table~\ref{tab:efficiency-deng}. On average, \PEELB{}-$\DENG$, \PEELB{}-$\DENGW$, and \PEELB{}-$\Fraudar$ demonstrated speedup factors of $6.33$, $12.51$, and $17.07$ times, respectively, over \GBBS{}-$\DENG$, \GBBS{}-$\DENGW$, and \GBBS{}-$\Fraudar$. This performance advantage is attributed to the fundamental difference in the peeling process between the two frameworks. In \GBBS{}, the basic unit of peeling is a bucket, where nodes in the same bucket have an identical peeling weight. This design limits parallelism in weighted graphs, such as those using $\DENGW$ and $\Fraudar$, where many buckets contain only a single node. In contrast, \PEELB{} peels batches of nodes at each step and maintains a global peeling threshold to enhance the efficiency of each iteration and eliminate long-tail iterations. This approach achieves higher parallelism and consistently outperforms \GBBS{}, particularly in weighted graph scenarios where the speedup is even more pronounced.

\expstitle{\kClist{} vs \PBBS{} vs \PEELB{}.} Both \kClist{} and \PBBS{} support the density metrics \TDS{} and \KCDS{}. In our comparisons, \PEELB{} demonstrates notable performance advantages. Specifically, \PEELB{}-\TDS{} is faster than \kClist{}-\TDS{} and \PBBS{}-\TDS{} by $39.85$ and $62.71$ times, respectively. Furthermore, \PEELB{}-\KCDS{} also outperforms \kClist{}-\KCDS{} by $4.16$ times. A significant finding is that \PBBS{}-\KCDS{} fails to complete computations within 7200 seconds on most datasets, highlighting the efficiency of \PEELB{} in handling more complex metrics under stringent time constraints.

\begin{table}[ht]
\centering
\caption{Overall efficiency evaluation conducted using 128 threads for density metrics $\DENG$, $\DENGW$, and $\Fraudar$. The notation `\TLE' signifies ``Time Limit Exceeded,'' indicating that the task did not finish within 7,200 seconds.}
\label{tab:efficiency-deng}
\resizebox{0.85\linewidth}{!}{
\begin{tabular}{l||c|c|c|c||c|c|c|c}
\hline
\textbf{Methods} & \textbf{Dataset} & $\DENG$ & $\DENGW$ & $\Fraudar$ & \textbf{Dataset} & $\DENG$ & $\DENGW$ & $\Fraudar$ \\ \hline
\Spade       & \multirow{6}{*}{\texttt{soc}}    & 23.46 & 28.10 & 30.67 & \multirow{6}{*}{\texttt{sk}}    & 182.28 & 215.27 & 210.26 \\
\GBBS       &     & 10.01 & 23.28 & 35.43 &     & 13.95 & 35.91 & 48.80 \\
\PKMC       &     & 103.16 & 105.77 & 106.96 &     & 388.85 & 393.01  & 431.04 \\
\FWA        &     & 704    & 1,241 & 1,381 &     & 3,092   & 4,934 & 4,716 \\
\ALENEX     &     & 138.73 & 120.19 & 128.81 &     & 411.75 & 488.49 & 473.45 \\ 
\PEELB      &     & \textbf{1.79}   & \textbf{2.26}  & \textbf{2.38}  &     & \textbf{3.87}   & \textbf{3.93}  & \textbf{4.25}  \\ \hline\hline
\Spade      & \multirow{6}{*}{\texttt{uk}}     & 67.08 & 93.66 & 83.20 & \multirow{6}{*}{\texttt{la}}     & 175.07 & 242.26 & 224.47 \\
\GBBS       &     & 7.26  & 27.15 & 50.29 &     & 50.21 & 80.66 & 50.29 \\
\PKMC       &     & 167.76 & 168.81 & 187.61 &     & 946.64 & 1048.48 & 1024.23 \\
\FWA        &     & 1,794   & \TLE{} & \TLE{} &     & 4995   & \TLE & \TLE{} \\
\ALENEX     &     & 238.36 & 209.30 & 213.80 &     & 215.98 & 154.17 & 194.06 \\ 
\PEELB      &     & \textbf{3.38}   & \textbf{3.57}  & \textbf{3.56}  &     & \textbf{4.55}   & \textbf{4.41}  & \textbf{5.16}  \\ \hline\hline

\Spade      & \multirow{6}{*}{\texttt{rv}}     & 135.89 & 135.05 & 132.61 & \multirow{6}{*}{\texttt{bio}}     & 1.99 & 2.02 & 2.11 \\
\GBBS       &     & 32.51 & 62.29 & 71.67 &     & 2.55  & 5.04  & 5.21 \\
\PKMC       &     & 659.04 & 659.15 & 693.83 &     & 11.38 & 11.39 & 12.97 \\
\FWA        &     & 3569   & \TLE{} & \TLE{} &     & 62    & 393 & 428 \\
\ALENEX     &     & 158.64 & 132.88 & 158.79 &     & 62.16 & 68.22 & 63.62 \\ 
\PEELB      &     & \textbf{3.34}   & \textbf{3.76}  & \textbf{3.89}  &     & \textbf{0.15}  & \textbf{0.20}  & \textbf{0.23} \\ \hline\hline

\Spade      & \multirow{6}{*}{\texttt{gfg}}    & 2.30 & 2.62 & 2.70 & \multirow{6}{*}{\texttt{kron}}    & 5.08  & 5.36  & 5.61 \\
\GBBS       &     & 0.41 & 2.10 & 5.74 &     & 1.47  & 6.62  & 9.58 \\
\PKMC       &     & 15.96 & 15.00 & 16.25 &     & 34.47 & 35.71 & 37.23 \\
\FWA        &     & 133    & 186 & 173 &     & 254   & 344 & 372 \\
\ALENEX     &     & 4.13   & 4.01 & 4.75 &     & 85.8  & 83.05 & 82.98 \\ 
\PEELB      &     & \textbf{0.29}   & \textbf{0.33}  & \textbf{0.35}  &     & \textbf{0.18}  & \textbf{0.36}  & \textbf{0.23} \\ \hline
\end{tabular}
}
\end{table}

\begin{table}[ht]
\centering
\caption{Overall efficiency evaluation conducted using 128 threads for density metrics \TDS{} and \KCDS{}.}
\label{tab:efficiency-tds}
\resizebox{0.85\linewidth}{!}{
\begin{tabular}{l||c|c|c||c|c|c}
\hline
\textbf{Methods} & \textbf{Dataset} & \TDS{} & \KCDS{} & \textbf{Dataset} & \TDS{} & \KCDS{} \\ \hline
\Spade       & \multirow{4}{*}{\texttt{soc}}    & \TLE & \TLE & \multirow{4}{*}{\texttt{sk}}    & \TLE & \TLE \\
\kClist      &     & 1516 & 1444 &     & \TLE & \TLE \\
\PBBS        &     & 3524.11 & \TLE &     & \TLE & \TLE \\
\PEELB       &     & \textbf{32.59} & \textbf{283.70} &     & \textbf{42.73} & \textbf{2636.17} \\ \hline\hline
\Spade       & \multirow{4}{*}{\texttt{uk}}     & \TLE & \TLE & \multirow{4}{*}{\texttt{la}}     & \TLE & \TLE \\
\kClist      &     & 494 & 447 &     & 10663 & 8003 \\
\PBBS        &     & 2900.83 & \TLE &     & \TLE & \TLE \\
\PEELB       &     & \textbf{16.66} & \textbf{186.31} &     & \textbf{145.72} & \textbf{2241.69} \\ \hline\hline
\Spade       & \multirow{4}{*}{\texttt{rv}}     & \TLE & \TLE & \multirow{4}{*}{\texttt{bio}}     & \TLE & \TLE \\
\kClist      &     & 5561 & 4655 &     & 230.0 & 255.0 \\
\PBBS        &     & \TLE & \TLE &     & 225.06 & \TLE \\
\PEELB       &     & \textbf{96.32} & \textbf{1009.45} &     & \textbf{6.25} & \textbf{82.54} \\ \hline\hline
\Spade       & \multirow{4}{*}{\texttt{gfg}}    & 5.66 & 5.84 & \multirow{4}{*}{\texttt{kron}}    & \TLE & \TLE \\
\kClist      &     & 10 & 9 &     & 406 & 450 \\
\PBBS        &     & 1.37 & 0.79 &     & 330.18 & \TLE \\
\PEELB       &     & \textbf{0.74} & \textbf{0.69} &     & \textbf{11.69} & \textbf{138.63} \\ \hline
\end{tabular}
}
\end{table}

\eat{
\begin{table*}[ht]
\centering
\caption{Overall efficiency evaluation conducted using 128 threads. The notation `\TLE' signifies ``Time Limit Exceeded,'' indicating that the task did not finish within 7,200 seconds. A dash (``-'') denotes that the framework does not support the density metrics.}
\label{tab:efficiency}
\resizebox{\linewidth}{!}{
\begin{tabular}{l||c|c|c|c|c|c||c|c|c|c|c|c}
\hline
\textbf{Methods} & \textbf{Datasets} & $\DENG$ & $\DENGW$ & $\Fraudar$ & \TDS & \KCDS{} & \textbf{Datasets} & $\DENG$ & $\DENGW$ & $\Fraudar$ & \TDS & \KCDS{}   \\ \hline

\Spade       & \multirow{7}{*}{\texttt{soc}}    & 23.46 & 28.10 & 30.67 & \TLE & \TLE & \multirow{7}{*}{\texttt{sk}}    & 182.28 & 215.27 & 210.26 & \TLE & \TLE \\
\GBBS       &     & 10.01 & 23.28 & 35.43 & \NS{} & \NS{} &     & 13.95 & 35.91 & 48.80 & \NS{} & \NS{} \\
\PKMC       &     &    103.16   &  \NS{}   &   \NS{}     &   \NS{}  &   \NS{}      &     &    388.85   &  393.01   &    \NS{}    &  \NS{}   & \NS{} \\
\FWA       &     &   704    &  \NS{}   &   \NS{}     &   \NS{}  &   \NS{}      &     &    3092   &  \NS{}   &    \NS{}    &  \NS{}   & \NS{} \\
\ALENEX       &     &   138.73    &  \NS{}   &   \NS{}     &   \NS{}  &   \NS{}      &     &    411.75   &  \NS{}   &    \NS{}    &  \NS{}   & \NS{} \\
\PBBS       &     & \NS{} & \NS{} & \NS{} & 3524.11 & \TLE &     & \NS{} & \NS{} & \NS{} & \TLE & \TLE \\
\kClist       &  & \NS{} & \NS{} & \NS{} & 1516 & 1444 &    & \NS{} & \NS{} & \NS{} & \TLE & \TLE  \\
\PEELB       &                                 & 1.79 & 2.26 & 2.38 & 32.59 & 283.70 &                                & 3.87 & 3.93 & 4.25 & 42.73 & 2636.17 \\\hline \hline 
\Spade{}       & \multirow{7}{*}{\texttt{uk}}     & 67.08 & 93.66 & 83.20 & \TLE & \TLE & \multirow{7}{*}{\texttt{la}}     & 175.07 & 242.26 & 224.47 & \TLE & \TLE \\
\GBBS       &     & 7.26 & 27.15 & 50.29 & \NS{} & \NS{} &      & 50.21 & 80.66 & 50.29 & \NS{} & \NS{} \\
\PKMC       &     &    167.76   &  \NS{}   &   \NS{}     &   \NS{}  &   \NS{}      &     &   946.64    &  1,048.41   &    \NS{}    &  \NS{}   & \NS{} \\
\FWA       &     &    1794   &  \NS{}   &   \NS{}     &   \NS{}  &   \NS{}      &     &    4995   &  \NS{}   &    \NS{}    &  \NS{}   & \NS{} \\
\ALENEX       &     &    238.36   &  \NS{}   &   \NS{}     &   \NS{}  &   \NS{}      &     &    215.98   &  \NS{}   &    \NS{}    &  \NS{}   & \NS{} \\
\PBBS       &     & \NS{} & \NS{} & \NS{} & 2900.83 & \TLE &      & \NS{} & \NS{} & \NS{} & \TLE & \TLE \\
\kClist       &  & \NS{} & \NS{} & \NS{} & 494 & 447 &    & \NS{} & \NS{} & \NS{} & 10663 & 8003  \\
\PEELB{}   &                                 & 3.38 & 3.57 & 3.56 & 16.66 & 186.31 &                                & 4.55 & 4.41 & 5.16 & 145.72 & 2241.69 \\ \hline \hline 
\Spade{}       & \multirow{7}{*}{\texttt{rv}}     & 135.89 & 135.05 & 132.61 & \TLE & \TLE & \multirow{7}{*}{\texttt{bio}}     & 1.99 & 2.02 & 2.11 & \TLE & \TLE \\
\GBBS       &      & 32.51 & 62.29 & 71.67 & \NS{} & \NS{} &      & 2.55 & 5.04 & 5.21 & \NS{} & \NS{} \\
\PKMC       &     &   659.04    &  \NS{}   &   \NS{}     &   \NS{}  &   \NS{}      &     &   11.38    &  \NS{}   &    \NS{}    &  \NS{}   & \NS{} \\
\FWA       &     &   3569    &  \NS{}   &   \NS{}     &   \NS{}  &   \NS{}      &     &    62   &  \NS{}   &    \NS{}    &  \NS{}   & \NS{} \\
\ALENEX       &     &    158.64   &  \NS{}   &   \NS{}     &   \NS{}  &   \NS{}      &     &   62.16    &  \NS{}   &    \NS{}    &  \NS{}   & \NS{} \\
\PBBS       &      & \NS{} & \NS{} & \NS{} & \TLE & \TLE &      & \NS{} & \NS{} & \NS{} & 225.06 & \TLE \\
\kClist       &  & \NS{} & \NS{} & \NS{} & 5561 & 4655 &    & \NS{} & \NS{} & \NS{} & 230.0 & 255.0  \\
\PEELB{}   &                                 & \keyres{3.34} & 3.76 & 3.89 & 96.32 & 1009.45 &                                & 0.15 & 0.20 & 0.23 & 6.25 & 82.54 \\ \hline \hline
\Spade{}       & \multirow{7}{*}{\texttt{gfg}}    & 2.30 & 2.62 & 2.70 & 5.66 & 5.84 & \multirow{7}{*}{\texttt{kron}}    & 5.08 & 5.36 & 5.61 & \TLE & \TLE \\
\GBBS       &  & 0.41 & 2.10 & 5.74 & \NS{} & \NS{} &   & 1.47 & 6.62 & 9.58 & \NS{} & \NS{} \\
\PKMC       &     &    15.96   &  \NS{}   &   \NS{}     &   \NS{}  &   \NS{}      &     &   34.47    &  \NS{}   &    \NS{}    &  \NS{}   & \NS{} \\
\FWA       &     &    133   &  \NS{}   &   \NS{}     &   \NS{}  &   \NS{}      &     &   254    &  \NS{}   &    \NS{}    &  \NS{}   & \NS{} \\
\ALENEX       &     &   4.13    &  \NS{}   &   \NS{}     &   \NS{}  &   \NS{}      &     &    85.8   &  \NS{}   &    \NS{}    &  \NS{}   & \NS{} \\
\PBBS       &  & \NS{} & \NS{} & \NS{} & 1.37 & 0.79 &   & \NS{} & \NS{} & \NS{} & 330.18 & \TLE \\ 
\kClist       &  & \NS{} & \NS{} & \NS{} & 10 & 9 &    & \NS{} & \NS{} & \NS{} & 406 & 450  \\
\PEELB       &                                 & 0.29 & 0.33 & 0.35 & 0.74 & 0.69 &                                & 0.18 & 0.36 & 0.23 & 11.69 & 138.63 \\ \hline
\end{tabular}
}
\end{table*}
}

\begin{table}[ht]
\centering
\caption{Overall density evaluation for $\DENG$, $\DENGW$, and $\Fraudar$.}
\label{tab:effectiveness-deng}
\resizebox{0.85\linewidth}{!}{
\begin{tabular}{l||c|c|c|c||c|c|c|c}
\hline
\textbf{Methods} & \textbf{Dataset} & $\DENG$ & $\DENGW$ & $\Fraudar$ & \textbf{Dataset} & $\DENG$ & $\DENGW$ & $\Fraudar$ \\ \hline\hline

\Spade{}       & \multirow{6}{*}{\texttt{soc}}    & 1,307 & 63,372 & 16,947 & \multirow{6}{*}{\texttt{sk}}    & 2,257 & 109,097 & 33,741 \\
\GBBS       &     & \keyres{1,307} & \keyres{63,372} & \keyres{16,883} &     & \keyres{2,257} & \keyres{109,097} & \keyres{28,118} \\
\PKMC       &     &   1,053    &  54,373   &   15,363     &     &    1,954   &  95,241   &    25,868    \\
\FWA       &     &    1,307   &  63,372   &   16,883     &     &    2,257  &  109,097   &    28,188    \\
\ALENEX       &     &  1,307    &  63,372   &   16,883   &     &   2,257   &  109,097   &    28,188    \\
\PEELB       &     & 1,286 & 59,379 & 16,234 &     & 2,235 & 98,669 & 27,067 \\ \hline\hline

\Spade{}       & \multirow{6}{*}{\texttt{uk}}    & 426 & 27,812 & 10,155 & \multirow{6}{*}{\texttt{la}}    & 1,877 & 89,641 & 22,198 \\
\GBBS       &     & \keyres{426} & \keyres{27,812} & \keyres{8,987}  &     & \keyres{1,877} & \keyres{89,641} & \keyres{21,774} \\
\PKMC       &     &  375    &  24,279   &   8,268     &     &    1,512   &  83,049   &    19,637    \\
\FWA       &     &    486   &  \TLE{}   &   \TLE{}     &     &  1,877    &  \TLE   &    \TLE{}    \\
\ALENEX       &     &   486   &  27,812   &   8,987     &     &   1,877   &  89,641   &    21,774    \\
\PEELB       &     & 309 & 24,796 & 8,424 &     & 1,843 & 87,671 & 20,610 \\ \hline\hline

\Spade{}       & \multirow{6}{*}{\texttt{rv}}    & 1,643 & 74,779 & 22,678 & \multirow{6}{*}{\texttt{bio}}    & 777 & 36,446 & 13,039 \\
\GBBS       &     & \keyres{1,643} & \keyres{74,779} & \keyres{20,749}   &     & \keyres{777} & \keyres{36,446} & \keyres{12,669} \\
\PKMC       &     &   1,437    &  67,328   &   16,753     &     &   721    &  32,284   &    11,283    \\
\FWA       &     &    1,643   &  \TLE{}   &   \TLE{}     &     &   787   &  36,446   &    12,669    \\
\ALENEX       &     &   1,643   &  74,779   &   20,749     &     &   787   &  36,446   &    12,669    \\
\PEELB       &     & 1,518 & 71,058 & 18,115 &     & 699 & 31,172 & 10,404 \\ \hline\hline

\Spade{}       & \multirow{6}{*}{\texttt{gfg}}    & 28 & 1,432 & 5,369 & \multirow{6}{*}{\texttt{kron}}    & 1,177 & 53,539 & 15,381 \\
\GBBS       &     & \keyres{28} & 1,432 & 5,018 &     & \keyres{1,177} & \keyres{53,539} & \keyres{14,861} \\
\PKMC       &     &   28    &  1,396   &   4,782     &     &   1,169    &  49,983   &    12,829    \\
\FWA       &     &   28    &  1,432  &   5,108     &     &    1,177  &  53,539   &    14,861    \\
\ALENEX       &     &   28   &  1,432   &   5,108     &     &   1,177   &  53,539   &    14,861    \\
\PEELB       &     & 26 & 1,405 & 4,879 &     & 1,177 & 52,695 & 13,912 \\ \hline

\end{tabular}
}
\end{table}

\begin{table}[ht]
\centering
\caption{Overall density evaluation for \TDS{} and \KCDS{}.}
\label{tab:effectiveness-tds}
\resizebox{0.85\linewidth}{!}{
\begin{tabular}{l||c|c|c||c|c|c}
\hline
\textbf{Methods} & \textbf{Dataset} & \TDS{} & \KCDS{} & \textbf{Dataset} & \TDS{} & \KCDS{} \\ \hline\hline

\Spade{}       & \multirow{4}{*}{\texttt{soc}}    & \TLE & \TLE & \multirow{4}{*}{\texttt{sk}}    & \TLE & \TLE \\
\kClist       &    & 1,525,517 & 645,536,400 &     & \TLE & \TLE  \\
\PBBS       &     & \keyres{1,994,617} & \TLE &     & \TLE & \TLE \\
\PEELB       &     & 1,533,939 & 606,708,980 &     & 9,995,529 & 14,890,921,876 \\ \hline\hline

\Spade{}       & \multirow{4}{*}{\texttt{uk}}    & \TLE & \TLE & \multirow{4}{*}{\texttt{la}}    & \TLE & \TLE \\
\kClist       &    & 188,524 & 44,272,600 &     & 3,865,986 & 3,306,718,000  \\
\PBBS       &     & \keyres{304,577} & \TLE &      & \TLE & \TLE \\
\PEELB       &     & 187,260 & 33,823,488 &     & 3,974,028 & 3,167,144,896 \\ \hline\hline

\Spade{}       & \multirow{4}{*}{\texttt{rv}}    & \TLE & \TLE & \multirow{4}{*}{\texttt{bio}}    & \TLE & \TLE \\
\kClist       &    & 3,561,222 & 3,280,083,600 &     & 1,043,100 & 542,860,800  \\
\PBBS       &     & \TLE & \TLE &     & \keyres{1,187,388} & \TLE \\
\PEELB       &     & 3,769,671 & 2,997,257,620 &     & 1,138,083 & 544,047,728 \\ \hline\hline

\Spade{}       & \multirow{4}{*}{\texttt{gfg}}    & 0 & 0 & \multirow{4}{*}{\texttt{kron}}    & \TLE & \TLE \\
\kClist       &    & 0 & 0 &     & 1,447,859 & 424,908,400  \\
\PBBS       &    & 0 & 0 &     & \keyres{1,447,859} & \TLE   \\
\PEELB       &     & 0 & 0 &     & 1,447,788 & 426,677,504 \\ \hline

\end{tabular}
}
\end{table}

\eat{
\begin{table*}[ht]
\centering
\caption{Overall effectiveness evaluation. The notation `\TLE' signifies ``Time Limit Exceeded,'' indicating that the task did not finish within 7,200 seconds. A dash (``-'') denotes that the framework does not support the specified density metrics.}
\label{tab:effectiveness}
\resizebox{0.85\linewidth}{!}{
\begin{tabular}{l||c|c|c|c|c|c||c|c|c|c|c|c}
\hline
\textbf{Methods} & \textbf{Datasets} & $\DENG$ & {$\DENGW$} & {$\Fraudar$} & \TDS & \KCDS{} & \textbf{Datasets} & {$\DENG$} & {$\DENGW$} & {$\Fraudar$} & \TDS & \KCDS{}   \\ \hline
\hline
\Spade{}       & \multirow{7}{*}{\texttt{soc}}    & 1,307 & 63,372 & 16,947 & \TLE & \TLE & \multirow{7}{*}{\texttt{sk}}    & 2,257 & 109,097 & 33,741 & \TLE & \TLE \\
\GBBS       &     & \keyres{1,307} & \keyres{63,372} & \keyres{16,883} & - & - &     & \keyres{2,257} & \keyres{109,097} & \keyres{28,118} & - & - \\
\PKMC       &     &   1,053    &  \NS{}   &   \NS{}     &   \NS{}  &   \NS{}      &     &    \NS{}   &  \NS{}   &    \NS{}    &  \NS{}   & \NS{} \\
\FWA       &     &    1,307   &  \NS{}   &   \NS{}     &   \NS{}  &   \NS{}      &     &    2,257  &  \NS{}   &    \NS{}    &  \NS{}   & \NS{} \\
\ALENEX       &     &  1,307    &  \NS{}   &   \NS{}     &   \NS{}  &   \NS{}      &     &   2,257   &  \NS{}   &    \NS{}    &  \NS{}   & \NS{} \\
\PBBS       &     & - & - & - & \keyres{1,994,617} & \TLE &     & - & - & - & \TLE & \TLE \\
\kClist       &    & \NS & \NS & \NS  & 1,525,517 & 645,536,400 &     & \NS & \NS & \NS & \TLE & \TLE  \\
\kClist++       &  & \NS{} & \NS{} & \NS{} &  &  &    & \NS{} & \NS{} & \NS{} &  &   \\
\PEELB       &                                 & 1,286 & 59,379 & 16,234 & 1,533,939 & 606,708,980 &                                & 2,235 & 98,669 & 27,067 & 9,995,529 & 14,890,921,876 \\ \hline \hline
\Spade{}       & \multirow{7}{*}{\texttt{uk}}    & 426 & 27,812 & 10,155 & \TLE & \TLE & \multirow{7}{*}{\texttt{la}}    & 1,877 & 89,641 & 22,198 & \TLE & \TLE \\
\GBBS       &     & \keyres{426} & \keyres{27,812} & \keyres{8,987}  & - & - &     & \keyres{1,877} & \keyres{89,641} & \keyres{21,774} & - & - \\
\PKMC       &     &  \NS{}    &  \NS{}   &   \NS{}     &   \NS{}  &   \NS{}      &     &    1,512   &  \NS{}   &    \NS{}    &  \NS{}   & \NS{} \\
\FWA       &     &    486   &  \NS{}   &   \NS{}     &   \NS{}  &   \NS{}      &     &  1,877    &  \NS{}   &    \NS{}    &  \NS{}   & \NS{} \\
\ALENEX       &     &   486   &  \NS{}   &   \NS{}     &   \NS{}  &   \NS{}      &     &   1,877   &  \NS{}   &    \NS{}    &  \NS{}   & \NS{} \\
\PBBS       &      & - & - & - & \keyres{304,577} & \TLE &      & - & - & - & \TLE & \TLE \\
\kClist       &    &\NS & \NS & \NS & 188,524 & 44,272,600 &     & \NS & \NS & \NS  & 3,865,986 & 3,306,718,000  \\
\PEELB       &                                 & 309 & 24,796 & 8,424 & 187,260 & 33,823,488 &                                & 1,843 & 87,671 & 20,610 & 3,974,028 & 3,167,144,896 \\ \hline \hline
\Spade{}       & \multirow{7}{*}{\texttt{rv}}    & 1,643 & 74,779 & 22,678 & \TLE & \TLE & \multirow{7}{*}{\texttt{bio}}    & 777 & 36,446 & 13,039 & \TLE & \TLE \\
\GBBS       &     & \keyres{1,643} & \keyres{74,779} & \keyres{20,749}   & - & - &     & \keyres{777} & \keyres{36,446} & \keyres{12,669} & - & - \\
\PKMC       &     &   1,437    &  \NS{}   &   \NS{}     &   \NS{}  &   \NS{}      &     &   721    &  \NS{}   &    \NS{}    &  \NS{}   & \NS{} \\
\FWA       &     &    1,643   &  \NS{}   &   \NS{}     &   \NS{}  &   \NS{}      &     &   787   &  \NS{}   &    \NS{}    &  \NS{}   & \NS{} \\
\ALENEX       &     &   1,643   &  \NS{}   &   \NS{}     &   \NS{}  &   \NS{}      &     &   787   &  \NS{}   &    \NS{}    &  \NS{}   & \NS{} \\
\PBBS       &     & - & - & - & \TLE & \TLE &     & - & - & - & \keyres{1,187,388} & \TLE \\
\kClist       &    & \NS & \NS & \NS & 3,561,222 & 3,280,083,600 &     & \NS & \NS & \NS & 1,043,100 & 542,860,800  \\
\PEELB       &                                 & 1,518 & 71,058 & 18,115 & 3,769,671 & 2,997,257,620 &                                & 699 & 31,172 & 10,404 & 1,138,083 & 544,047,728 \\ \hline \hline
\Spade{}       & \multirow{7}{*}{\texttt{gfg}}    & 28 & 1,432 & 5,369 & \NS & \NS & \multirow{7}{*}{\texttt{kron}}    & 1,177 & 53,539 & 15,381 & \TLE & \TLE \\
\GBBS       &     & \keyres{28} & 1,432 & 5,018 & \NS{} & \NS{} &     & \keyres{1,177} & \keyres{53,539} & \keyres{14,861} & \NS & \NS \\
\PKMC       &     &   \NS{}    &  \NS{}   &   \NS{}     &   \NS{}  &   \NS{}      &     &   1,169    &  \NS{}   &    \NS{}    &  \NS{}   & \NS{} \\
\FWA       &     &   \NS{}    &  \NS{}   &   \NS{}     &   \NS{}  &   \NS{}      &     &    1,177  &  \NS{}   &    \NS{}    &  \NS{}   & \NS{} \\
\ALENEX       &     &   \NS{}   &  \NS{}   &   \NS{}     &   \NS{}  &   \NS{}      &     &   1,177   &  \NS{}   &    \NS{}    &  \NS{}   & \NS{} \\
\PBBS       &    & \NS{} & \NS{} & \NS{} & - & - &     & \NS{} & \NS{} & \NS{} & \keyres{1,447,859} & \TLE   \\
\kClist       &    & \NS & \NS & \NS & \NS & \NS &     & \NS & \NS & \NS & 1,447,859 & 424,908,400  \\
\PEELB       &                                 & 26 & 1,405 & 4,879 & \NS{} & \NS{} &                                & 1,177 & 52,695 & 13,912 & 1,447,788 & 426,677,504 \\ \hline
\end{tabular}
}
\end{table*}
}

\expstitle{Impact of \PEELLT{}.} Next, we evaluate the acceleration effect of \PEELLT{}, using \PEELB{} as a baseline for comparison. Notably, \PEELLT{}-$\DENG$ is found to be $1.14$ times faster, \PEELLT{}-$\DENGW$ is $1.11$ times faster, \PEELLT{}-\TDS{} is $1.10$ times faster, and \PEELLT{}-\KCDS{} is $1.35$ times faster than their respective \PEELB{} implementations. The acceleration is particularly significant for \KCDS{} on larger datasets, such as \texttt{soc},  \texttt{rv}, \texttt{kron}, and \texttt{sk}. \PEELLT{}-\KCDS{} is $1.40$, $1.54$, $1.51$, and $1.61$ times faster than \PEELB{}-\KCDS{} on \texttt{soc},  \texttt{rv}, \texttt{kron}, and \texttt{sk}, respectively. This is attributed to the greater number of iterations required due to the larger dataset sizes and the higher count of $k$-Cliques, necessitating more peeling iterations. In such scenarios, \PEELLT{} effectively prunes the tail iterations early, thereby reducing the number of rounds needed.

\expstitle{Impact of \PEELLTP{}.} We next compared the efficiency of \PEELB{} and \PEELLTP{}. On average, though \PEELLTP{}-$\DENG{}$ was 10.09\% slower than \PEELB{}-$\DENG{}$, the speed is still comparable. For example, on the \texttt{soc} dataset, \PEELLTP{}-$\DENG{}$ was 5.94\% faster than \PEELB{}-$\DENG{}$. This is because \PEELLTP{} requires some lock operations to update the peeling weights of neighboring vertices during local peeling optimization. However, the advantage is that \PEELLTP{} can detect up to 26.26\% denser subgraphs, as detailed in Section~\ref{sec:effectiveness}. On larger datasets, such as \texttt{soc}, \texttt{rv}, and \texttt{la}, the vertices being trimmed have relatively low correlations with each other, so the impact of locks is minimal. For example, on the \texttt{soc} (resp. \texttt{rv}), \PEELLTP{}-$\DENG{}$ was 5.94\% (resp. 9.94\%) faster than \PEELB{}-$\DENG{}$. In contrast, on smaller datasets, such as \texttt{bio} and \texttt{kron}, the extra cost of locks is higher.


\expstitle{Impact of the Number of Threads \( t \).} All methods are influenced by the concurrency level, \ie the number of threads. Generally, the greater the number of threads, the faster the execution speed. In our experiments, using the \texttt{la} dataset as a benchmark, we varied the number of threads \( t \) from $2$ to $128$ to observe the impact on runtime performance. For \GBBS-\(\DENG\), runtime stabilizes and accelerates by 3.33 times as thread count increases from $2$ to $8$. For \GBBS-\(\DENGW\), performance deteriorates when the thread count exceeds $32$, but from $2$ to $32$ threads, it accelerates by $3.50$ times. Similarly, \GBBS-\(\Fraudar\) shows no significant change in runtime beyond $32$ threads, with an acceleration of $3.09$ times from $2$ to $32$ threads. \GBBS-\TDS{} and \GBBS-\KCDS{} fail to complete computations on the \texttt{la} dataset. In contrast, all five methods tested on \PEELB{} exhibit good scalability. As the number of threads increases from $2$ to $128$, \PEELB-\(\DENG\) accelerates by $9.91$ times, \PEELB-\(\DENGW\) by $12.69$ times, \PEELB-\(\Fraudar\) by $13.03$ times, and \PEELB-\TDS{} by $28.06$ times. Although \PEELB-\KCDS{} fails to produce results when the thread count is below $16$, it accelerates by $1.32$ times from $32$ to $128$ threads. These results demonstrate \PEELB{}'s strong scalability across varied threading levels.

\expstitle{CPU Utilization (Figure~\ref{fig:breakdown:main}).} Using $\DENG{}$ as an example, as the number of threads increases from 2 to 64, the stall rates of \GBBS\ grow more rapidly than those of \Dupin, demonstrating \Dupin’s superior scalability. In particular, at 64 threads, the busy rate of \Dupin\ is approximately 2.49 times (resp. 3.81 times) that of \GBBS\ and 1.04 times (resp. 1.67 times) that of \ALENEX\ on \texttt{soc} (resp. \texttt{sk}). Additional CPU utilization data across datasets are provided in Appendix~\ref{appendix:exp} of~\cite{dupin2024}.

\begin{figure}[tb]
        \includegraphics[width=0.8\linewidth]{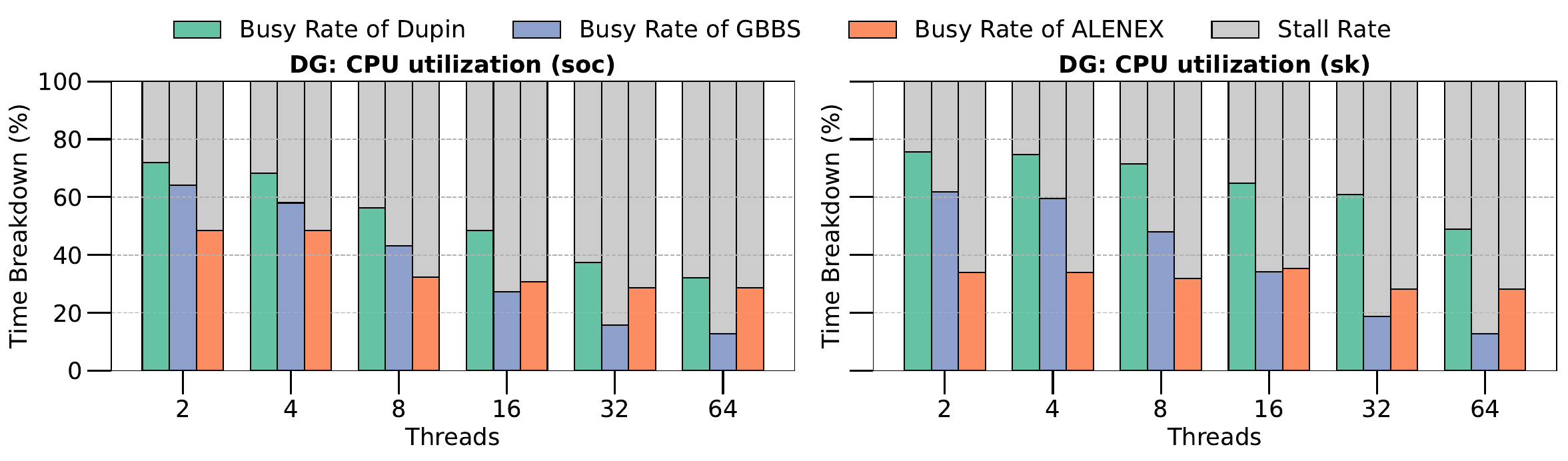}
  \caption{CPU Utilization \wrt $\DENG$.}\label{fig:breakdown:main}
\end{figure}


\subsection{Effectiveness of \Dupin{}}\label{sec:effectiveness}


\begin{figure*}
    \centering
\includegraphics[width=0.95\linewidth]{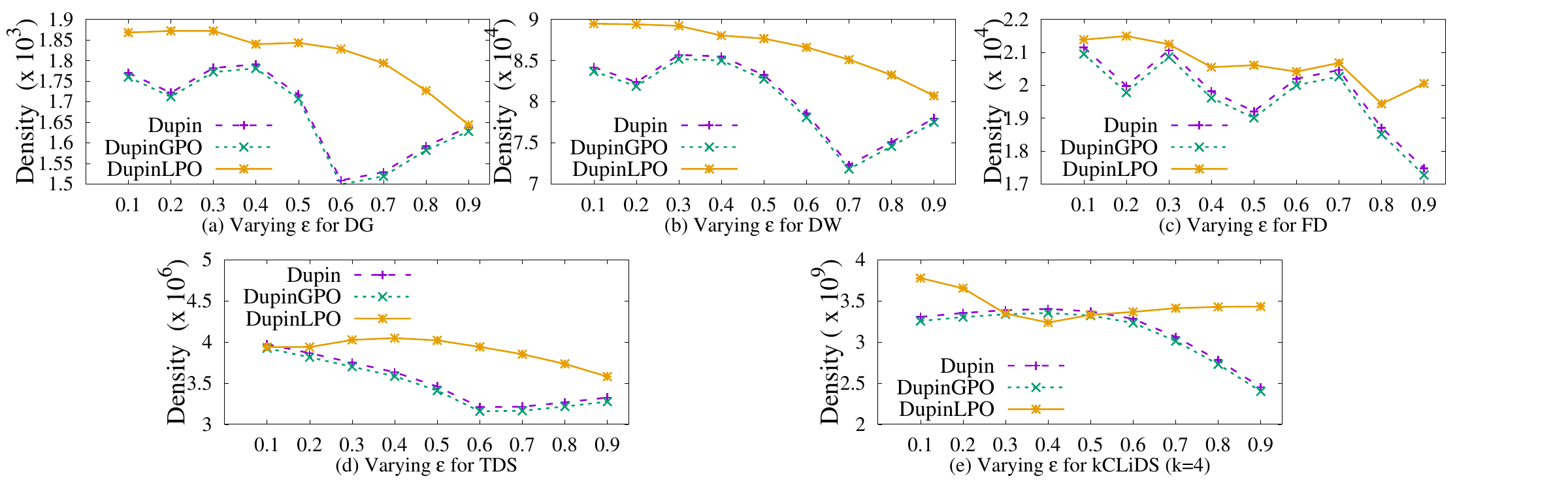}
    \caption{Effectiveness of \Dupin{} by varying $\epsilon$.}
    \label{fig:threads}
\end{figure*}

In addition to assessing performance speedup, we evaluate the density of the subgraphs detected by different frameworks. For parallel methods, we consider \GBBS{}, \PBBS{}, and \kClist{}, because they generally detect subgraphs with higher density than those detected by other methods across most datasets. For incremental methods, we use \Spade{} as a representative. Detailed density comparisons are presented in Table~\ref{tab:effectiveness-deng} and Table~\ref{tab:effectiveness-tds}.

\expstitle{\GBBS{} vs \Spade{} vs \PEELB{}.} Firstly, we compare the densities of the subgraphs identified by \GBBS{}-$\DENG$ (respectively \GBBS{}-$\DENGW$, and \GBBS{}-$\Fraudar$) to those detected by \PEELB{} in parallel execution. Our experimental results indicate that \PEELB{} maintains comparable subgraph densities with \GBBS{}, signifying that the increase in computational efficiency does not compromise the accuracy of the detected subgraphs. Specifically, the density of the subgraph detected by \GBBS{}-$\DENG$ (respectively \GBBS{}-$\DENGW$, and \GBBS{}-$\Fraudar$) is 7.08\% (resp. 6.48\% and 7.43\%) greater than that detected by \PEELB{}-$\DENG$ (respectively \PEELB{}-$\DENGW$, and \PEELB{}-$\Fraudar$) on average. \Spade{} shares similar densities with \GBBS{}. Due to \Spade{}'s tendency to accumulate errors over time, its density can become inflated. We will explain its limitations in practical applications in more detail in the case study.

\expstitle{\kClist{} vs \PBBS{} vs \PEELB{}.} a) For \KCDS{}, \PBBS{}-\KCDS{} could not complete on most of our test datasets, so we did not compare their densities. \kClist{}-\KCDS{} could not complete on the \texttt{sk} dataset. On the other datasets, \PEELB{}-\KCDS{}'s density was only 6.97\% smaller than \kClist{}-\KCDS{}. On some datasets, such as \texttt{bio} and \texttt{kron}, \PEELB{}-\KCDS{} even detected subgraphs with greater density. b) For \TDS{}, \PBBS{}-\TDS{} could detect subgraphs on smaller graphs but failed on billion-scale graphs such as \texttt{rv}, \texttt{sk}, and \texttt{la}. \PEELB{}-\TDS{}'s density was 2.03\% greater than \kClist{}-\TDS{}. Therefore, we conclude that \PEELB{} is comparable to existing parallel methods in terms of the density of the detected dense subgraphs.

\expstitle{Ablation Study of \Dupin{}.} This experiment delves into the effectiveness of our optimizations, \PEELLT{} and \PEELLTP{}. Our analysis primarily focuses on comparing the densities of dense subgraphs detected by both \PEELLT{} and \PEELLTP{} in various settings. \PEELLT{} achieves the same density as \Dupin{} across five density metrics because it trims the long tail, enhancing efficiency. For \PEELLTP{}-$\DENG$ (resp. \PEELLTP{}-$\DENGW$, \PEELLTP{}-$\Fraudar{}$, \PEELLTP{}-\TDS{}, and \PEELLT{}-\KCDS{}), we observe that the detected subgraphs exhibit densities 11.09\% (resp. 7.32\%, 8.00\%, 1.01\%, and 26.26\%) greater than those detected by \Dupin{}-$\DENG$, \Dupin{}-$\DENGW$, \Dupin{}-$\Fraudar{}$, \Dupin{}-\TDS{}, and \Dupin{}-\KCDS{}, respectively, demonstrating notable outcomes. This is because \PEELLTP{} trims the graph in each iteration, resulting in denser subgraphs compared to \Dupin{}.

\expstitle{Impact of the Approximation Ratio \( \epsilon \).} The approximation ratio \( \epsilon \) serves to control accuracy. In our experiment, we vary \( \epsilon \) from 0.1 to 1 to investigate its influence on the accuracy of \Dupin{}. We observe that as \( \epsilon \) increases from 0.1 to 1, the density of subgraphs detected by \Dupin{}, \PEELLT{}, and \PEELLTP{} generally decreases. For $\DENG$, the densities of \Dupin{}-$\DENG$, \PEELLT{}-$\DENG$, and \PEELLTP{}-$\DENG$ decrease by 22.71\%, 22.71\%, and 11.99\%, respectively. The trends for $\DENGW$ and $\DENG$ are similar, with the worst performance observed at \( \epsilon = 0.6 \) for both density metrics. The trends for $\Fraudar$, \TDS{}, and \KCDS{} are also similar. \PEELLTP{} detects denser subgraphs across most density metrics and is less affected by \( \epsilon \), due to its iterative trimming. For example, for $\Fraudar$, \PEELLTP{}'s density decreases by only 6.23\%, while \PEELLT{} and \Dupin{}'s densities decrease by 17.37\%.

\begin{figure}[tb]
        \includegraphics[width=0.85\linewidth]{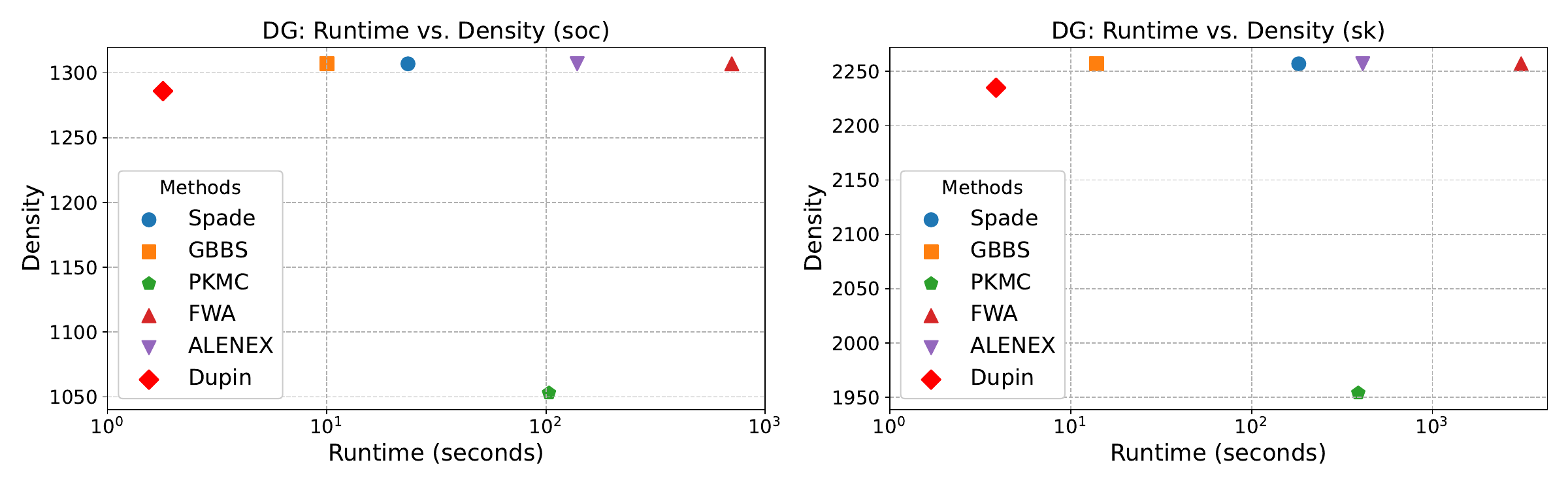}
  \caption{Comparison of runtime and density across different methods on two datasets, \wrt $\DENG$.}\label{fig:DG:runtime:density}
\end{figure}

\expstitle{Runtime-Density Trade-off Analysis.} Figure~\ref{fig:DG:runtime:density} shows runtime versus the density for six methods \Spade, \GBBS, \PKMC, \FWA, \ALENEX, and \Dupin{} over two representative datasets (\texttt{soc} and \texttt{sk}). \Spade{} and \FWA{} achieve high densities at the cost of longer runtimes, \PKMC\ is fast but yields lower densities, while \GBBS\ and \ALENEX\ strike a moderate balance. Notably, \Dupin\ consistently finishes the fastest with competitive densities, making it ideal for real-time or large-scale applications. Additional results for other datasets and density measures are provided in Appendix~\ref{appendix:exp} of~\cite{dupin2024}.

\subsection{Case study}

The comparison in real-world scenarios focusing on accuracy, latency, and prevention ratio.

\begin{figure}
\centering
   \includegraphics[width=0.8\linewidth]{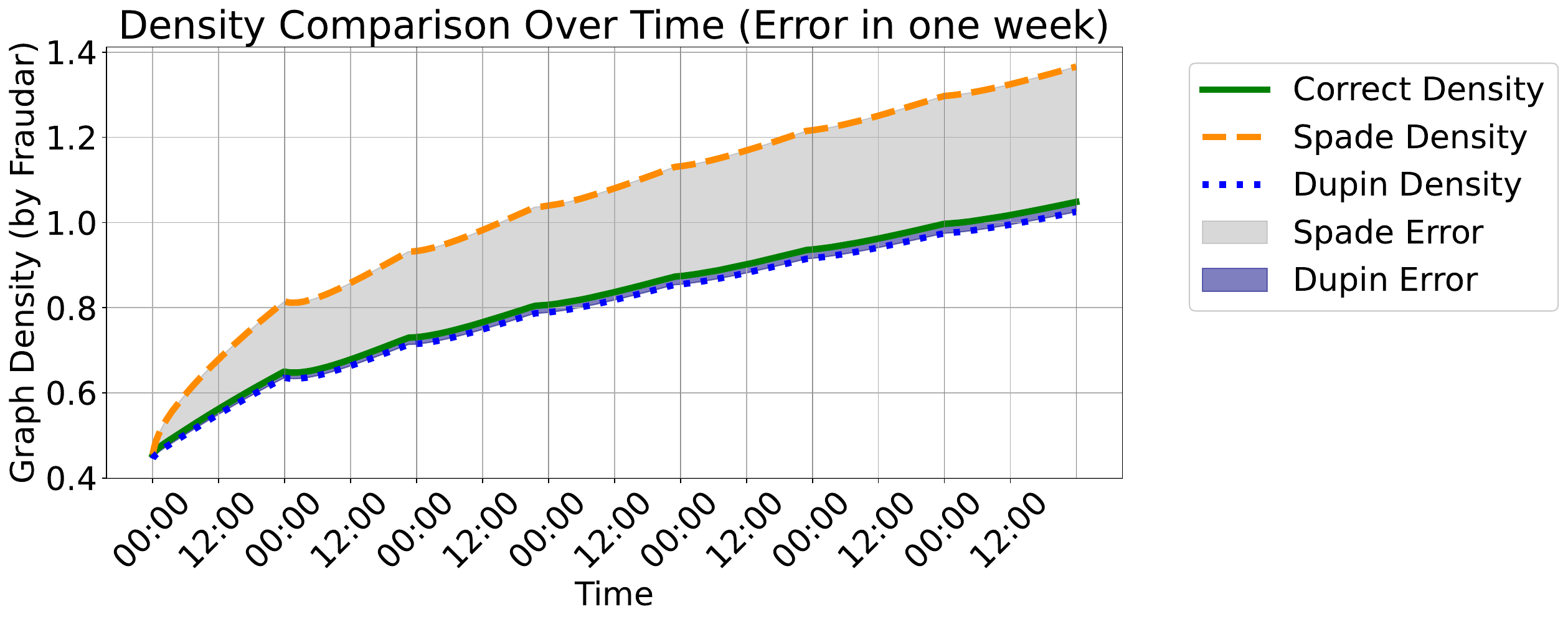}
        \caption{Weekly Graph Density Comparison in $\Grab$— Shows Correct Density (green), Spade Density (orange dashed), and Dupin Density (blue dotted) over one week. Gray and navy shaded areas indicate error margins for \Spade{} and \Dupin{}, highlighting their deviations from the ideal density.}\label{fig:density-error-accumulation}  
\end{figure}

\expstitle{Accuracy.} As discussed in Section~\ref{sec:intro}, although \Spade{} can quickly respond to new transactions, it tends to accumulate errors over time. Our analysis utilizes transaction data from $\Grab$, with the transaction network comprising \(|E|=2\) billion edges. Under the incremental computation framework of \Spade{}, the assumption is made that the weights of existing edges remain constant, thereby neglecting the impact of newly inserted edges on these weights, leading to error accumulation. For instance, using \(\Fraudar\) density metric, we observe that errors accumulate up to 24.4\% within a day and increase to 30.3\% over a week, as illustrated in Figure~\ref{fig:density-error-accumulation}. Although \Dupin{} sacrifices some accuracy to achieve parallel computational efficiency, our experiments show that \Dupin{}'s errors remain relatively stable between 3\% and 5\%. Hence, the detection precision of \Spade{} declines from 87.9\% to 68.3\%, while \Dupin{}'s precision remains above 86\%.


\begin{table}[tb]
\centering
\resizebox{0.85\linewidth}{!}{
\begin{tabular}{|l|cc|cc|cc|cc|}
\hline
\textbf{} & \multicolumn{2}{c|}{$\DENG$} & \multicolumn{2}{c|}{$\DENGW$} & \multicolumn{2}{c|}{$\Fraudar$} & \multicolumn{2}{c|}{\TDS{}} \\
\hline
\textbf{} & \( \mathcal{L} \) & \( \mathcal{R} \) & \( \mathcal{L} \)& \( \mathcal{R} \) & \( \mathcal{L} \) & \( \mathcal{R} \) & \( \mathcal{L} \) & \( \mathcal{R} \) \\
\hline
\Dupin{} & 3.10 & 78\% & 3.54 & 86\% & 3.59 & 94.5\% & 2145.00 & 32\% \\
\hline
\Spade{} & 165.20 & 58\% & 235.63 & 63\% & 197.61 & 45\% & \TLE & \TLE \\
\hline
\GBBS  & 927.88 & 12\% & \TLE & \TLE & 6014.00 & 3\% & - & - \\
\hline
\end{tabular}
}
\caption{Latency vs. Prevention Ratio by Method: A method's performance in latency (\( \mathcal{L} \): seconds) against its prevention ratio (\( \mathcal{R} \)), detailing the effectiveness of various computational strategies in real-time environments.}
\label{fig:latency2prevention}
\end{table}

\expstitle{Latency and Prevention Ratio.} If a transaction is identified as fraudulent, subsequent related transactions are blocked to mitigate potential losses. We define the ratio of successfully prevented suspicious transactions to the total number of suspicious transactions as \( \mathcal{R} \). As shown in Figure~\ref{fig:latency2prevention}, the prevention ratio \( \mathcal{R} \) continues to decrease as latency increases on $\Grab$'s datasets. Using the default density metric, \(\Fraudar\), in $\Grab$, our results show that \Dupin-\(\Fraudar\) prevents 94.5\% of fraudulent activities, \GBBS{}-\(\Fraudar\) prevents 3.0\%, and \Spade-\(\Fraudar\) prevents 45.3\%. The reason for the lower performance of \GBBS{}-\(\Fraudar\) is its low parallelism in the peeling process, which results in long peeling times and thus an average latency of 6{,}014 seconds, delaying the prevention of many fraudulent activities. Although \Spade-\(\Fraudar\) performs well in reordering techniques on smaller datasets, in industrial scenarios with billion-scale graphs, latency significantly worsens due to the extensive reordering range required in the peeling sequences as normal users transition to fraudsters, increasing the cost of incremental computations. Additional tests were conducted with other density metrics; for \(\DENG\), \Dupin-\(\DENG\) prevents 78.8\%, while other methods prevent varying amounts. Similarly, for \(\DENGW\), \Dupin-\(\DENGW\) prevents 86.1\%, while other methods demonstrate different prevention capabilities.





\section{Conclusion and future works}\label{sec:conclusion}

In this paper, we presented \Dupin, a novel parallel fraud detection framework tailored for massive graphs. Our experiments demonstrated that \Dupin's advanced peeling algorithms and long-tail pruning techniques, \PEELLT{} and \PEELLTP{}, significantly enhance the detection efficiency of dense subgraphs without sacrificing accuracy. Comparative studies with \GBBS{}, \PBBS{}, \kClist{}, \Spade{}, and \PEELB{} highlighted \Dupin's superior performance in terms of computational efficiency and accuracy. Through case studies, we validated that \Dupin{} achieves lower latency and higher detection accuracy compared to existing fraud detection systems on billion-scale graphs.

\stitle{Acknowledgement.} This research is supported by the National Research Foundation, Singapore under its AI Singapore Programme (AISG Award No: AISG2-TC-2021-002), the Ministry of Education, Singapore, under its Academic Research Fund Tier 2 (Award No: MOE2019-T2-2-065), and the Singapore Ministry of Education (MOE) Academic Research Fund (AcRF) Tier 1 grant, 23-SIS-SMU-063. Yuchen Li is partially supported by the Lee Kong Chian fellowship.


\newpage
\clearpage

\bibliographystyle{abbrv}
\bibliography{ref}

\newpage
\clearpage

\appendix

\section*{Appendix}

\section{More Experiments}\label{appendix:exp}

\begin{table*}[ht]
\centering
\captionsetup{width=1\linewidth} 
\caption{Comparison on Different Hardware Platforms. A dash (``-'') denotes that the framework does not support the density metrics. \textbf{TLE} = Time Limit Exceeded.}
\label{tab:performance-comparison}
\resizebox{\textwidth}{!}{%
\begin{tabular}{l*{10}{c}}
\toprule
\multirow{2}{*}{\textbf{System}} & \multicolumn{2}{c}{$\DENG$} & \multicolumn{2}{c}{$\DENGW$} & \multicolumn{2}{c}{$\Fraudar$} & \multicolumn{2}{c}{\TDS{}} & \multicolumn{2}{c}{\KCDS{}} \\
\cmidrule(lr){2-3} \cmidrule(lr){4-5} \cmidrule(lr){6-7} \cmidrule(lr){8-9} \cmidrule(lr){10-11}
& X5650 & EPYC 7742 & X5650 & EPYC 7742 & X5650 & EPYC 7742 & X5650 & EPYC 7742 & X5650 & EPYC 7742 \\
\midrule
Spade   & 23.46 & 20.45 & 28.10 & 26.04 & 30.67 & 28.41 & -      & -      & -      & -      \\
\FWA    & 704   & 559   & 1,241 & 839   & 1,381 & 914   & -      & -      & -      & -      \\
\GBBS   & 10.01 & 8.41  & 23.28 & 15.73 & 35.43 & 24.08 & -      & -      & -      & -      \\
\PBBS   & -     & -     & -     & -     & -     & -     & 3,524.11 & 1,970.90 & \TLE & \TLE \\
Dupin   & 1.79  & 0.80  & 2.26  & 1.09  & 2.38  & 1.18  & 32.59   & 18.09   & 283.70 & 148.83 \\
\bottomrule
\end{tabular}
}
\end{table*}

\stitle{Hardware Evaluation.} In our experiments, we used the Xeon X5650 processor to ensure a fair comparison with the setup reported in Spade’s original work. We further conducted tests on a more modern EPYC 7742 processor. Our results show that Dupin achieves a 1.91–2.25x speedup on EPYC 7742 compared to its performance on the Xeon X5650, whereas Spade only attains a 1.08–1.15x speedup. Notably, Dupin’s performance scales nearly linearly with the number of cores, effectively leveraging modern hardware’s memory bandwidth and parallelism. In contrast, Spade’s largely sequential design limits its scalability. Similarly, both \GBBS\ and \PBBS\ exhibit lower speedup improvements on the EPYC 7742 than Dupin, with respective speedups of 1.19x and 1.79x, further highlighting the superior scalability of Dupin.

\begin{figure}[tb]
        \includegraphics[width=\linewidth]{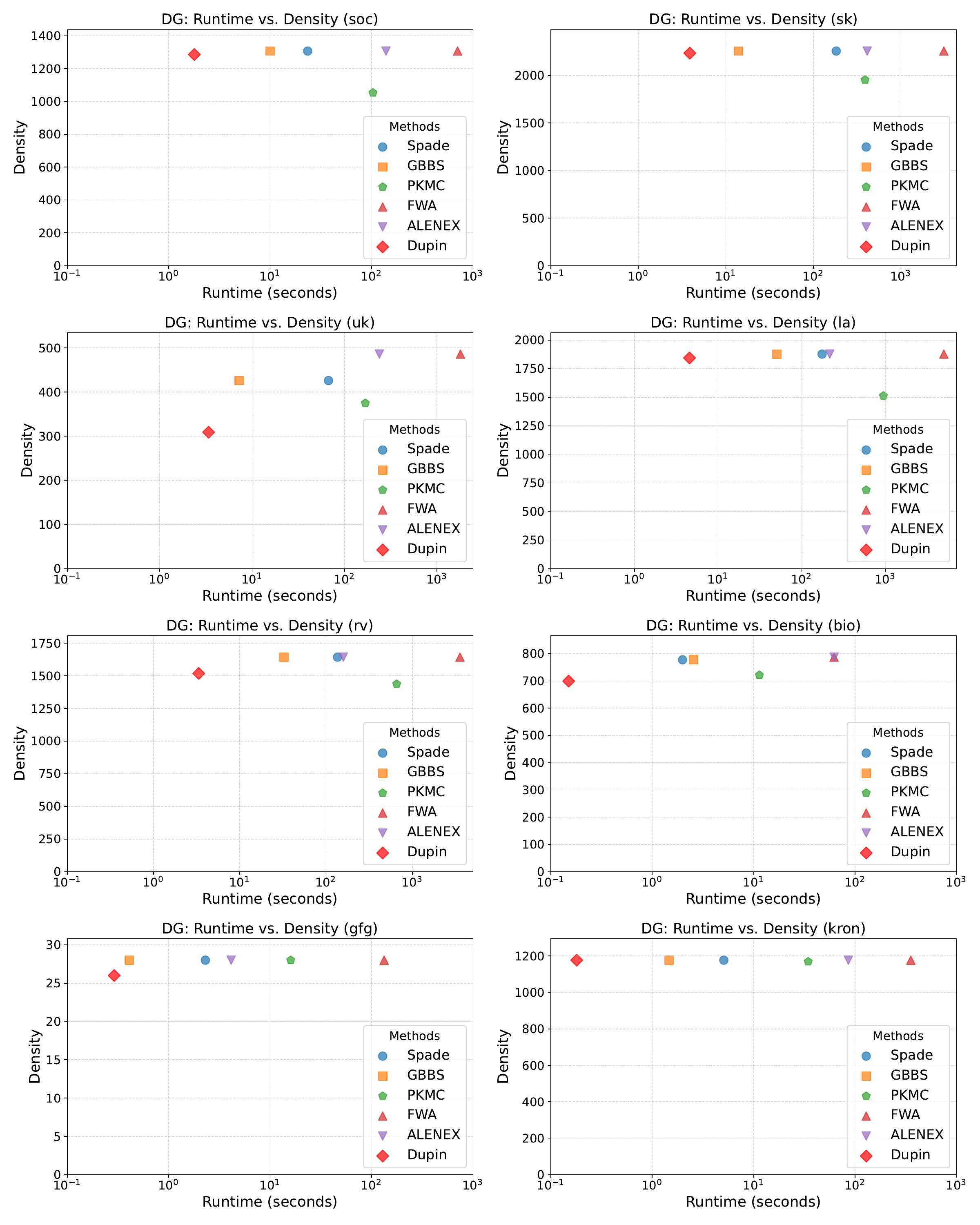}
  \caption{Comparison of runtime and density across different methods, \wrt $\DENG$}\label{fig:DG:runtime:density:appendix}
\end{figure}

\begin{figure}[tb]
        \includegraphics[width=\linewidth]{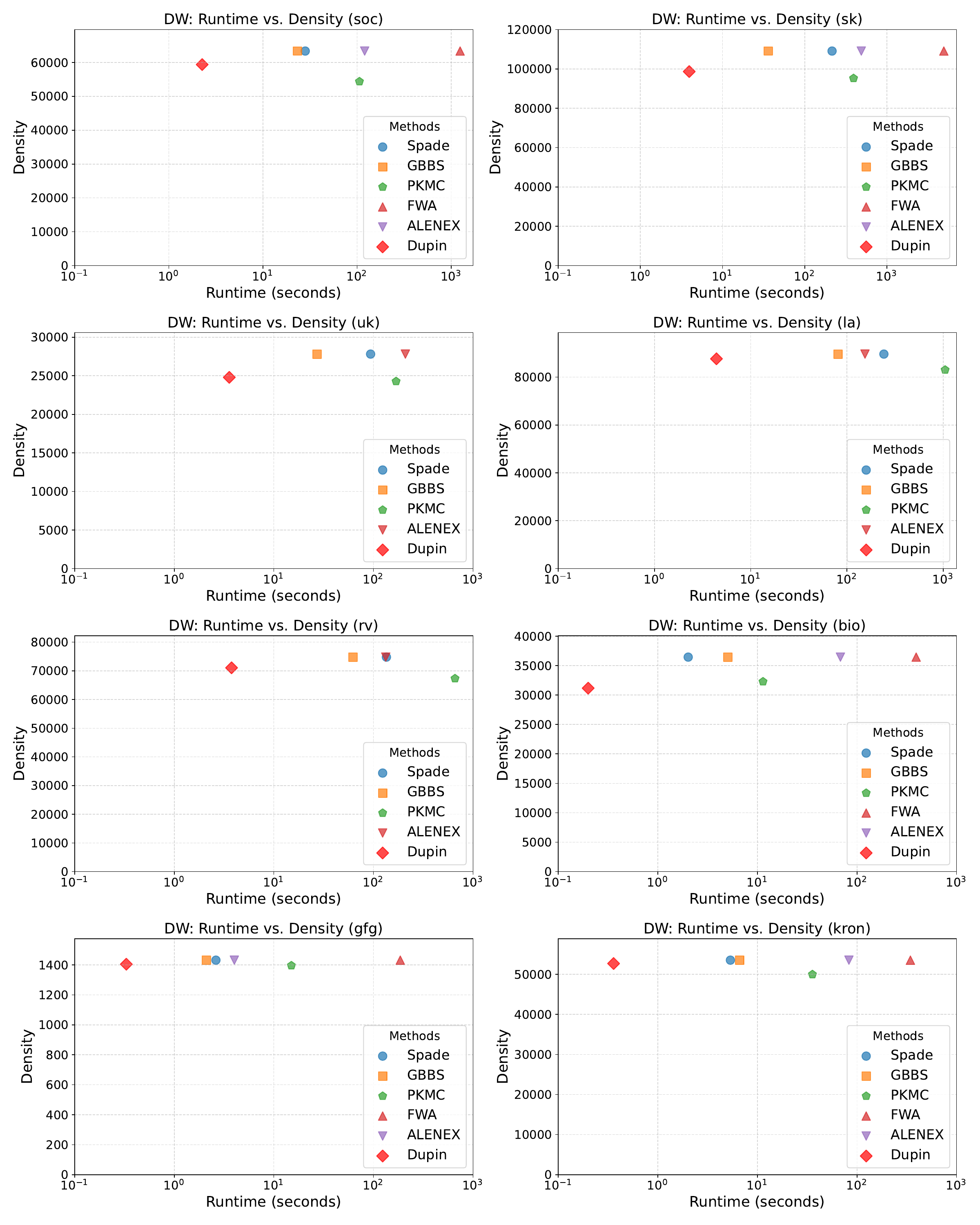}
  \caption{Comparison of runtime and density across different methods, \wrt $\DENGW$}\label{fig:DW:runtime:density:appendix}
\end{figure}

\begin{figure}[tb]
        \includegraphics[width=\linewidth]{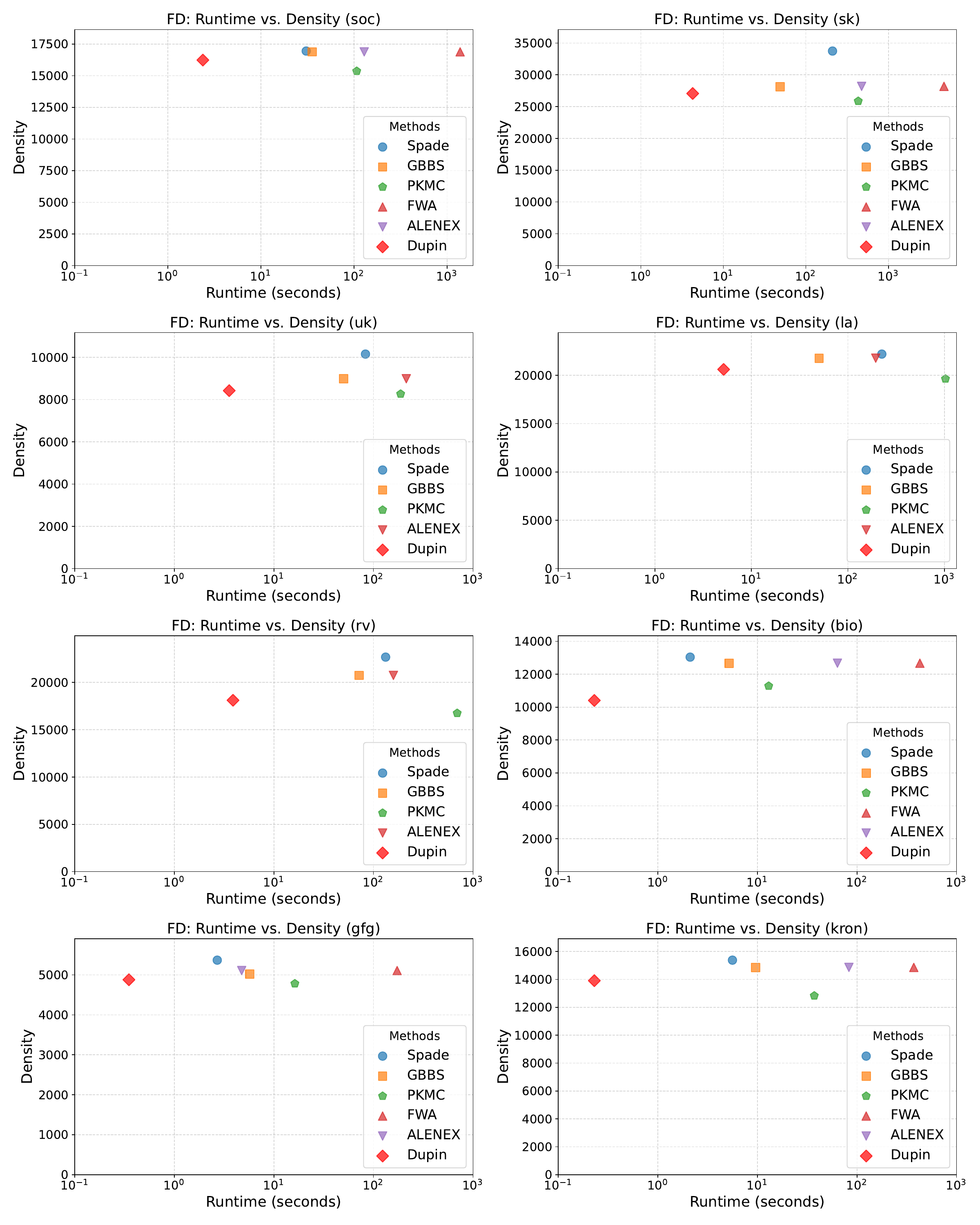}
  \caption{Comparison of runtime and density across different methods, \wrt $\Fraudar$}\label{fig:FD:runtime:density:appendix}
\end{figure}

\stitle{Additional CPU Utilization Experiments.} 
To comprehensively evaluate \Dupin{}'s performance, we conduct additional experiments on CPU utilization across various density metrics. These experiments demonstrate how \Dupin{} efficiently manages CPU resources across different density evaluation scenarios, highlighting its scalability and effectiveness in processing large-scale graphs. Specifically, we analyze:
\begin{enumerate}
    \item \textbf{CPU Utilization \wrt $\DENG$} (Figure~\ref{fig:breakdown:appendix})
    \item \textbf{CPU Utilization \wrt $\DENGW$} (Figure~\ref{fig:breakdown:appendix:dw})
    \item \textbf{CPU Utilization \wrt $\Fraudar$} (Figure~\ref{fig:breakdown:appendix:fd})
\end{enumerate}

\begin{figure}[h]
        \includegraphics[width=\linewidth]{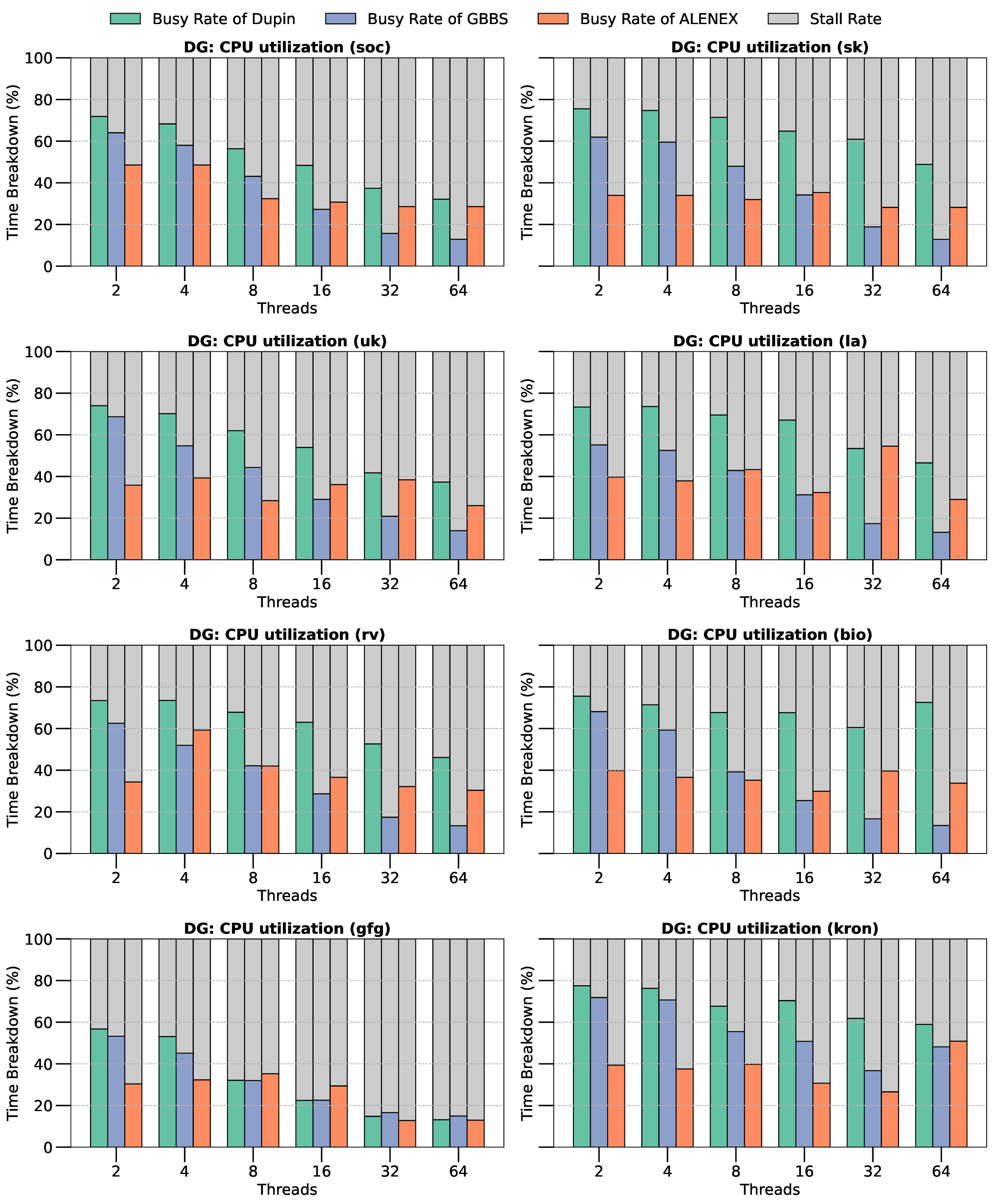}
  \caption{CPU Utilization \wrt $\DENG$}\label{fig:breakdown:appendix}
\end{figure}

\begin{figure}[h]
        \includegraphics[width=\linewidth]{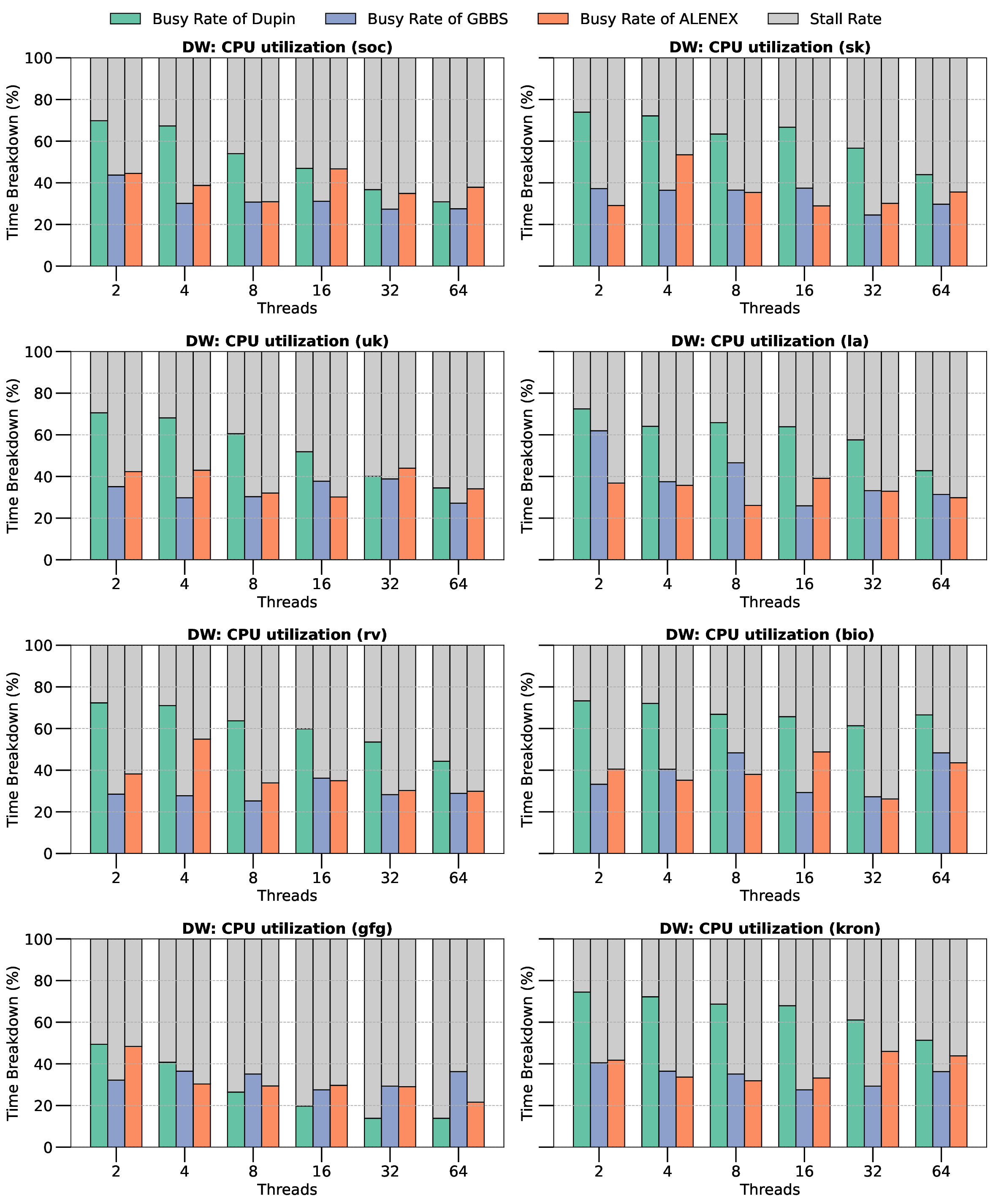}
  \caption{CPU Utilization \wrt $\DENGW$}\label{fig:breakdown:appendix:dw}
\end{figure}

\begin{figure}[h]
        \includegraphics[width=\linewidth]{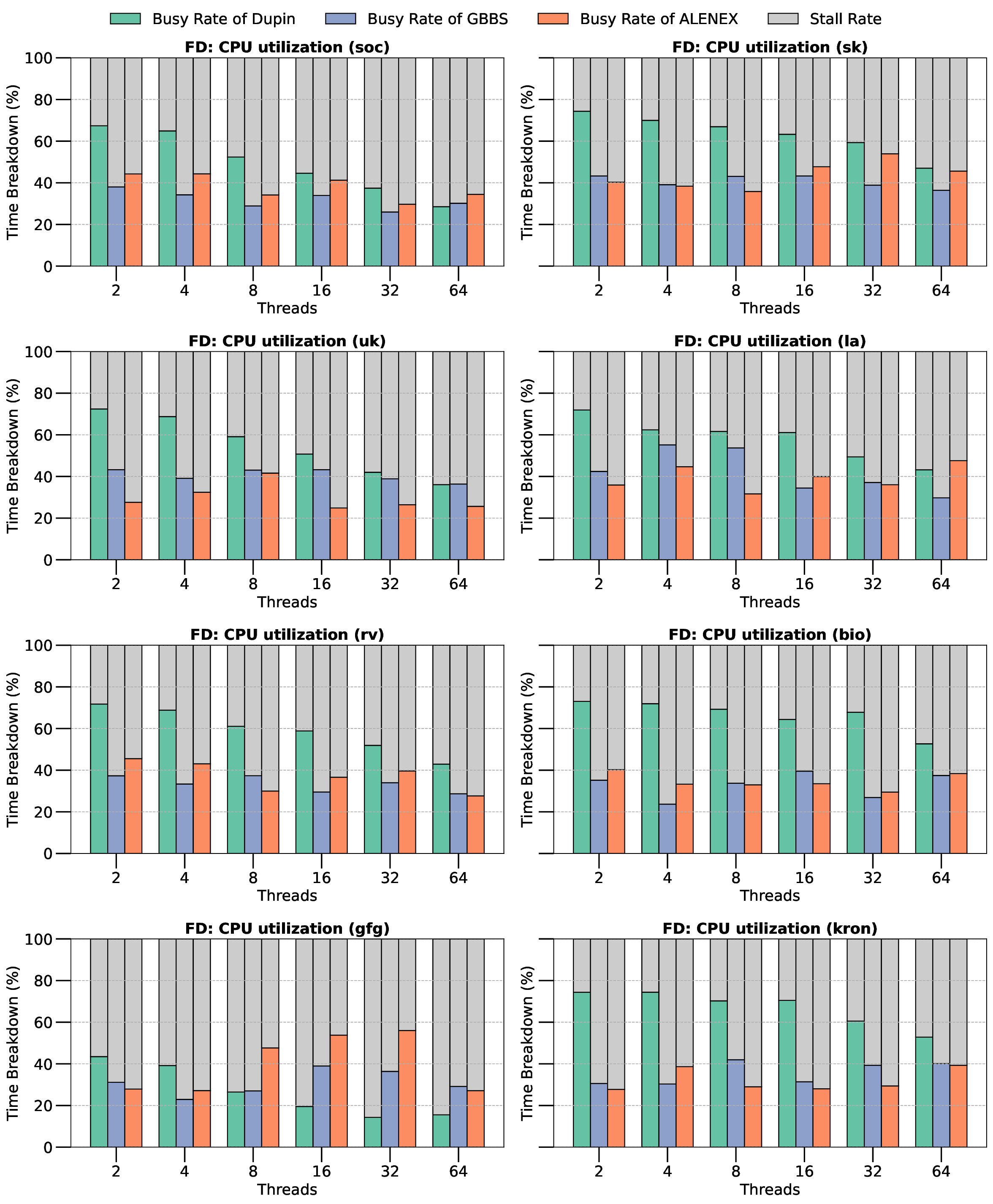}
  \caption{CPU Utilization \wrt $\Fraudar$}\label{fig:breakdown:appendix:fd}
\end{figure}

\stitle{Comparison of Runtime and Density Across Different Methods.}
Figures~\ref{fig:DG:runtime:density:appendix},~\ref{fig:DW:runtime:density:appendix}, and~\ref{fig:FD:runtime:density:appendix} present a comprehensive comparison of runtime and density for six methods—Spade, GBBS, PKMC, FWA, ALENEX, and Dupin—across two representative datasets with respect to the density metrics $\DENG{}$, $\DENGW{}$, and $\Fraudar{}$, respectively. Across all metrics, \Dupin{} consistently achieves the fastest runtimes while maintaining competitive density scores. For instance, in scenarios evaluated with $\DENG{}$, $\DENGW{}$, and $\Fraudar{}$, \Dupin{} outperforms state-of-the-art algorithms by significantly reducing processing times without substantial sacrifices in density. This balanced performance makes \Dupin{} particularly well-suited for large-scale and real-time fraud detection applications, where both speed and accuracy are critical. These results demonstrate \Dupin{}'s superior scalability and effectiveness in handling massive graphs, underscoring its advantage over existing methods in diverse density evaluation contexts.

\clearpage

\section{\DDS{} implementation in \Dupin}\label{appendix:implementation}


    

\begin{lstlisting}[language=C++, style=apistyle, caption=Implementation of $\DENG$ in \Dupin{}, captionpos=b, label=lst:deng]
double vsusp(const Vertex& v, const Graph& g) {
  // For unweighted graphs, we ignore vertex weights.
  return 0.0;
}
double esusp(const Edge& e, const Graph& g) {
  // For unweighted graphs, each edge has weight 1.0.
  return 1.0;
}
int main() {
  Dupin dupin;
  dupin.VSusp(vsusp); // Plug in vertex suspiciousness
  dupin.ESusp(esusp); // Plug in edge suspiciousness
  dupin.setEpsilon(0.1); // Set approximation parameter
  dupin.LoadGraph("graph_unweighted_path");
  vector<Vertex> result = dupin.ParDetect(); // Detect densest subgraph
  return 0;
}
\end{lstlisting}

\begin{lstlisting}[language=C++, style=apistyle, caption=Implementation of $\DENGW$ in \Dupin{}, captionpos=b, label=lst:dengw]
double vsusp(const Vertex& v, const Graph& g) {
  return g.weight[v]; // Side information on the vertex
}
double esusp(const Edge& e, const Graph& g) {
  return g.weight[e];// Side information on the edge
}
int main() {
  Dupin dupin;
  dupin.VSusp(vsusp);  // Plug in vertex suspiciousness
  dupin.ESusp(esusp);  // Plug in edge suspiciousness
  dupin.setEpsilon(0.1); // Set approximation parameter
  dupin.LoadGraph("graph_weighted_path");
  vector<Vertex> result = dupin.ParDetect();
  return 0;
}
\end{lstlisting}

\begin{lstlisting}[language=C++, style=apistyle, 
  caption={Implementation of \TDS{} and \KCDS{} in \Dupin{}}, 
  captionpos=b, label=lst:kcds_tds]
double vsusp(const Vertex& v, const Graph& g) {
  static const int k = 3; // k=3 => TDS; k>3 => kCLiDS
  // returns how many k-cliques vertex 'v' forms with active vertices.
  return Dupin::countKCliques(v, g, k) / k;
}
double esusp(const Edge& e, const Graph& g) {
  return 0.0;
}
int main() {
  Dupin dupin;
  // Assign suspiciousness functions for TDS or KCDS (depending on 'k')
  dupin.VSusp(vsusp);
  dupin.ESusp(esusp);
  dupin.setEpsilon(0.1);
  dupin.setK(3);
  dupin.LoadGraph("graph_kclique_path"); 
  std::vector<Vertex> result = dupin.ParDetect();
  return 0;
}
\end{lstlisting}

\section{Proofs}\label{subsec:appendix:proof}
\begin{manualtheorem}{\ref{thm:2ppr}}

\begin{figure}[b]
\includegraphics[width=0.8\linewidth]{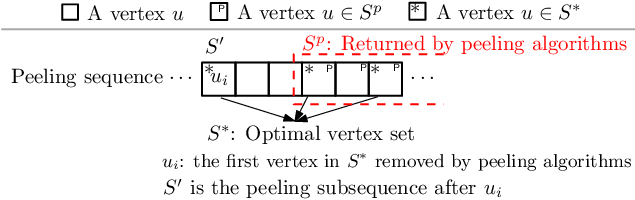}
    \caption{Illustration of key notations in peeling algorithms}\label{fig:proof}    
\end{figure}

Let $S^p$ be the vertex set returned by Algorithm~\ref{algo:peeling} and $S^*$ be any optimal solution. For each of the density metrics $\DENG$, $\DENGW$, and $\Fraudar$, the following approximation guarantee holds:
\[
g(S^p) \;\ge\; \tfrac{1}{2}\,g(S^*).
\]
\end{manualtheorem}

\begin{proof}\label{eq:proof}
Consider the first vertex \(u_i\in S^*\) that the peeling algorithm removes. Let \(S'\subseteq S^*\) be the subset immediately before \(u_i\) is removed. By definition, \(u_i\) is chosen because removing it has the smallest impact on the function \(f(\cdot)\). Hence,
\begin{equation}
    g(S') \;=\; \frac{f(S')}{|S'|} 
    \;\;\ge\;\; \frac{\sum_{u_j\in S'} w_{u_j}(S')}{2\,|S'|}
    \;\;\ge\;\; \frac{w_{u_i}(S')}{2}.
\end{equation}
Since Algorithm~\ref{algo:peeling} ensures \(g(S^p)\ge g(S')\), we have
\[
g(S^p)
\;\;\ge\;\;
g(S')
\;\;=\;\;
\frac{f(S')}{|S'|}
\;\;\ge\;\;
\frac{w_{u_i}(S')}{2}
\;\;\ge\;\;
\frac{w_{u_i}(S^*)}{2}
\;\;\ge\;\;
\frac{g(S^*)}{2}.
\]
The last inequality uses the fact that \(w_{u_i}(S') \ge w_{u_i}(S^*)\) directly implies \(g(S') \ge g(S^*)\) up to a factor of \(2\) under the $\DENG$, $\DENGW$, or $\Fraudar$ metrics.
\end{proof}

\begin{manualtheorem}{\ref{thm:3ppr_cli}}
Let $S^p$ be the vertex set returned by Algorithm~\ref{algo:peeling} and $S^*$ be any optimal solution. For the density metrics \TDS{} and \KCDS{}, where each edge is counted $k$ times (and $k=3$ for \TDS{}), the following guarantee holds:
\[
g(S^p) \;\ge\; \frac{g(S^*)}{k}.
\]
\end{manualtheorem}

\begin{proof}
Consider the first vertex \(u_i \in S^*\) that the peeling algorithm removes. Let \(S'\subseteq S^*\) be the subset immediately before \(u_i\) is removed. Because \(u_i\) minimally decreases \(f(S')\) and each edge is counted \(k\) times in \TDS{} or \KCDS{}, it follows that
\begin{equation}\label{eq:proof}
    g(S') \;=\; \frac{f(S')}{|S'|}
    \;\;\ge\;\;
    \frac{\sum_{u_j \in S'} w_{u_j}(S')}{k\,|S'|}
    \;\;\ge\;\;
    \frac{w_{u_i}(S')}{k}.
\end{equation}
By construction, \(g(S^p)\!\geq\!g(S')\). Therefore,
\[
g(S^p)
\;\;\ge\;\;
g(S')
\;=\;
\frac{f(S')}{|S'|}
\;\;\ge\;\;
\frac{w_{u_i}(S')}{k}
\;\;\ge\;\;
\frac{w_{u_i}(S^*)}{k}
\;\;\ge\;\;
\frac{g(S^*)}{k},
\]
which completes the proof.
\end{proof}

\appendix








\end{document}